\documentclass[10pt]{article}
\expandafter\let\csname equation*\endcsname\relax
\expandafter\let\csname endequation*\endcsname\relax
\usepackage{bm}
\usepackage{enumerate}
\usepackage[top=1in, bottom=1in, left=1in, right=1in]{geometry}
\usepackage{mathrsfs}
\usepackage{amsfonts}\usepackage{multirow}
\usepackage{amssymb,dsfont}
\usepackage{caption,comment}
\usepackage{amsmath}
\usepackage{float}
\usepackage[stretch=8,shrink=10]{microtype}
\usepackage{amssymb,amsthm}
\usepackage[toc,page]{appendix}
\usepackage{graphicx,afterpage}
\usepackage{epstopdf}
\usepackage[normalem]{ulem}
\usepackage{authblk}
\usepackage[pdftex,colorlinks=true]{hyperref}
\numberwithin{equation}{section}
\numberwithin{figure}{section}

\newcommand\tabcaption{\def\@captype{table}\caption}

\newtheorem{thm}{Theorem}[section]
\newtheorem{cor}[thm]{Corollary}
\newtheorem{lem}[thm]{Lemma}

\theoremstyle{remark}
\newtheorem{rem}[thm]{Remark}
\newtheorem{example}[thm]{Example}

\newcommand{\ran}{\mathrm{ran}}
\newcommand{\im}{\mathrm{im}}
\newcommand{\tr}{\mathrm{trace}}
\newcommand{\ke}{\mathrm{ke}}

\newcommand{\diag}{\mathrm{diag}}

\newcommand{\cT}{\mathcal T}

\newcommand{\argmin}{\mathrm{argmin}}

\newcommand{\argmax}{\mathrm{argmax}}

\renewcommand{\d}{\,\mathrm{d}}
\newcommand{\R}{\mathbb R}

\newcommand{\eps}{\varepsilon}
\newcommand{\E}{\mathbb E}

\title{Mathematical description of continuous time and space replicator-mutator equations for quadratic fitness landscapes}
\date{December 2024}

\author[1]{Sahani Pathiraja}
\author[2]{Philipp Wacker}

\affil[1]{School of Mathematics \& Statistics, UNSW Sydney, Australia, \url{https://orcid.org/0000-0002-0114-3164}}
\affil[2]{School of Mathematics and Statistics, University of Canterbury, New Zealand, \url{https://orcid.org/0000-0001-8718-4313}}

\begin{document}

	\maketitle

 \begin{abstract}
    The replicator-mutator equation is a model for populations of individuals carrying different traits, with a fitness function mediating their ability to replicate, and a stochastic model for mutation. We derive analytical solutions for the replicator-mutator equation in continuous time and for continuous traits for a quadratic fitness function. Using these results we can explain and quantify (without the need for numerical in-silico simulations) a series of evolutionary phenomena, in particular the flying kite effect, survival of the flattest, and the ability of a population to sustain itself while tracking an optimal feature which may be fixed, moving with bounded velocity in trait space, oscillating, or randomly fluctuating. 

 \end{abstract}
	\tableofcontents

\section{Introduction}
The dynamics of replication (``Who gets to have offspring?'') and mutation (``How different is offspring from their parents?'') are central to evolutionary theory, and to mathematical models describing evolution. The replicator-mutator equation is a foundational model showcasing much of the complexity we can see in nature. Its replication component describes how a certain fitness function mediates selection on traits, and its mutation component mimics imperfect reproduction. This leads to a rich interplay of competing effects: Replication-based adaptation to a (possibly non-constant) optimal trait, and mutation-based diffusion (possibly) away from this optimum. 

The literature in (mathematical) evolution is rich with the description of interesting, nontrivial, and often counterintuitive behaviour of evolutionary systems. An example for this would be the so-called \textit{Flying Kite Effect} \cite{jones2004evolution,kopp2018phenotypic}: A population of individuals/traits often does not adapt to an optimal trait along the shortest Euclidean connection in trait space, but curves away from the shortest connecting path, preconditioned by the population's (empirical) covariance. Or, the \textit{Survival of the Flattest} phenomenon \cite{wilke2001evolution}: Effects like this are often observed from microscopic models, i.e., models where individuals are simulated individually, with methods like adaptive walks \cite{Matuszewski2014} or cellular automata \cite{sardanyes2008simple}. The downside of these microscopic models is that they are not entirely trivial to set up, computationally intensive, and (because they are simulations) can only give examples for specific parameter settings rather than a general theory.

In this manuscript we study and illustrate these interesting phenomena with a relatively simple replicator-mutator equation (with Gaussian mutation and a quadratic fitness function). For this setup, an exact analytical solution (for arbitrary parameters) can often be derived, although this requires mathematical techniques from control theory, where the replicator-mutator equation is studied in the context of Riccati equations. 

The structure of this article is as follows: Section \ref{sec:evomodels} considers a few variants of the replicator-mutator equation relevant in our setting, fixes notation, and gives an (incomplete) literature review of the field. Section \ref{sec:explicit} is the mathematical heart of this article, with explicit solutions for the replicator-mutator equation being derived in various settings, and (to the best of our knowledge) more generally than previously done. Another novelty of this manuscript is that we keep track of total population size: Standard treatments of the replicator-mutator equation model a probability distribution of traits which necessarily keeps total mass constant to $1$, but we consciously allow the total mass of the population to change in time (as essentially driven by the mean fitness of the population), and we can derive exact solutions for this total population mass in many cases, too. Section \ref{sec:examples} employs these results to illustrate the \textit{Flying Kite effect}, \textit{Survival of the Flattest}, adaptation to a fixed optimum, as well as adaptation to optima moving with bounded speed, oscillating motion, or with random fluctuations. We work with analytic solutions, so in many cases we can give exact statements about what parameter settings lead to certain behaviour, like ``What is the maximum mutation strength such that a population does not go extinct while trying to adapt?''. 

We hope this manuscript can help illustrate a series of interesting phenomena in mathematical models of evolution with more elementary arguments using closed-form solutions instead of numerical simulations. Put succinctly, our contributions are as follows:
\begin{enumerate}
    \item Derivation of exact solutions for the replicator-mutator equation in continuous time and for a continuous trait space, with a multivariate Gaussian initial trait distribution, and subject to a general quadratic fitness function, driven by multivariate Gaussian mutation. $\to$ see Section \ref{sec:explicit}. To the best of our knowledge, that has not yet appeared in the literature. 
    \item Quantitative description of total population growth or decay over time from the unnormalised quadratic replicator-mutator equation.
    \item We point out that the replicator-mutator equation is very closely related to Tikhonov-regularised optimisation, as well as to (mis-specified) filtering via the Kalman-Bucy filter. $\to$ see Section \ref{sec:explicit}
    \item Mathematical verification of evolutionary phenomena (Flying Kite effect; Survival of the flattest; Adaptation to fixed, moving, periodic, and fluctuating optima) previously described mostly via numerical experiments. This allows us to understand in rigorous mathematical terms how these phenomena arise, and we can quantify objects like the concrete fixed lag of a population, or the asymptotic growth rate of a population in explicit terms rather than relying on numerical simulations. The benefit of having exact mathematical expressions is that this allows us to quantify the direct effect of all parameters involved. For example, if mutation strength increases by a certain factor, what does this mean for the exponential growth rate of the population? $\to$ see Section \ref{sec:examples}
\end{enumerate}

We consider a \textit{trait space }$X$, which (with exception of section \ref{sec:evomodels}) is always assumed to be Euclidean space $X = \R^d$, with $d$ being the dimension of the trait space. The \textit{feature map} $A: X\to Y$ is understood as a linear map from this trait space into a \textit{feature space} $Y = \R^k$.  The role of this map is to translate a trait (or genotype) into a feature (or phenotype) which is more directly coupled to fitness. In the following we will often consider an \textit{optimal feature} $y(t)$ mediating fitness via (essentially) squared deviation from this feature, i.e., $-\|Ax - y(t)\|_\Gamma^2$ being a measure of fitness, punishing or rewarding a trait $x$ depending on its $\Gamma$-preconditioned quadratic distance from $y(t)$. For further details on mathematical notation and terms used in this manuscript, see the appendix.

\section{Some mathematical models of evolution based on replication and mutation} \label{sec:evomodels}

The replicator-mutator equation is a broad class of dynamical systems modelling the response of a distribution of traits to evolutionary adaptation to an external fitness landscape. 
Early research on the replicator-mutator equation started with \cite{crowkimura,kimura_stochastic_1965}, which is the reason for its alternative name, Crow-Kimura equation, and consequent work in \cite{akin1979geometry,schuster1983replicator,stadler1994immune} and many others, in the beginning focussing mostly on the discrete trait case, but then generalised to the continuous setting in works like \cite{cressman2006stability,gil_mathematical_2017,vlasic_markovian_2020}. Dynamical systems and game theory communities have studied and classified its behaviour, see for example \cite{bomze1990dynamical,oechssler2001evolutionary,cressman2014replicator,garay2018ess} and works like \cite{baez2016relative,baez2021fundamental,Harper2009,harper2009replicator,fujiwara1995gradient,Chalub2021} developed interesting connections to (information) geometry, in particular Fisher information and the Shashahani metric. Models from the mathematical evolution, biology, and ecology literature seem to have a preference for discrete time, as in \cite{kopp_rapid_2014,kopp2018phenotypic,Matuszewski2014}, possibly because it mimics generational reproduction more naturally. 

Below we describe a few variants of the replicator(-mutator) equation and how they relate to one another. These variants are all 16 combinations of the following 4 design choices:
\begin{itemize}
    \item Does time evolve in discrete time steps or continuously?
    \item Is there a discrete set of possible traits, or can we model the trait space as Euclidean space?
    \item Do we consider pure replication or additional mutation?
    \item Are we tracking the total mass/individual count of the population, or are we only interested in relative frequencies of traits to one another?
\end{itemize}
In the end we will fix a model to work with for the rest of this article (a continuous-time, continuous-trait replicator-mutator equation), and we mention a few models that we will not cover (but which are similar and also interesting directions for future research).

\subsection{Replicator equations}
We start with pure replication, considering mutation in a later section. The replicator equation which is possibly easiest to motivate is the discrete-in-time, discrete-trait replicator equation which comes in the following two forms:

\paragraph{Discrete-time discrete-trait (d-d) replicator equation}
\begin{align}
   \label{unnormddre} q_i(t_{n+1}) &= q_i(t_n) \cdot \rho_i(t_n, q) &&\text{(unnormalised)}\\
    \label{normddre}p_i(t_{n+1}) &= p_i(t_n) \cdot \frac{\rho_i(t_n, q)}{\bar \rho(t_n, q)} &&\text{(normalised)}
\end{align}

The unnormalised form of the d-d replicator equation models a distribution of traits $i$ in discrete time steps, where trait $i$ has a prevalence $q_i(t_n)$ at time $t_n$, and replicates with a factor of $\rho_i$, possibly depending on both time $t_n$ and the current distribution $q = \{q_i\}_i$. The normalised form works in the same way, but conserves total mass, so starting with a probability distribution $\{p_i(t_0)\}_i$ ensures that $\{p_i(t_n)\}_i$ is also a discrete probability distribution for all time steps, with $p_i = q_i / \sum_j q_j$ and $\bar \rho(t_n, q) = \sum_i p_i(t_n)\rho_i(t_n)$.

First work on these models in the context of evolution was done by \cite{crowkimura} (and this equation is also called the Crow-Kimura equation), followed up by \cite{akin1979geometry,schuster1983replicator}, and many others. 

Clearly \eqref{normddre} can be derived from \eqref{unnormddre} by seeing that
\begin{align*}
 p_i(t_{n+1}) &= \frac{q_i(t_{n+1})}{\sum_j q_j(t_{n+1})} = \frac{ q_i(t_n) \cdot \rho_i(t_n, q)}{\sum_j  q_j(t_n) \cdot \rho_j(t_n, q)} = \frac{ \frac{q_i(t_n)}{\sum_j  q_j(t_n)} \cdot \rho_i(t_n, q)}{\sum_j  \frac{q_j(t_n)}{\sum_k  q_k(t_n)} \cdot \rho_j(t_n, q)} = \frac{p_i(t_n)\cdot \rho_i(t_n, q)}{\sum_j p_j(t_n)\rho_j(t_n, q)}.
\end{align*}

\begin{rem}
    The normalised d-d replicator equation \eqref{normddre} is not self-contained in its full generality: Note that the density-dependent term $\rho_i$ could in principle depend on the unnormalised distribution $q=\{q_j\}_j$ (rather than just on relative frequencies $p=\{p_j\}_j$). For that reason, in contexts where the normalised version is studied, the replication term is usually assumed to just depend on relative frequencies, so that
    \[p_i(t_{n+1}) = p_i(t_n) \cdot \frac{\rho_i(t_n, p)}{\bar \rho(t_n, p)},\]
    which is a self-contained equation without reference to the (possibly unknown unnormalised distribution $q$), but restricts to models where replication cannot depend on total population size (which is relevant for environments with a finite carrying capacity). This is true for all variants of the replicator(-mutator) equation further below. Since this point does not bother us too much, we will hide frequency-dependence behind plain time-dependence, i.e., a replication term $\pi_i(t)$ will be allowed to depend on the full distribution at time $t$.
\end{rem}

By letting the time steps $t_{n+1}-t_n$ go to $0$, we obtain a model with comparable dynamical properties, in continuous time:

\paragraph{Continuous-time discrete-trait (c-d) replicator equation}
The following corresponds to a coupled system of $d$ ODEs where $d$ denotes the number of traits (e.g., fixed number of gene sequences of a certain length)  
\begin{align}
    \label{unnormcdre}\dot q_i(t) &= q_i(t) \cdot \pi_i(t)&&\text{(unnormalised)}\\    
    \label{normcdre}\dot p_i(t) &= p_i(t) \cdot (\pi_i(t) - \bar \pi(t)) &&\text{(normalised)}
\end{align}

Here, $p_i = q_i / \sum_j q_j$, $\bar\pi(t) = \sum_i p_i(t)\pi_i(t)$, and  \eqref{normcdre} can be derived from \eqref{unnormcdre} by seeing that
\begin{align*}
    \dot p_i(t) &= \frac{\d }{\d t}\left(\frac{q_i(t)}{\sum_j q_j(t)} \right) = \frac{\dot q_i(t) \sum_j q_j(t) - q_i(t) \sum_j \dot q_j(t)}{\left(\sum_j q_j(t)\right)^2} = \frac{\dot q_i(t)}{\sum_j q_j(t)} - \frac{q_i(t)}{\sum_j q_j(t)}\cdot \frac{\sum_j \dot q_j(t)}{\sum_j q_j(t)}\\
    &= p_i(t) \pi_i(t) - p_i(t) \cdot \frac{\sum_j q_j(t) \cdot \pi_j(t)}{\sum_k q_k(t)} = p_i(t)\left( \pi_i(t) - \sum_j p_j(t) \cdot \pi_j(t)\right)
\end{align*}

The d-d and c-d replicator equation are different from each other for non-vanishing time steps, but there is a canonical translation between the formulations in the fine time discretisation limit, which is the content of the next lemma.
\begin{lem}[Connection between the d-d and c-d replicator equation]\label{lem:conn_rep}
    The d-d and the c-d replicator equation are asymptotically equivalent for $t_n = nh$, and if we set $\rho_i(t_n) = e^{h\pi_i(t_n)}$. This is to be understood in the sense that $q_i^{cd}(t_n) \simeq q_i^{dd}(t_n)$ for $h\to 0$ and $q^{cd}$, $q^{dd}$ are the c-c and c-d variants, respectively (outside of the context of this lemma we are not going to distinguish this notationally).
\end{lem}
\begin{proof} We show this by comparing the difference equation for the d-d equation with the differential equation for the c-d equation in the fine time limit, recognising that they are asymptotically identical:
    \begin{align*}
        \frac{q_i^{dd}(t_{n+1}) - q_i^{dd}(t_n)}{h}&=\frac{\rho_i(t_n) - 1}{h}q_i^{dd}(t_n) = \frac{e^{h\pi_i(t_n)}-1}{h}q_i^{dd}(t_n) \\
        &\simeq \pi_i(t_n)q_i^{dd}(t_n).
    \end{align*}
    The same statement for $p^{dd}$ and $p^{cd}$ then follows from the fact that they can be derived from the expression for $q^{dd}$ and  $q^{cd}$, respectively.
\end{proof}

\begin{rem}\label{rem:wright_malthus}
    The preceding lemma clarifies that a d-d replicator equation of type $ q_i(t_{n+1}) = q_i(t_n) \cdot \rho_i(t_n)$ is \textit{not} comparable to a c-d equation of form $\dot q_i(t) = q_i(t) \rho_i(t)$, but rather that there is an exponential-logarithm translation necessary between these forms. For example, the continuous-time version of the discrete-time dynamics $q(t_{n+1}) = q(t_n) \cdot 1$ is $\dot q(t) = q(t) \cdot 0$ since (essentially) $0  = \log 1$. This is related to the ideas of \textit{Malthusian} fitness (corresponding to the replication term $\pi_i$ in the c-d equation) vs. \textit{Wrightian} fitness (corresponding to the replication term $\rho_i$ in the d-d equation), see \cite{wu2013interpretations} for a more elaborate discussion.
\end{rem}

Generalisation to continuous trait spaces is done simply by replacing the index $i$ by a continuous label $x$:

\paragraph{Discrete-time continuous-trait (d-c) replicator equation}
The following describes the distribution of values of a single trait $x \in \mathbb{R}$ or $x \in \mathbb{R}^+$ or a set of $d$ traits taking on a continuum of values $x \in \mathbb{R}^d$ (e.g. neck length).  
\begin{align}
    \label{unnormdcre}q(t_{n+1}, x) &= q(t_n, x) \cdot \rho(t_n, x)&&\text{(unnormalised)}\\    
    \label{normdcre}p(t_{n+1}, x) &= p(t_n, x) \cdot \frac{\rho(t_n, x)}{\bar \rho(t_n)} &&\text{(normalised)}
\end{align}
Here, $p(t_n,x) = \frac{q(t_n,x)}{\int q(t_n,z)\d z}$, $\bar \rho(t_n) = \int \rho(t_n,x)p(t_n,x)\d x$, and \eqref{normdcre} can be derived from \eqref{unnormdcre} in a similar way as shown above for the d-d replicator equation. The meaning of $q(t,\cdot)$ is that of a mass density of traits in a population, i.e., $\int_A q(t,x)\d x$ is the total mass/number of traits $x\in A$ at time $t$. The distribution $p$ is normalised to $1$ at all times, and thus corresponds to a probability density (given sufficient regularity that the density does not collapse to something not absolutely continuous). Again, a continuous-time limit is easy to derive:

\paragraph{Continuous-time continuous-trait (c-c) replicator equation} 
\begin{align}
    \label{eq:ccre_unnorm} \partial_t q(t,x) &= q(t,x)  \pi_{q_t}(x)&&\text{(unnormalised)}\\  
    \label{eq:ccre_norm}\partial_t p(t,x) &= p(t,x) (\pi_{p_t}(x) -  \pi_{p_t,p_t})  &&\text{(normalised)}
\end{align}
Here, $p(t,x) = \frac{q(t,x)}{\int q(t,z)\d z}$, $\pi_{q_t}(x) = \frac{\int f(x,z) q(t,z) \d z}{\int q(t,z)\d z}$ = \textit{frequency dependent fitness} of the trait vector $x \in \mathbb{R}^d$ which depends on the so-called \textit{non-local fitness landscape}, $f(x,z): \mathbb{R}^d \times \mathbb{R}^d \rightarrow \mathbb{R}$, and $\pi_{q_t,q_t} = \int \pi_{q_t}(z)q(t,z)\d z / \int q(t,z)\d z = \iint f(x,z)q(t,x)q(t,z)\d(x,z)/(\int q(t,z)\d z )^2$. Again, the normalised form \eqref{eq:ccre_norm} can be obtained from the unnormalised form \eqref{eq:ccre_unnorm} similar to above.

\begin{lem}[Connection between the d-c and c-c replicator equation]
    In the small time-step limit $h\to 0$, where $t_n = n\cdot h$, the d-d and the c-d replicator equation are equivalent if $\rho(t_n,x) = e^{h\pi_{q_{t_n}}(x)}$
\end{lem}
\begin{proof}
    This is proven in the same way as lemma \ref{lem:conn_rep}.
\end{proof}

\subsection{Replicator-mutator equations}

Now we consider additional mutation at time of replication, i.e., the offspring of an individual with trait $i$ is not necessarily also carrying trait $i$, but mutation to a different trait $j$ is allowed to occur, with this mutation being modelled by a probability distribution. It is most straightforward to start with the discrete-time case, and think of mutation as an additional perturbation step following replication, governed by a Markov transition matrix. This means, a (not necessarily normalised) distribution of traits $i$ at time $n$ given by $q_i(n)$ is first evolved via replication in the same way as with the replicator equation, and the resulting offspring population is mutated by a Markov transition matrix $T$, with $T_{ji}$ being the probability for a mutation from trait $j$ to trait $i$. This yields the unnormalised replicator-mutator equation \eqref{unnormddrme}. Note that $T$ is a stochastic matrix, i.e. $\sum_i T_{ji} = 1$ for all $j$. \footnote{In principle, the transition matrix could be allowed to depend on time, and the current distribution of traits, but we are not interested in this level of generality.}

\paragraph{Discrete-time discrete-trait (d-d) replicator-mutator equation}

\begin{align}
\label{unnormddrme} q_i(t_{n+1}) &= \sum_j q_j(t_n) \rho_j(t_n) T_{ji}&&\text{(unnormalised)}\\  
\label{normddrme}    p_i(t_{n+1}) &= \frac{\sum_j p_j(t_n) \rho_j(t_n) T_{ji}}{\bar \rho(t_n)} &&\text{(normalised)}
\end{align}

The normalised form is also sometimes called the Moran-Kingman model, see \cite{cerf_quasispecies_2022}. For a specific choice of fitness and mutation, this becomes the discrete Moran process \cite{moran1958random}.

Here, $p_i(t_n) = \frac{q_i(t_n)}{\sum_k q_k(t_n)}$, $\bar \rho(t_n) = \sum_i p_i(t_n)\rho_i(t_n)$, and \eqref{normddrme} can be derived from \eqref{unnormddrme} as follows:
\begin{align*}
    p_i(t_{n+1}) &= \frac{q_i(t_{n+1})}{\sum_k q_k(t_{n+1})} = \frac{\sum_j q_j(t_n)\rho_j(t_n) T_{ji}}{\sum_{j,k} q_j(t_n)\rho_j(t_n) T_{jk}} = \frac{\sum_j \frac{q_j(t_n)}{\sum_m q_m(t_n)}\rho_j(t_n) T_{ji}}{\sum_{j} \frac{q_j(t_n)}{\sum_m q_m(t_n)}\rho_j(t_n) }
\end{align*}

By letting the time steps in the d-d replicator mutator equation $t_{n+1}-t_n$ become arbitrarily small, we obtain the following:

\paragraph{Continuous-time discrete-trait (c-d) replicator-mutator equation}
In the time-continuous setting, the replicator-mutator equation becomes the following expression, which features a matrix $G$ which is the generator of a Markov transition matrix. In particular, $\sum_i G_{ji} = 0$, which means that the mutation procedure does not create additional mass, but rather distributes $q$ in a certain way. The connection between $G$ and $T$ (as well as between $\rho$ and $\pi$) is fleshed out in more detail in lemma \ref{lem:conn_rep_mut} below.
\begin{align}
   \label{unnormcdrme} \dot q_i(t) &= q_i(t)\cdot \pi_i(t) + \sum_j q_j(t) G_{ji}&&\text{(unnormalised)}\\  
\label{normcdrme}    \dot p_i(t) &= p_i(t)\cdot (\pi_i(t) - \bar\pi(t)) + \sum_j p_j(t) G_{ji} &&\text{(normalised)} 
\end{align}
 Here, $p_i = q_i / \sum_j q_j$, $\bar\pi(t) = \sum_i p_i(t)\pi_i(t)$, and  \eqref{normcdre} can be derived from \eqref{unnormcdre} by seeing that
\begin{align*}
    \dot p_i(t) &= \frac{\d }{\d t}\left(\frac{q_i(t)}{\sum_j q_j(t)} \right) = \frac{\dot q_i(t) \sum_j q_j(t) - q_i(t) \sum_j \dot q_j(t)}{\left(\sum_j q_j(t)\right)^2} = \frac{\dot q_i(t)}{\sum_j q_j(t)} - \frac{q_i(t)}{\sum_j q_j(t)}\cdot \frac{\sum_j \dot q_j(t)}{\sum_j q_j(t)}\\
    &= p_i(t) \pi_i(t) + \sum_j p_j(t) G_{ji}- p_i(t) \cdot \frac{\sum_j \left(q_j(t) \cdot \pi_j(t) + \sum_k q_k(t)G_{k,j}\right)}{\sum_k q_k(t)} \\
    &= p_i(t)\left( \pi_i(t) - \sum_j p_j(t) \cdot \pi_j(t)\right) + \sum_j p_j(t) G_{ji},
\end{align*}
where in the last equality we used that $\sum_j G_{kj} = 0$. The normalised c-d replicator-mutator equation is also called the Crow-Kimura equation \cite{crowkimura}. The fitness function $\pi_i(t)$ is sometimes taken to be a linear map of prevalent traits, i.e., $\pi_i(t) = \sum_j A_{ij}q_j(t)$, as in \cite{adams2007analysis,morsky2017homophilic}, or even more simply (in the two-species case) just linearly dependent on the other species, as in \cite{cooney2022long}.

The d-d and c-d replicator-mutator equation are different from each other for non-vanishing time steps (just like the pure replicator equation), but there is a canonical translation between the formulations in the fine time discretisation limit.
\begin{lem}[Connection between the d-d and c-d replicator-mutator equation]\label{lem:conn_rep_mut}
    In the small time-step limit $h\to 0$, where $t_n = n\cdot h$, the d-d and the c-d replicator-mutator equation are equivalent if $\rho_i(t_n) = e^{h\pi_i(t_n)}$ and $T = e^{h G}$, i.e., $G$ is the generator of the Markov transition matrix $T$. This is to be understood in the sense that $q_i^{cd}(t_n) = q_i^{dd}(t_n)$ and $q^{cd}$, $q^{dd}$ are the c-c and c-d variants, respectively.
\end{lem}
\begin{proof}
We prove this on the vector $(q_i(t_n))_i$ directly, and notationally we set $R = \diag(\rho_i(t_n))_i$ and $L = \diag(\pi_i(t_n))$. Then from \eqref{unnormddrme},
\begin{align*}
    q(t_{n+1})^\top &= q(t_n)^\top\cdot R\cdot T = q(t_n)^\top \cdot \frac{R\cdot T - \mathbf I}{h} = q(t_n)^\top\cdot \frac{e^{hL}e^{hG} - \mathbf I  }{h} = q(t_n)^\top\cdot \frac{\left(\left[e^{hL}e^{hG}\right]^{1/h}  \right)^h - \mathbf I}{h}\\
    &\simeq q(t_n)^\top\cdot \frac{\left( e^{L+G}\right)^h - \mathbf I}{h} = q(t_n)^\top\cdot \frac{e^{h(L+G)} - \mathbf I}{h} \simeq q(t_n)^\top\cdot (L+G),
\end{align*}
where we used the Lie product formula, and the last expression is the vectorised form of \eqref{unnormddrme}.
\end{proof}

\begin{rem}
    Note that for discrete time, it matters whether replication or mutation happens first, i.e. the following version of the discrete-discrete replicator-mutator equation is not equivalent to \eqref{unnormddrme} and could justifiably be called the \textbf{discrete-time discrete-time mutator-replicator\footnote{Note the swapped order of replication and mutation.} equation:}
    \begin{align*}
        q_i(t_{n+1}) &= \left(\sum_j q_j(t_n) T_{ji}\right)\rho_i(t_n)
    \end{align*}
    In the time-continuous limit, though, this difference vanishes. This can be seen by observing that in the proof of lemma \ref{lem:conn_rep_mut}, the Lie product formula produces a term which does not depend on the order of application of $R$ and $T$, since 
    \begin{align*}
        \left[e^{hL}e^{hG}\right]^{1/h} &\simeq e^{L + G} \simeq \left[e^{hG}e^{hL}\right]^{1/h}.
    \end{align*}
    As a result of this fact, the continuous-discrete replicator-mutator equation has just one canonical form \eqref{unnormcdrme}, and there is no separate continuous-discrete mutator-replicator equation. In some sense, this equation exhibits simultaneous replication and mutation by virtue of the fact that these two effects are coupled by a sum rather than a matrix product (which is what happens for the discrete-discrete replicator-mutator equation \eqref{unnormddrme}).
\end{rem}

\begin{rem} The fact that the matrix $G$ in the c-d replicator-mutator equation is a generator matrix, and not a stochastic matrix, is critical. Its entries $G_{ji}$ cannot be given the interpretation of a probability (of mutation from $j$ to $i$), which is is overlooked in several publications. For example, if $G$ was a stochastic matrix, we could not have derived the normalised c-d replicator equation in the form above, and it would not conserve mass. 
    Similarly, the following is \textit{not} a continuous-time replicator-mutator equation (but rather again a type of quasispecies equation):
    \begin{equation}
        {\dot {q_{i}}}=\sum _{j=1}^{n}{q_{j}\pi_{j}(x)G_{ji}}-\bar \pi(t) (x)q_{i},
    \end{equation}
    regardless of whether $G$ is chosen to be a generator or a stochastic matrix. This can lead to some confusion, but the reasoning is similar to what was discussed in remark \ref{rem:wright_malthus}: Going from the discrete-time to the continuous-time setting we need to essentially perform a logarithmic transformation from ``Wrightian fitness and mutation'' to a ``Malthusian fitness and mutation'' (although this is not a standard term), which means that the product between fitness and mutation in \ref{unnormddrme} becomes a sum in \ref{unnormcdrme} (for the same reason that $\log(ab) = \log a + \log b$). This nuance is not always made clear, even in eminent publications in the field like \cite[eqs. (1.2),(1.3)]{Hofbauer1985}.
\end{rem}

The continuous state-space case works similarly:

\paragraph{Discrete-time continuous-trait (d-c) replicator-mutator equation}

The following version which has discrete time steps but a continuous trait space seems to be relatively rare, with only little coverage (exceptions being \cite{burger1994distribution,jones2003stability}), but we describe it for the sake of completeness:

We consider the adjoint semigroup $\cT^\star$ with $(\cT^\star \mu)(\d x) = \int P(y,\d x)\d \mu(y)$ for a transition function $P(\cdot, \cdot)$ of a Markov process. For notational simplicity we will write $\cT^\star \mu(x)$ for the density of $\cT^\star \mu$ in $x$. Then the following two equations describe replication and mutation for discrete time and a continuous trait space:

\begin{align}
    \label{unnormdcrme}q(t_{n+1}, x) &= \cT^\star (q(t_n,\cdot)\rho(t_n,\cdot))(x)&&\text{(unnormalised)}\\
    \label{normdcrme}p(t_{n+1}, x) &= \cT^\star \left(q(t_n,\cdot)\frac{\rho(t_n,\cdot)}{\bar\rho(t_n)}\right)(x)&&\text{(normalised)}
\end{align}
Here, $p(t_n,x) = \frac{q(t_n,x)}{\int q(t_n,z)\d z}$, $\bar \rho(t_n) = \int \rho(t_n,x)p(t_n,x)\d x$, and \eqref{normdcre} can be derived from \eqref{unnormdcre} like this:
\begin{align*}
    p(t_{n+1}, x) &= \frac{q(t_{n+1},x)}{\int q(t_{n+1},w)\d w} = \frac{\int q(t_n,z)\rho(t_n,z)P(z,x)\d z}{\iint q(t_n,z)\rho(t_n,z)P(z,w)\d (w,z)} = \frac{\int q(t_n,z)\rho(t_n,z)P(z,x)\d z}{\int q(t_n,z)\rho(t_n,z)\d z} \\
    &= \frac{\int \frac{q(t_n,z)}{\int q(t_n,w)\d w}\rho(t_n,z)P(z,x)\d z}{\int \frac{q(t_n,z)}{\int q(t_n,w)\d w}\rho(t_n,z)\d z} = \frac{\int p(t_n,z) \rho(t_n,z) P(z,x)}{\int p(t_n,z)\rho(t_n,z)\d z}
\end{align*}

We finally arrive at the model which we are going to use for the rest of the manuscript, which models continuous-time replication and mutation in a continuous trait space. The following is very general, but we will work with a specific form of this replicator-mutator equation, where mutation is driven by Gaussian noise and replication is mediated by a quadratic fitness functional.

\paragraph{Continuous-time continuous-trait (c-c) replicator-mutator equations}
$\mathcal R$ is a generator of a semigroup, i.e. $\mathcal R[\mathbf 1] = 0$ where $\mathbf 1$ is the constant function with value $1$.
\begin{align}
    \partial_t q(t,x) &= q(t,x)\pi_{q_t}(x) + \mathcal Rq(t,x)\\    
    \partial_t p(t,x) &= p(t,x)(\pi_{q_t}(x) - \pi_{q_t,q_t}) + \mathcal Rp(t,x)
\end{align}

\begin{lem}[Connection between the d-c and c-c replicator-mutator equation]\label{lem:conn_crep_mut}
    In the small time-step limit $h\to 0$, where $t_n = n\cdot h$, the d-d and the c-d replicator-mutator equation are equivalent if $\rho_i(t_n) = e^{h\pi_i(t_n)}$ and $\cT^\star = e^{h \mathcal R^\star}$, i.e., $\mathcal R^\star$ is the generator of the semigroup $\cT^\star$.
\end{lem}
\begin{proof}
This is proven by extension of the same formal idea in lemma \ref{lem:conn_rep_mut}.
\end{proof}

\begin{rem}
It should be noted that there are also variants possible, with mutation being an additive term, as described in \cite{Hofbauer1985}, but this is beyond the scope of this manuscript. Also related is the quasispecies equation (sometimes referred to as eigenspecies equation) \cite{Eigen1988}
\begin{align}
    \label{eq:eigenctstimectstrait}
    \text{Quasispecies model:} &\quad  \partial_t p(t,x) = \int f(x,y) q(x,y) p(t,y) dy - \iint f(z,y)p(t,z)p(t,y)\d z \d y \cdot p(t,x),  
\end{align} which has different dynamical properties, and which we will not be considering in the following. 
\end{rem}

\section{Explicit solutions for the replicator-mutator equation under the influence of a quadratic fitness function} \label{sec:explicit}
In this section we present some explicit solutions and analytical properties of the unnormalised replicator-mutator equation
\begin{align}
    \label{eq:repmutlin}
    \partial_t q(t,x) = \underbrace{- \nabla \cdot (q(t,x) Gx) + \frac{1}{2}  \nabla \cdot (\Sigma \nabla q(t,x))}_{\text{mutation}}  +  \underbrace{q(t,x)\pi_{q_t}(x) }_{\text{replication}}
\end{align}
where the fitness function $\pi_{q_t}$ is given by
\begin{align}
        \label{eq:utilitylinear}
        \pi_{q_t}(x) &= \frac{\int f(x,z)q_t(z)\d z}{\int q_t(z) \d z}
\end{align}
with fitness payoff function
\begin{align}
        f(x,z) &= -\frac{1}{2}\|Ax-y(t)\|_\Gamma^2 -\frac{1}{2}\|Az-y(t)\|_\Gamma^2 + s \langle Ax-y(t), Az-y(t)\rangle_\Gamma + K \\
        &= -\frac{s}{2}\|Ax - Az\|_\Gamma^2 -\frac{1-s}{2}\|Ax-y(t)\|^2_\Gamma -\frac{1-s}{2}\|Az-y(t)\|^2_\Gamma + K.
\end{align}

This payoff function $f(x,z)$ describes the evolutionary fitness of an individual with trait $x$ in relation to an individual with trait $z$, and depending on a possibly time-dependent observation $y(t)$. A slightly simplified version of this fitness function can be found in \cite{cressman2006stability}. The object $y(t)$ is to describe a target ``feature'' with optimal fitness, i.e. individuals with a trait $x$ such that the corresponding feature $Ax$ satisfies $Ax \approx y(t)$ are better adapted at time $t$ than others. The constant $K$ describes the overall fitness-independent reproductive capacity, and $s\in(-\infty, 1]$ modulates the effect of ``community'' or ``uniformity'' on an individual's fitness. 

By integrating over $q_t$ in the definition of the fitness function in \eqref{eq:utilitylinear}, the fitness $\pi_{q_t}(x)$ is the averaged fitness utility of $x$ in the context of its ambient population $q_t$.

The following computation will be repeatedly useful: If $q_t =P(t)\cdot \mathcal N(m(t),C(t))$ is an unnormalised multivariate Gaussian with total mass $P(t)$, then the fitness $\pi_{q_t}$ of trait $x$ in a population $q_t$ is
\begin{equation}
\begin{split}
    \label{eq:fitness}
    \pi_{q_t}(x) &= P(t)^{-1}\int f(x,z) q_t(z)\d z = \int f(x,z) \d N(m,C)(z)\\
    &=  -\frac{1}{2}\|Ax-y(t)\|_\Gamma^2 -\frac{1}{2}\|Am-y(t)\|_\Gamma^2 + s \langle Ax-y(t), Am(t)-y(t)\rangle_\Gamma - \tr(A^\top \Gamma^{-1} A C) + K.
\end{split}
\end{equation}

\begin{rem}
    The specific type of mutation in \eqref{eq:repmutlin} corresponds to an Ornstein-Uhlenbeck diffusion $dX_t = GX_t dt + \Sigma^{1/2} dB_t $ on trait space. Most relevant biological applications are likely to have $G=0$, i.e., mutation is modelled by a random walk.
\end{rem}

\begin{rem}
    We will see later that we can also consider the simpler setting where $f(x,z) = -\frac{1}{2}\|Ax-y(t)\|_\Gamma^2 + K$, i.e., fitness of trait $x$ does not depend other individuals and their traits. 
\end{rem}

\begin{rem}
    The role of the parameter $s\in (-\infty, 1]$ is to provide a way of modelling fitness benefits of ``staying in the herd'', i.e., of an individual trait $x$ being close to the population mean. If $s=0$, then there is no such effect. To elucidate this argument, we choose a multivariate Gaussian distribution $q_t$ so we can use the simple analytical expression \eqref{eq:fitness} for the fitness.
    \begin{itemize}
            \item If $s=0$, then clearly $\argmax_x \,\pi_{q_t}(x) = \argmin_x\, \|Ax-y(t)\|_\Gamma^2$, i.e. the fittest trait is the one with the best match with the optimal feature $y(t)$.
        \item $s =1$: Since in this case, $\pi_{q_t}(x) = \|Ax-Am\|_\Gamma^2$, traits $x$ with the property $Ax = Am$ are most fit. In a certain sense, ``conformity'' or ``collaboration'' (in the sense of being close to the population's feature average) is beneficial in such an environment. In this extreme case, the optimal feature $y$ does not impact the evolutionary dynamics, and fitness is communicated purely in terms of conformity.
        \item $s \in(0,1)$ corresponds to a fitness function with a trade-off between maximising utility and conformity, respectively.  
        \item $s < 0$ corresponds to a utility function that benefits from non-conformity: Even if $x$ is quite close to minimising the misfit $\|Ax-y\|_\Gamma^2$, individual fitness can be improved by moving a bit further away from the population mean, even at the cost of increasing the misfit term.
    \end{itemize}
\end{rem}

We will see that in this specific setting, if the initial distribution $q(0,\cdot)$ is an unnormalised multivariate Gaussian distribution, then $q(t,\cdot)$ stays a multivariate Gaussian for all times $t$, and the mean and covariance, as well as the total mass $\int q(t,x)\d x$ follow simple and (in some cases) explicitly solvable differential equations.

Note that the standard formulation of the replicator-mutator equation, only considers relative frequency of traits, i.e., we don't track overall population growth or decline, only the behaviour of the relative frequency of any given trait. We can recover this standard (normalised) replicator-mutator equation by modifying the replication term such that it conserves mass (the mutation term already conserves mass):
\begin{align}
    \label{eq:repmutlin_normalised}
    \partial_t q(t,x) = \underbrace{- \nabla \cdot (q(t,x) Gx) + \frac{1}{2}  \nabla \cdot (\Sigma \nabla q(t,x))}_{\text{mutation}}  +  \underbrace{q(t,x)(\pi_{q_t}(x) -  \pi_{q_t,q_t}) }_{\text{replication}}
\end{align}

with the population $q_t$'s mean fitness $\pi_{p_t,p_t} = P(t)^{-1}\int \pi_{p_t}(x)p_t(x)\d x$. 
Note that the overall reproduction capacity parameter $K$ vanishes by subtracting mean fitness. 
This form of quadratic fitness function has been analysed in simplified form in \cite{cressman2006stability} for $y = 0$ and $s=0$. In the special case of $s = 1$, this is equivalent to
\[\pi_{p_t}(x) - \pi_{p_t, p_t} = -\frac{1}{2}\|A(x-m(t))\|_\Gamma^2 + \text{const.}\]
But since incorporating total population size does not increase mathematical complexity much, we consider the unnormalised version \eqref{eq:repmutlin} here, because it allows us to study population growth and decay, as well as extinction events, depending on the various parameters involved.

If we want to recover the probability distribution predicted by the usual (normalised, constant total population mass $=1$) replicator-mutator equation, we just need to normalise, i.e. divide by the population size $P(t)$.

We will prove the following statements in the remainder of the article:
\begin{itemize}
    \item Lemma \ref{lem:moment_eq} in section \ref{sec:moments} shows that the replicator-mutator equation conserves Gaussianity, i.e., if we start with an initial multivariate (unnormalised) Gaussian $q_0$, then for all future times $t$, the measure $q_t$ is also an unnormalised Gaussian. Then we derive ordinary differential equations (ODE) for the dynamics of the mean trait $m(t)$, the covariance of the distribution $C(t)$, and the population's total mass $P(t)$, which we will call \textit{moment equations} (since $P,m,C$ are the $0^{\mathrm{th}}$, $1^{\mathrm{st}}$ and $2^{\mathrm{nd}}$ moments of $q_t$).
    \item Section \ref{sec:solutions_replication} considers the replicator-mutator equation for the special case of $\Sigma = 0$ and $G=0$, which corresponds to the no-mutation case (pure replicator equation). Lemma \ref{lem:cov_ccrep} explicitly solves the ODE for the covariance $C$ derived in the preceding section. The remaining two lemmata explicitly solve the ODE for the mean $m$, and characterises its asymptotic ($t\to\infty$) behavior in two special cases: Lemma \ref{lem:explicitsoln} considers the case  of a constant optimal feature $y(t) = y$, i.e., the fitness landscape does not depend on time. Lemma \ref{lem:explicitsoln_nonconst} allows an optimal feature $y(t)$ varying in time, but we require $s=0$ (the ``fitness via uniformity'' coefficient) out of mathematical convenience. 
    \item In section \ref{sec:regularised_opt} we point out connections to inverse problems and Tikhonov regularisation.
    \item Section \ref{sec:repl_mut_solutions} considers the full replicator-mutator equation for non-vanishing $\Sigma$, but $G=0$, i.e. for a random walk-type mutation. In a similar way to the preceding section, lemma \ref{lem:repmuttime} first solves the ODE for the covariance explicitly, then lemma \ref{lem:repmuttime_mean} derives a characterisation of the dynamics of the mean, with corollary \ref{cor:expl_special} recording some particularly interesting special cases.
    \item Section \ref{sec:examples} applies the results from the preceding section \ref{sec:explicit} to demonstrate, illustrate, and explain a variety of phenomena which have made their name in evolutionary and mathematical biology, but which often relied on computationally intensive simulation studies (e.g., via simulation of a large number of ``adaptive walks'' or genetic algorithms). We show that many of these phenomena can be explained by simple mathematical expressions and arguments, once the explicit results derived in section \ref{sec:explicit} can be leveraged. These interesting phenomena include the \textit{flying kite effect}, the principle of \textit{survival of the flattest}, and both the population dynamics and the issue of extinction when the optimal feature $y(t)$ moves with constant velocity (exhibting the \textit{fixed lag effect}), with bounded speed, in an oscillating motion, or based on random fluctuations.
\end{itemize}

\subsection{Derivation of moment equations}\label{sec:moments}

Our first main result is the following lemma containing a) the statement that the population conserves Gaussianity of an initial distribution, and b) the differential equations governing the mean, covariance, and total mass of the distribution under the influence of selection and mutation over time, as modelled by \eqref{eq:repmutlin}.

Note that there is a strong connection between these moment equations and the mathematical theories of filtering and control \cite{kalman1960contributions,bucy1967global}, where similar or virtually identical ordinary differential equations (ODEs) appear. There is a large body of research on how to solve these equations. Older references typically work on analytical solutions, while newer research tends to focus on numerical methods. In the context of filtering, numerical procedures (like the Ensemble Kalman filter \cite{evensen2003ensemble}) tend to be popular because they also allow dimensionality reduction and to deal with nonlinear mappings between trait space and feature space. In this manuscript, we are interested in explicit solution because these will allow us to showcase elementary special scenarios of the behavior of the replicator-mutator equation. In addition, we are not necessarily motivated by extremely high-dimensional trait spaces, so linear algebra operations like matrix inversion are cheap enough to make these explicit solutions feasible (which is the crux of filtering and control settings, requiring numerical approaches).

\begin{lem}\label{lem:moment_eq}
We consider the replicator-mutator equation \eqref{eq:repmutlin}, restated here for convenience: 
\begin{align*}
    \partial_t q(t,x) = - \nabla \cdot (q(t,x) Gx) + \frac{1}{2}  \nabla \cdot (\Sigma \nabla q(t,x))  +  q(t,x)\pi_{q_t}(x) 
\end{align*}
and the utility function (\ref{eq:utilitylinear}) tracking a possibly time-dependent optimal feature $y(t)\in \R^k$ with $A\in \R^{k\times d}$,  and $s \leq 1$. 

    Furthermore, assume that $q_0 = P_0\cdot \mathcal N(m_0,C_0)$ is a multivariate Gaussian measure with mass $P_0$ (i.e., $q_0$ is not necessarily a probability measure, only a finite measure). Then the following statements hold:
    The replicator-mutator equation describes the evolution of an unnormalised multivariate Gaussian measure $q_t = P(t)\cdot \mathcal N(m(t),C(t))$ with total mass $P(t)$ for all $t\geq 0$, and its mean, covariance, and total mass satisfy the following system of ODEs: 
    \begin{equation}\label{eq:moment_ODEs}\begin{split}
        \dot m(t) &= Gm(t) - (1-s )C(t)A^\top \Gamma^{-1}\left( Am(t) - y(t)\right)\\
        \dot C(t) &= GC(t) + C(t)G^\top + \Sigma -C(t)A^\top \Gamma^{-1} A C(t)\\
        \dot P(t) &= P(t)\cdot \Big(K  -(1-s)\|Am(t)-y(t)\|_\Gamma^2   -\tr [A^\top \Gamma^{-1}AC(t)] \Big)
    \end{split}\end{equation}
    
\end{lem}

\begin{proof}[Proof of lemma \ref{lem:moment_eq}]
  We start by proving that $q_t$ remains the density of an unnormalised multivariate Gaussian measure for all times $t\geq 0$, given $q_0$ is Gaussian.

The following equations are again true if $q(t,\cdot)$ is a multivariate Gaussian measure ($\nabla$ is always implicity assumed to refer to derivatives with respect to the trait variable $x$, not time $t$).
    \begin{align*}
        \frac{\nabla q(t,x)}{q(t,x)} &= - C^{-1}(x-m)\\
        \frac{\nabla \cdot (q(t,x)Gx)}{q(t,x)} &= - (x-m)^\top C^{-1}Gx + \tr G \\
        &= -\frac{1}{2}x^\top C^{-1}Gx - \frac{1}{2}x^\top G^\top C^{-1} x+ x^\top G^\top C^{-1}m + \tr G\\
        \frac{\nabla\cdot (\Sigma \nabla q(t,x))}{q(t,x)}&= (x-m)^\top C^{-1}\Sigma C^{-1}(x-m) - \tr [\Sigma C^{-1}]
    \end{align*}
    Using these differential properties, we can derive (again under the assumption of Gaussianity)
    \begin{align*}
        \frac{\d }{\d t}\log q(t,x) &= \frac{\partial_t q(t,x)}{q(t,x)} = -\frac{\nabla \cdot(q(t,x) Gx)}{q(t,x)} + \frac{1}{2}\frac{\nabla \cdot (\Sigma \nabla q(t,x))}{q(t,x)} + \pi_{q_t}(x) \\
        &=\frac{1}{2}x^\top C^{-1}Gx + \frac{1}{2}x^\top G^\top C^{-1} x- x^\top G^\top C^{-1}m - \tr G\\
        &+ \frac{1}{2}(x-m)^\top C^{-1}\Sigma C^{-1}(x-m) - \frac{1}{2}\tr [\Sigma C^{-1}] \\
        &-\frac{1}{2}\|Ax-y\|_\Gamma^2  -\frac{1}{2}\|Am-y\|_\Gamma^2 -\frac{1}{2}\tr A^\top \Gamma^{-1}AC + s \langle Ax-y, Am-y\rangle_\Gamma + K \\
        &= -\frac{1}{2} x^\top \left[-C^{-1}G - G^\top C^{-1} - C^{-1}\Sigma C^{-1} + A^\top \Gamma^{-1}A \right] x \\
        &+ x^\top \left[-G^\top C^{-1}m - C^{-1}\Sigma C^{-1} m+A^\top \Gamma^{-1} y +sA^\top \Gamma^{-1}(Am -y) \right]\\
        &+ \tr G - \frac{1}{2}\tr[\Sigma C^{-1}]-\frac{1}{2}\|Am-y\|_\Gamma^2 -\frac{1}{2}\tr A^\top \Gamma^{-1}AC  + K\\
        &= -\frac{1}{2}x^\top R x + x^\top b + c       
    \end{align*}
    where 
    \begin{align*}
        R &= -C^{-1}G - G^\top C^{-1} - C^{-1}\Sigma C^{-1} + A^\top \Gamma^{-1}A\\
        b &= -G^\top C^{-1}m - C^{-1}\Sigma C^{-1} m+A^\top \Gamma^{-1} y +sA^\top\Gamma^{-1}(Am -y)
    \end{align*}
    For $t=0$, all of the expressions above are known to be correct, which means that, in particular, $\log q(0,\cdot)$ and $\frac{\d}{\d t} \log q(0,\cdot)$ are quadratic forms. For some $h>0$, then also $\log q(0,\cdot) + h \frac{\d}{\d t} \log q(0,\cdot)$ is quadratic. Iterating this argument, this means that for any positive stepsize $h> 0$, the Euler discretisation in time of \eqref{eq:repmutlin} yields a sequence of exponentials of quadratic multivariate polynomials. Taking the limit $h\to 0$ proves that $\log q(t,\cdot)$ indeed is a quadratic function, too, which proves Gaussianity of all $q(t,\cdot)$. 
    
    In order to derive expressions for the mean $m(t)$, covariance matrix $C(t)$ and mass $P(t)$ of this unnormalised multivariate Gaussian, we take the expressions above and compare them to an alternative way of computing $\frac{\d }{\d t}\log q(t,x)$ using its Gaussian structure:
    \begin{align*}
        \frac{\d }{\d t}\log q(t,x) &= -\frac{\d}{\d t}\frac{1}{2}(x-m(t))^\top C(t)^{-1}(x- m(t))\\
        &= (x-m)^\top C^{-1} \dot m + \frac{1}{2}(x-m)^\top C^{-1}\dot C C^{-1}(x-m)\\
        &= -\frac{1}{2}x^\top \left[- C^{-1}\dot C C^{-1}\right] x + x^\top \left[C^{-1}\dot m - C^{-1}\dot C C^{-1}m \right] + c
    \end{align*}
    Comparing this term by term with the expression below we get
    \begin{align}
        \label{eq:Req}- C^{-1}\dot C C^{-1} &= -C^{-1}G - G^\top C^{-1} - C^{-1}\Sigma C^{-1} + A^\top \Gamma^{-1}A\\
        \label{eq:beq}-G^\top C^{-1}m - C^{-1}\Sigma C^{-1} m&+A^\top \Gamma^{-1} y +sA^\top\Gamma^{-1}(Am -y) =C^{-1}\dot m - C^{-1}\dot C C^{-1}m
    \end{align}
    The first equation \eqref{eq:Req} is equivalent to $\dot C = GC + CG^\top + \Sigma - rCA^\top \Gamma^{-1} AC$, which is the ODE for the covariance matrix in \eqref{eq:moment_ODEs}.
    The second equation \eqref{eq:beq} is equivalent to
    \begin{align*}
        \dot m &= -CG^\top C^{-1} m - \Sigma C^{-1}m + A^\top \Gamma^{-1} y + sCA^\top\Gamma^{-1} (Am-y) +\dot C C^{-1} m\\
        &= -C G^\top C^{-1}m -\Sigma C^{-1}m + A^\top \Gamma^{-1} y + sCA^\top\Gamma^{-1}(Am - y)\\
        &+ Gm + CG^\top C^{-1}m + \Sigma C^{-1}m - rCA^\top\Gamma^{-1}Am\\
        &=Gm -(1-s)CA^\top \Gamma^{-1}(Am-y),
    \end{align*}
    which is the ODE for the mean evolution in \eqref{eq:moment_ODEs}.

    It remains to derive the ODE for $P(t)$, but since $P(t) = \int q(t,x)\d x$, this is just
    \begin{align*}
        \dot P(t) &= \frac{\d }{\d t}\int q(t,x)\d x = \int \partial_t q(t,x)\d x = \int \pi_{q_t}(x) q(t,x)\d x = \pi_{q_t,q_t}{\cdot P(t)}\\
         &= \left(-(1-s)\|Am-y\|_\Gamma^2   -\tr A^\top \Gamma^{-1}AC + K\right){\cdot P(t)}\qedhere
    \end{align*}
\end{proof}
\begin{rem} For the case $s=0$, the set of differential equations for $m$ and $C$ in \eqref{eq:moment_ODEs} are exactly the equations characterising the Kalman-Bucy filter \cite{bucy2005filtering}: Consider a hidden state $x(t)$ and an associated time series of measurements $z(t)$ described by the following stochastic differential equations,
\begin{align}
\label{eq:kalmanbucy}
    \d x(t) &= G x(t)\d t + \Sigma^{1/2}\d W_t\\
    \d z(t) &= A x(t) \d t+ \Gamma^{1/2}\d V_t.\label{eq:observ_KB}
\end{align}
for matrices $G,A$ and covariance matrices $\Sigma,\Gamma$, with $W_t,V_t$ being independent Brownian motion terms. 
    
The Kalman-Bucy filter describes the mean $\hat x(t)$ and covariance  $R(t)$ of the Gaussian posterior distribution, given by
     \begin{align}
         \d \hat x(t) &= G\hat x(t) \d t - R(t) A^\top\Gamma^{-1}(A\hat x(t) \d t - \d z(t))\\
         \frac{\d}{\d t} R(t) &= GR(t) + R(t) G^\top + \Sigma - R(t) A^\top \Gamma^{-1}AR(t)
     \end{align}

    If $z(t)$ was replaced by $\d z(t) = Ax(t) \d t$, then it has a derivative $\dot z(t) = A x(t) =y(t)$, and then the ODE for $\widehat {x(t)}$ could equivalently be written as
    \begin{align*}
        \dot {\widehat {x(t)}} &= G\hat x(t) - R(t) A^\top\Gamma^{-1}(A\hat x(t)  - \dot z(t)) \\
        &=G\hat x(t) - R(t) A^\top\Gamma^{-1}(A\hat x(t)  - y(t)),
    \end{align*}
    which is exactly the equation for $m$ in \eqref{eq:moment_ODEs}.
    
    This means that the replicator-mutator equation ``tracking'' a target feature with optimal fitness $y(t)$ has the same dynamics as a motion tracking algorithm for an unknown state $x(t)$ via observations $y(t)$, after the observation process $z(t)$ from \eqref{eq:observ_KB} has been replaced with a smooth and noiseless version $\dot z(t) = Ax(t) =: y(t)$ (but without changing the filter equations). That is, the replicator-mutator equation behaves exactly as the filtering equations for a particular linear-Gaussian dynamical system.

\end{rem}

\begin{rem}
    A simple modification of the proof of lemma \ref{lem:moment_eq} shows that we can also consider the simpler fitness functional 
    \begin{equation}
        \pi_{q_t}(x) = -\frac{1}{2}\|Ax-y\|_\Gamma^2 + K.
    \end{equation}
    This means that the fitness of trait $x$ does not depend on its ambient population, but is just a function of the trait alone. By performing the same quadratic expansion, we can derive
    \begin{equation}\label{eq:moment_ODEs_simpler}\begin{split}
        \dot m(t) &= Gm(t) - C(t)A^\top \Gamma^{-1}\left( Am(t) - y(t)\right)\\
        \dot C(t) &= GC(t) + C(t)G^\top + \Sigma -C(t)A^\top \Gamma^{-1} A C(t)\\
        \dot P(t) &= P(t)\cdot \left(K  -\frac{1}{2}\|Am(t)-y(t)\|_\Gamma^2   -\tr [A^\top \Gamma^{-1}AC(t)] \right)
    \end{split}\end{equation}
    The structure of these equations is almost identical to \eqref{eq:moment_ODEs} and can be solved with the same techniques. Note that there is no choice of parameters that makes \eqref{eq:moment_ODEs} and \eqref{eq:moment_ODEs_simpler} identical, though, since the moment equations for $m$ are identical if and only if $s = 0$, but the moment equations for $P$ are identical if and only if $s = 1/2$. Nevertheless, all solutions derived in later parts of this manuscript can be adapted to this special case.
\end{rem}

\subsection{Pure replication}\label{sec:solutions_replication}
In this section we derive explicit solutions to the differential equations for $m$, $C$, and $P$ in the previous section, under the additional assumption that $G= 0$ and $\Sigma = 0$. This means that we are considering the case of the (no-mutation) replicator equation
\begin{equation*}
    \partial q_t(x) = q(t,x)\pi_{q_t}(x).
\end{equation*}

We treat the (more interesting, but also more involved) replicator-mutator equation in the following section, but here we can showcase some basic techniques that we will continue to use for the replicator-mutator equation.

The biological motivation for considering a system like this is as follows: We imagine an initial distribution of traits $q_0$ (we can pretend that this models a very large population of individuals with traits $x_i$ such that the empirical distribution of the traits is well-characterised by the continuous distribution $q_0$). We also consider an optimal feature $y(t)$, possibly depending on time. The fitness function $\pi_{q_t}$ is given by
\begin{align}
        \pi_{q_t}(x) &= \frac{\int f(x,z)q_t(z)\d z}{\int q_t(z) \d z}
        \intertext{with utility (or payoff) function}
        f(x,z) &= -\frac{1}{2}\|Ax-y(t)\|_\Gamma^2 -\frac{1}{2}\|Az-y(t)\|_\Gamma^2 + s \langle Ax-y(t), Az-y(t)\rangle_\Gamma + K
\end{align} 
and describes how the optimal feature $y(t)$ rewards or adversely affects individuals depending on whether $Ax \approx y(t)$ at time $t$.

Since we consider pure replication, traits proliferate or die out depending on the fitness function, but there is no mutation. Of course, in this case we expect the distribution to concentrate around an optimal trait $x$ such that $Ax \approx y(t)$. Two major difficulties can arise: First, if $y(t)$ indeed depends on time, then selection is also time-dependent, and there will not be an optimal trait fixed in time. This immediately leads to the question whether a population under the influence of pure replication can successfully track a time-dependent optimal feature. This we answer to the negative later.

Second, if $A$ is degenerate (i.e., if several traits $x$ can lead to the same optimal feature $y$ --  i.e., non-injectivity of $A$; or if the optimal feature cannot be achieved by mapping via $A$ -- i.e., non-surjectivity of $A$), linear algebra dictates how the system responds to this issue, and it mainly boils down to projection onto the range of $A$ (if $A$ is not surjective), at minimal distance to the initial mean (if $A$ is not injective).

    The following lemma considers the system \eqref{eq:moment_ODEs} in terms of $m$ and $C$ only. In particular at this point we won't think about whether $P(t)$ decays (i.e., the total population dies out), or stabilises. This is done further below. The first statement for the explicit solution is a well-known result for Riccati equations, see, e.g., \cite{garbuno2020interacting}, but we will also need the more geometric statement about the asymptotic case corresponding to a specific projection, which we prove here quickly for convenience.

\begin{lem}[Covariance of the replicator equation]\label{lem:cov_ccrep}
We consider the special case of the replicator equation
    \begin{align*}
    \frac{dq_t}{dt}(x) =  q(t,x)\pi_{q_t}(x)
\end{align*} The variance of $q_t$ has explicit solution
    \begin{equation}
        C(t) = M(t) C_0,
    \end{equation}
    where $M(t) = (I + rtC_0A^\top \Gamma^{-1}A)^{-1}$. For $t\to\infty$, $M(t)$ converges to the $C_0$-orthogonal projection operator onto the kernel of $A$, i.e. $\lim_{t\to\infty}M(t) = \Pi_{\ke A}^{C_0}$. This means that $\lim_{t\to\infty} C(t) = \Pi_{\ke A}^{C_0}\cdot C_0 $  
\end{lem}
\begin{rem}
In the simplest case where $C_0A^\top\Gamma^{-1}A$ is invertible, this means that $C(t)\xrightarrow{t\to\infty} 0$ (i.e., the population contracts further and further and approaches a Dirac point measure). If $C_0A^\top\Gamma^{-1}A$ is not invertible, for example because $A$ is not injective, then $C(t)$ degenerates in the directions where $A$ is not singular, i.e. where the mediation of $A$ allows the population to get closer to the optimal feature. The generality of the lemma also allows us to consider cases like a degenerate $C_0$ (in which case contraction happens on the orthogonal complement of the $\ke C_0$).
\end{rem}
\begin{proof}[Proof of lemma \ref{lem:cov_ccrep}] From lemma \ref{lem:moment_eq} we know that the distribution $q_t$ stays a multivariate Gaussian and that its covariance follows the ODE in \ref{eq:moment_ODEs}.
       The solution of the matrix ODE is given by $C(t) = (C_0^{-1} + rtA^\top\Gamma^{-1} A)^{-1} = (I + rtC_0A^\top \Gamma^{-1}A)^{-1}C_0$, which can be derived from the ODE for $C^{-1}$ given by $\frac{\d}{\d t}C^{-1}(t) = -C^{-1}(t)\dot C(t)C^{-1}(t) = A^\top \Gamma^{-1} A$, yielding $C^{-1}(t) = C_0^{-1} + rtA^\top \Gamma^{-1} A$, and taking the matrix inverse again.  

        Using the Woodbury matrix identity,
        \begin{align*}
            M(t) &= (I + rtC_0 A^\top \Gamma^{-1}A)^{-1}= (I + C_0^{1/2} (rt I) (C_0^{1/2}A^\top \Gamma^{-1}A))^{-1} \\
            &= I - C_0^{1/2}\left(t^{-1}I + C_0^{1/2}A^\top \Gamma^{-1}AC_0^{1/2}\right)^{-1}C_0^{1/2}A^\top \Gamma^{-1} A \\
            &= I - C_0^{1/2}\left(t^{-1}I + (\Gamma^{-1/2}AC_0^{1/2})^\top (\Gamma^{-1/2}AC_0^{1/2})\right)^{-1} (\Gamma^{-1/2}AC_0^{1/2})^\top \Gamma^{-1/2}A
        \end{align*}
        Now we use that $\lim_{t\to\infty}(t^{-1}I+ Q^\top Q)^{-1} Q^\top = Q^+$ \cite[Theorem 4.3]{barata2012moore}, which means that
        \begin{align*}
            \lim_{t\to\infty} M(t) = I - C_0^{1/2}(\Gamma^{-1/2}AC_0^{1/2})^+ \Gamma^{-1/2}A = C_0^{1/2}\Pi_{\ke (\Gamma^{-1/2}A C_0^{1/2})}C_0^{-1/2} = C_0^{1/2}\Pi_{C_0^{-1/2}\ke A}C_0^{-1/2},
        \end{align*}
        i.e. $\lim_{t\to\infty} M(t) =\Pi_{\ke A}^{C_0}$, where $\Pi_{V}^{Q}$ is the $Q$-orthogonal projection operator onto the subspace $V$, i.e. for all vectors $u$, we have $\Pi_V^Q u \in V$, also $\Pi_V^Q$ is idempotent, and $\langle \Pi_{V}^{Q}u, (I-\Pi_{V}^{Q})(u)\rangle_{Q} = 0$. These properties follow from straightforward calculation using the pseudoinverse's property $(M^+M)^\top = M^+ M$.
        
        Also, immediately by definition of $M(t)$,
        \begin{equation*}
            \lim_{t\to\infty} C(t) =  \Pi_{\ke A}^{C_0}\cdot C_0 = C_0^{1/2}\Pi_{\ke AC_0^{1/2}}C_0^{1/2}.\qedhere
        \end{equation*}
\end{proof}

An explicit solution for the mean ODE in \eqref{eq:moment_ODEs} can be given in a variety of settings. The simplest case is when the optimal feature $y(t)$ does not depend on time. This statement was proven in \cite{bungert2023complete} in the context of the Ensemble Kalman method for inversion, using explicit diagonalisation, and can possibly be considered canonical in the context of filtering, see, e.g., \cite{garbuno2020interacting}. In order to prepare ourselves for the more difficult replicator-mutator equation further below, we are going to show this here in a more direct and illustrative way which can be generalised to the setting of mutation.

\begin{lem}[Mean evolution of the replicator equation for a constant observation $y$]
\label{lem:explicitsoln}
    We consider the moment equations and utility function in \eqref{eq:moment_ODEs}.  Let $y(t) = y$ be a constant optimal feature. Then the following statements about the solution of \eqref{eq:moment_ODEs} hold.
    \begin{enumerate}
    \item \label{item:solution} \textbf{Evolution of mean for constant $y$.}
        There is an explicit analytical solution for the dynamics of the mean of the population,
        \begin{align}\label{eq:mean_weighted}
            m(t) &= M^{1-s}(t) m_0 + \left[I - M^{1-s}(t) \right] (A^\top \Gamma^{-1} A)^{+}A^\top \Gamma^{-1}y\\
                &= m_0 + \left[I - M^{1-s}(t) \right] \left( (A^\top \Gamma^{-1}A)^{+}A^\top \Gamma^{-1}y - m_0 \right)
        \end{align}
        where $M(t) = (I + rtC_0A^\top \Gamma^{-1}A)^{-1}$. In particular, if $s=0$, then        
        \begin{align}
            m(t) &= M(t) m_0 + \left[I - M(t) \right] (A^\top \Gamma^{-1} A)^{+}A^\top \Gamma^{-1}y
        \end{align}
       
    \item \label{item:variational} \textbf{Asymptotic ($t\to\infty$) behavior of mean.}  
    \begin{enumerate}
        \item     If $s < 1$ and $C_0$ is positive definite, then $\lim_{t\to\infty} m(t) = \Pi_{\ke A}^{C_0}(m_0) + (\Gamma^{-1/2}A)^+ \Gamma^{-1/2}y$.  Furthermore, the limit can be characterised as  $\lim_{t\to\infty} m(t) = \mathrm{argmin}_{u}\{\|u-m_0\|_{C_0}^2: u \in \arg \min_v \|Av -y\|_\Gamma^2\}$. 

        In particular, if. $A^\top\Gamma^{-1}A$ is invertible, then there is a unique minimiser $m_\infty = \arg \min_v \|Av -y\|_\Gamma^2$, and $\lim_{t\to\infty} m(t) = m_\infty$.
        \item If $s = 1$, then $m(t)=m_0$ for all $t\geq 0$. 
    \end{enumerate}
    \end{enumerate}
\end{lem}

Before we prove this, we first give a quick interpretation of the statements of this lemma. The replicator equation models a population $p_t$ of traits under the influence of replication mediated by a fitness function which promotes traits $x$ such that $Ax \approx y$ in the sense that $\|Ax-y\|_\Gamma^2$ is small. This is exactly the setting of a Bayesian linear inverse problem \cite{stuart2010inverse} of type 
\begin{equation}\label{eq:BIP}
    y = Ax + \eps,
\end{equation}
where $\eps \sim \mathcal N(0,\Gamma)$ is Gaussian additive noise, and we assume a Gaussian prior $x\sim p_0$. The artificial time $t$ acts like a homotopy between prior $p_0$, is approximately equal to the Bayesian posterior at $t=1$ (this can be made exact using for example with an additional covariance inflation term, see \cite{bergemann2010localization}) and concentrates around the maximum likelihood estimator for $t\to\infty$.

In the context of evolution we can see from \eqref{eq:mean_weighted} that the mean $m(t)$ is a combination of the initial mean $m_0$ and a term enforcing minimality of $\|Am-y\|_\Gamma^2$, with the trade-off being mediated by the matrix $M(t)$. For $t=0$, this weighted combination is equal to $m(t) = I\cdot m_0 + 0$, while the second statement of lemma \ref{lem:explicitsoln} shows that $m(t)$ for $t\to\infty$ has an enticing geometric interpretation as the maximum-likelihood estimator $\argmin \|Av - y\|_\Gamma^2$, with ties (due to possible non-injectivity of $A$) being broken by minimal $C_0$-weighted distance to the initial population's mean. See figure \ref{fig:minnormsol} for a visualisation of this behavior.

\begin{figure}
    \centering
    \includegraphics[width=0.5\linewidth]{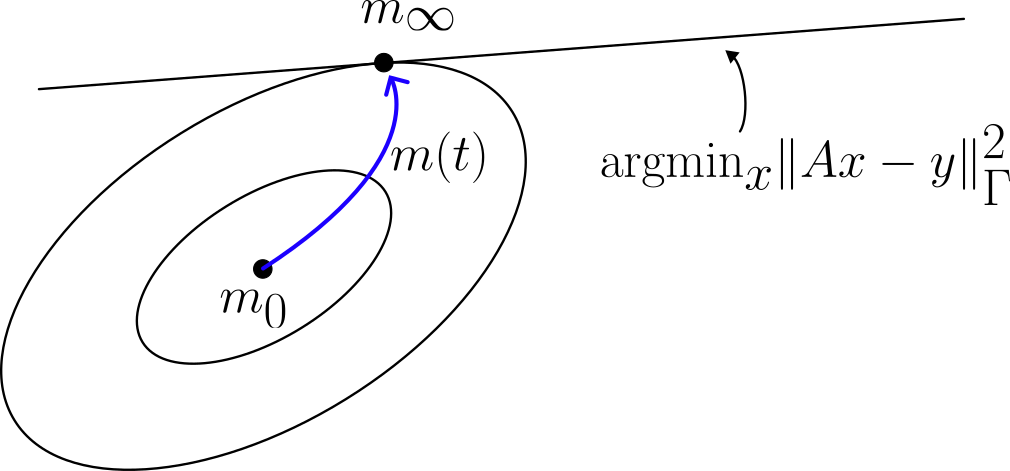}
    \caption{The trajectory and asymptotic state of the mean $m(t)$ follows the geometry of the initial distribution $p_0$ (ellipses visualise covariance structure $C_0$), and approaches the $\|\cdot - m_0\|_{C_0}$-minimum-norm solution of traits $x$ that minimise $\|Ax-y\|_\Gamma^2$. The shape of the trajectory is explained by the flying kite effect which is investigated in more detail later.}
    \label{fig:minnormsol}
\end{figure}

\begin{proof}
            Solving the ODE for $m$ is a bit elaborate. The idea will be to find a suitable matrix $L(t)$, analyse $\frac{\d}{\d t}[L(t)m(t)]$, and solve the resulting (easier) ODE directly.
        
        We begin by defining $M(t) = C(t)C_0^{-1} = (I+rtC_0A^\top \Gamma^{-1} A)^{-1}$ (suppressing the dependence on $t$ notationally from here on), and $L := M^\alpha = \exp(\alpha \log M)$. Using the results of lemma \ref{lem:properties_M}, $L$ is well-defined and we can compute 
        \begin{align*}
            \frac{\d}{\d t}[L\cdot m] &= \alpha \dot M M^{\alpha-1} m+ M^\alpha \dot m\\
            &= \alpha \dot C C_0^{-1}M^{\alpha-1}m - (1-s ) M^\alpha CA^\top \Gamma^{-1}(Am -y)\\
            &= -\alpha CA^\top \Gamma^{-1}A CC_0^{-1}M^{\alpha-1} m - (1-s ) M^\alpha CA^\top\Gamma^{-1}(Am -y)\\
            &= -\alpha M (C_0A^\top \Gamma^{-1}A) M M^{\alpha-1} m - (1-s) M^\alpha M C_0 A^\top \Gamma^{-1}(Am -y)\\
            &= -(\alpha + (1-s))M^{\alpha+1} C_0 A^\top\Gamma^{-1} A) m + (1-s) M^{\alpha + 1}C_0 A^\top\Gamma^{-1} y
        \end{align*}
        Now we choose $\alpha = -(1-s)$, which means that the first term vanishes, and $\alpha + 1 = s$. This means that
        \begin{align*}
            \frac{\d}{\d t}[M^{-(1-s)}\cdot m] = (1-s) M^{s} C_0 A^\top \Gamma^{-1} y
        \end{align*}
        We use the fact that $B^\top = B^\top B B^+$ for $B:=\Gamma^{-1/2}A$, hence
        \begin{align*}
            C_0 A^\top \Gamma^{-1} y &= C_0 (\Gamma^{-1/2}A)^\top \Gamma^{-1/2}y = C_0 (\Gamma^{-1/2}A)^\top  (\Gamma^{-1/2}A)  (\Gamma^{-1/2}A)^+ \Gamma^{-1/2}y \\
            &= (C_0 A^\top \Gamma^{-1} A) (\Gamma^{-1/2}A)^+ \Gamma^{-1/2}y, 
        \end{align*}
        i.e. 
        \begin{align*}
            \frac{\d}{\d t}[M^{-(1-s)}\cdot m] = (1-s) M^{s} (C_0 A^\top \Gamma^{-1} A) (\Gamma^{-1/2}A)^+ \Gamma^{-1/2}y
        \end{align*}
        Integrating this quantity and application of property 5. of lemma \ref{lem:properties_M} yields 
        \begin{align}
            M(t)^{-(1-s)}m(t) - m(0) &= (1-s)\int_0^t M^{s}(u) \d u \cdot  (C_0 A^\top \Gamma^{-1} A) (\Gamma^{-1/2}A)^+ \Gamma^{-1/2}y\\
            &= - (1-s)  (s - 1)^{-1} (M(t)^{s-1} - I)(\Gamma^{-1/2}A)^+ \Gamma^{-1/2}y,
        \end{align}
        i.e., $m(t) =  M(t)^{(1-s)}m(0) + (I - M(t)^{(1-s)})(\Gamma^{-1/2}A)^+ \Gamma^{-1/2}y$ where we can alternatively write
        \begin{align*}
            (\Gamma^{-1/2}A)^+ \Gamma^{-1/2}y = (A^\top \Gamma^{-1}A)^+ A^\top \Gamma^{-1}y
        \end{align*}
        using the Hermitian form of the pseudoinverse $B^+ = (B^\top B)^+ B^\top$ for $B:= \Gamma^{-1/2}A$.
                By the characterisation of the asymptotic covariance in lemma \ref{lem:cov_ccrep}, for $s < 1$,
        \begin{align*}
            \lim_{t\to\infty} m(t) &= \Pi_{\ke A}^{C_0}(m_0) + (I -  \Pi_{\ke A}^{C_0}) (\Gamma^{-1/2}A)^+ \Gamma^{-1/2}y =  \Pi_{\ke A}^{C_0}(m_0) + (\Gamma^{-1/2}A)^+ \Gamma^{-1/2}y,
        \end{align*}
        where the second term is simplified by the fact that $(I -  \Pi_{\ke A}^{C_0})$ projects onto (recall that $\Gamma$ is non-degenerate) $\ran (A^\top) = \ran ((\Gamma^{-1/2}A)^\top) = \ran ((\Gamma^{-1/2}A)^+)$
        
        We first prove this for the special case $s =0$, using the characterisation we just proved: Since $m(t)$ is the solution of \eqref{eq:repl_tikhonov} for all $t>0$, we know that $\lim_{t\to\infty} m(t)$ is equivalently given by 
        \[ \arg\min_{u\in V_y} \frac{1}{2}\|u-m_0\|_{C_0}^2,\]
        where $V_y = \{\text{minimisers of } \frac{1}{2}\|Au - y\|_\Gamma^2\}$. This proves property \ref{item:variational}. for $s =0$. But since the limit $\lim_{t\to\infty} m(t)$ is independent of the choice of $s $ (as long as $s < 1$), this holds more generally, too. 
        
\end{proof}

The following lemma provides an explicit formula for moment evolution in the case where $s=0$, but $y(t)$ is allowed to depend on $t$. Possible interesting examples for this behavior include an optimal feature moving with constant speed, in an oscillatory motion, or fluctuating randomly, see \cite{burger1995evolution,burger1997adaptation,gomulkiewicz2009demographic,chevin2013genetic,kopp_rapid_2014,matuszewski2015catch}, and also section \ref{sec:examples}.
\begin{lem}[Mean evolution of the replicator equation for non-constant $y(t)$ but $s=0$]
\label{lem:explicitsoln_nonconst}
    We consider the moment equations and utility function in \eqref{eq:moment_ODEs}. Let $s=0$. Then the solution of \eqref{eq:moment_ODEs} is given by   
        \begin{align}
            m(t) &= M(t) m_0 + \left[I - M(t) \right] (A^\top \Gamma^{-1} A)^{+}A^\top \Gamma^{-1} \frac{1}{t}\int_0^t y(u)\d u
        \end{align}
        where $M(t) = (I + rtC_0A^\top \Gamma^{-1}A)^{-1}$.
\end{lem}

\begin{proof}
We could follow the same proof technique as in lemma \ref{lem:explicitsoln}, item \ref{item:solution}, but since $s=0$, we can more straightforwardly compute
\begin{align*}
    \frac{\d }{\d t} C^{-1}m&= - C^{-1}\dot C C^{-1}m + C^{-1}\dot m = C^{-1}CA^\top \Gamma^{-1}A CC^{-1}m - C^{-1}CA^\top\Gamma^{-1}(Am-y(t))\\
    &= A^\top \Gamma^{-1}y(t),
\end{align*}
so $C^{-1}(t) m(t) - C_0^{-1}m_0 =  A^\top \Gamma^{-1}\int_0^t y(u)\d u$, i.e., using the identity $B^\top  = B^\top BB^+$ for $B = \Gamma^{-1/2}A$ 
\begin{align*}
    m(t) &= C(t)C_0^{-1}m_0+ rtC(t)A^\top\Gamma^{-1} \frac{1}{t}\int_0^t y(u)\d u \\
    &= C(t)C_0^{-1}m_0 + C(t)C_0^{-1} \cdot (rtC_0A^\top\Gamma^{-1}A) (\Gamma^{-1/2}A)^+ \Gamma^{-1/2}\frac{1}{t}\int_0^t y(u)\d u\\
    &= M(t) m_0 + M(t) (M(t)^{-1} - I) (\Gamma^{-1/2}A)^+ \Gamma^{-1/2}\frac{1}{t}\int_0^t y(u)\d u\\
    &= M(t) m_0 + (I - M(t))(A^\top \Gamma^{-1}A)^+ A^\top \Gamma^{-1}\frac{1}{t}\int_0^t y(u)\d u
\end{align*}
                
\end{proof}

\begin{rem}
    It is possible to also write down a similar formula for the solution of $m(t)$ for $s>0$ and $y$ depending on time, but the integral in question is a bit unwieldy and can seemingly not be simplified in general.
\end{rem}

\begin{rem} \label{rem:examples}
    In the case of a moving optimum, the long-time average $ \frac{1}{t}\int_0^t y(u)\d u$ becomes a quantity of interest. If we assume $A$ to be invertible (which is an interesting special case, but a rather strong assumption), then $m(t)\simeq \frac{1}{t}\int_0^t A^{-1}y(u)\d u$. Here are a few examples of relevant examples.
    \begin{enumerate}
        \item $y(t) = y_0$ (constant, but initial population starts away from this optimum). Here, $ \frac{1}{t}\int_0^t y(u)\d u = y_0$ and the population asymptotically concentrates around the optimal trait $A^{-1}y_0$.
        \item $y(t) = y_0 + v\cdot t$ (moving with constant velocity). Then $ \frac{1}{t}\int_0^t y(u)\d u = y_0 + \tfrac{1}{2}v \cdot t$ (i.e., following the same trajectory, but with half the velocity of the moving optimum). This means that the population does not catch up with the moving optimum, it actually lags behind more and more, with the distance growing with half the moving optimum's velocity.
        \item $y(t) = \sin(kt)$ (oscillating/periodic movement). Then, $ \frac{1}{t}\int_0^t y(u)\d u = \frac{1-\cos(kt)}{kt} \to 0$ for $t\to \infty$, which means that the mean converges to the center of the oscillating dynamics.
        \item $y(t) = W_t$ (Brownian motion, corresponding to fluctuating optimum without drift). Then, asymptotically, $ \frac{1}{t}\int_0^t y(u)\d u  \sim \mathcal N(0,t^{-1})$, which converges to $0$ in probability.
    \end{enumerate}
\end{rem}

\subsection{Connection to regularised optimisation}\label{sec:regularised_opt}

Ignoring covariance $C$ and population size $P$ for now, we can derive the following connection to regularised optimisation: The mean of the population driven by a replicator-mutator equation follows the path of a Tikhonov-regularised solution to the inverse problem $y(t) = Ax + \eps$, with a regularisation parameter that vanishes in the long time limit $t\to \infty$. This connection is well-known in the inversion community, see \cite{stuart2010inverse}, but seems to not have been pointed out explicitly in the context of evolution.

\begin{lem}[The mean evolution of the replicator equation as a Tikhonov-regularised estimator]
\label{lem:tikhonov}
    We consider the moment equations and utility function in \eqref{eq:moment_ODEs}. Then for a given time $t > 0$, the mean $m(t)$ for $r = 1, s = 0$ is a solution of the following generalised Tikhonov-regularised least squares problem 
    \begin{align}\label{eq:repl_tikhonov}
        {\arg \min}_x \left\|Ax - \tfrac{1}{t}\int_0^ty(u)\d u\right\|_\Gamma^2 + \frac{1}{t} \|x - x_0\|_{C_0^{-1}}^2.
    \end{align}
    In particular, if $y(t) = y$ is constant, then $m(t)$ solves
    \begin{align}\label{eq:repl_tikhonov_constanty}
        {\arg \min}_x \left\|Ax - y\right\|_\Gamma^2 + \frac{1}{t} \|x - x_0\|_{C_0^{-1}}^2.
    \end{align}

  \begin{proof}
  The generalised Tikhonov regularisation problem given by 
\begin{align*}
    x^\ast = {\arg \min}_x \|Ax - y\|_P^2 + \|x - x_0\|_Q^2 
\end{align*}
where $\|x - x_0\|_Q^2 = (x - x_0)^\top Q(x - x_0)$, is known to have solution 
\begin{align*}
    x^\ast = (A^\top HA + Q)^{-1}(A^\top Hy + Qx_0)
\end{align*}
Recall the explicit solution for the mean in replicator dynamics problem for the setting $r = 1, s = 0$,
\begin{align*}
    m(t) & = \left( I -  (I + tC_0A^\top  \Gamma^{-1}A)^{-1}  \right) (A^\top \Gamma^{-1}A)^+ A^\top \Gamma^{-1}y + (I + tC_0A^\top  \Gamma^{-1}A)^{-1}m_0 \\
    & = \left( I -  I + tC_0(I + tA^\top  \Gamma^{-1}AC_0)^{-1}A^\top  \Gamma^{-1}A  \right)  (A^\top \Gamma^{-1}A)^+ A^\top \Gamma^{-1}y + (I + tC_0A^\top \Gamma^{-1}A)^{-1}m_0\\
    & =  tC_0(I + tA^\top \Gamma^{-1}AC_0)^{-1}A^\top \Gamma^{-1}y + (I + tC_0A^\top \Gamma^{-1}A)^{-1}m_0 \\
    \label{eq:msimp}
    & = (I + tC_0A^\top \Gamma^{-1}A)^{-1} tC_0A^\top \Gamma^{-1}y + (I + tC_0A^\top \Gamma^{-1}A)^{-1}m_0\\
    &= \left(tC_0A^\top \Gamma^{-1}A + I \right)^{-1} \left(tC_0A^\top \Gamma^{-1}y + m_0\right)\\
    &=\left(A^\top \Gamma^{-1}A + \frac{1}{t} C_0^{-1}\right)^{-1}\left(A^\top \Gamma^{-1}y + \frac{1}{t}C_0^{-1}m_0\right)
\end{align*}
using the Woodbury identity in the second line and also the identities $(B^\top B)^+B^\top = B^+$ and $B^\top  = B^\top BB^+$. This is equivalent to the characterisation of $x^\star$ above for $H = \Gamma$, $Q = \frac{1}{t}C_0^{-1}$ and $x_0=m_0$. \qedhere
\end{proof}
\end{lem}

\subsection{Replication and mutation}\label{sec:repl_mut_solutions}
Now we consider the replicator-mutator equation, but we still assume that $G=0$, which corresponds to an assumption that there is a pure diffusion process driving the dynamics of the optimal feature $y(t)$.

Biologically, this corresponds to the same situation as in the preceding section, but now we have additional mutation on the trait space via a Brownian motion scaled by a covariance matrix $\Sigma$. We will see that this allows the population to track a moving optimum $y(t)$ (which was not possible for a population purely under the influence of selection, see remark \ref{rem:examples}). This mutation setup corresponds to ``blind'' or symmetric Gaussian mutation, since the drift term $G$ in \eqref{eq:repmutlin} is set to $0$.

Before we derive explicit solutions of the replicator-mutator equation we point out a connection to misspecified model filtering. Heuristically, Bayesian filtering \cite{majda2012filtering} is the task of recovering an unknown signal $x(t)$ assumed to follow a dynamical system like $\dot x(t) = Gx(t) + \Sigma^{1/2} \dot W(t)$, where $W$ is white noise, from an observation $y(t) = Ax(t) + \eps_t$, with $\eps_t\sim \mathcal N(0,\Gamma)$. This is done by modelling a filtering distribution in the space where $x(t)$ lives, by simultaneously evolving candidate reconstructions of $x$ forward in time via the signal's model dynamics, and incorporating incoming data $y(t)$ in an online way.

In some sense, evolution works in a similar way: There is a distribution of traits which is being influenced by ``data'' $y(t)$ (the optimal feature) exerting selection pressure via the fitness functional (just like the influence of the data via the likelihood function in Bayesian filtering), and simultaneous mutation (as described by the mutation noise model). Evolutionary dynamics is  performing filtering, implicitly (via its underlying mathematical model) assuming that the optimal feature is given by $y(t) = Ax(t) + \eps_t$, i.e., that there is an ephemeral hidden ``optimal trait'' $x(t)$ that would be optimally adapted to the current state of the optimal feature $y(t)$.  This connection has been discussed by several authors, e.g. \cite{Shalizi2009,Czegel2022,harper2009replicator,pathirajawacker24}.  Since individual traits are perturbed by a specific mutation model, this would mean that (again, implicitly) evolution assumes that the ``optimal trait'' $x(t)$ is also accurately modelled by a similar mutation model. This is of course not a biologically meaningful model: There is no ``optimal giraffe'' that all living giraffes are imitating in their pursuit of reproductional fitness, but fitness is rather a result of external (and possibly erratically changing) circumstances.

A particularly challenging situation is \textit{misspecified model filtering} \cite{calvet2015robust,teichner2023discrete}, where the dynamics used in the filtering algorithm is not the exact dynamics with which the true signal $x(t)$ is afflicted. This often arises due to incorrect model assumptions on the physical system governing $x(t)$. 

Biological evolution can be interpreted as a misspecified model filtering problem: There is no ``true'' optimal, time-dependent trait $x(t)$ from which the optimal feature $y(t)$ is generated, but $y(t)$ reflects external environmental factors like seasonality, climate change, influx of competing species, or other externalities. On the other hand, mutation follows a well-defined stochastic model, but this will in most cases not be compatible with the dynamics of $y(t)$. This means that the evolutionary ``filtering algorithm'' uses an internal, misspecified model.

The main issue of the following lemma will be finding an explicit solution to a certain Riccati-type equation. This is a research field with a long history, but we haven't been able to find the exact statement we need here, which is why we prove it here rather than combining various piecemeal statements to be found in papers spanning several decades, although we try to give references where possible.  To the best of our knowledge, we provide a more general closed-form solution to this specific Riccati-type equation than can be found elsewhere. In addition, while the fact that the asymptotic solution of Riccati equations have a geometric mean property, we believe that this has not been investigated in the context of the replicator-mutator equation, but it provides a beautiful geometric explanation for observable phenomena that are difficult to communicate otherwise. Concretely, the asymptotic covariance of the trait distribution consists of the geometric mean between the mutation covariance matrix (pushing the distribution away from the optimum) and the inverse ``penalising'' matrix $\Gamma$ (pulling it back in), balancing these two conflicting effects via their geometric mean.
\begin{lem}[Covariance evolution for the replicator-mutator equation]
\label{lem:repmuttime}
We consider the special case of the replicator-mutator equation with mutation kernel corresponding to pure diffusion

    \begin{align*}
    \frac{dq_t}{dt}(x) = \frac{1}{2} \nabla \cdot (\Sigma \cdot \nabla q_t) + q(t,x)\pi_{q_t}(x)
\end{align*}
which is \eqref{eq:repmutlin} for $G=0$. Let $p_0 = \mathcal N(m_0, C_0)$. In the following we set $Q := A^\top \Gamma^{-1} A$ for brevity. We assume that $\Sigma$ is symmetric and positive definite, and that $Q$ is symmetric and positive semi-definite. 
\begin{enumerate}

\item \textbf{Explicit solution for covariance.} The matrix $\Sigma Q$ is diagonalisable, and there is an invertible matrix $W$ such that $\Sigma Q W = W \Lambda^2$ for a diagonal matrix $\Lambda$ where $\Lambda = \mathrm{diag}(\mathbf 0, \Lambda_\bot)$ and $\Lambda_\bot$ contains all non-vanishing eigenvalues. The continuous-time Riccati equation for $C$ has solution $C(t) = D(t) E(t)^{-1}$ where
\begin{equation}
    \label{eq:explicit_t_general}
\begin{split}
D(t) &= W \begin{pmatrix}
            \mathbf I & \mathbf 0 \\ \mathbf 0 & \cosh(t\Lambda_\bot)
        \end{pmatrix} W^{-1}D(0) + W\begin{pmatrix}
            t\mathbf I & \mathbf 0 \\ \mathbf 0 & \Lambda_\bot^{-1}\sinh(t\Lambda_\bot)
        \end{pmatrix} W^{-1}\Sigma C(0)^{-1}D(0)\\
E(t) &=\Sigma^{-1}W\begin{pmatrix}
            \mathbf 0 & \mathbf 0 \\ \mathbf 0 & \sinh(t\Lambda_\bot)\Lambda_\bot
        \end{pmatrix}W^{-1} D(0) + \Sigma^{-1} W \begin{pmatrix}
            \mathbf I & \mathbf 0 \\ \mathbf 0 & \cosh(t\Lambda_\bot)
        \end{pmatrix} W^{-1}\Sigma C(0)^{-1}D(0)
\end{split}
\end{equation}
There are various degrees of freedom in choosing $D(0)$ and $E(0)$ such that $C(0) = D(0) E(0)^{-1}= C_0$ holds, the simplest being $D(0) = C_0$ and $E(0) = I$, but we won't need to specify this choice here.

\begin{itemize}
    \item In the special case where $Q = 0$, 
    \begin{align}\label{eq:explicit_t_Q0}
        C(t) &= C(0) + t\Sigma
    \end{align}
    \item In the special case where $Q$ is non-degenerate,
    \begin{equation}\label{eq:explicit_t_Qnondeg}
    \begin{split}
        D(t) &= W \cosh({t\Lambda}) W^{-1}D(0) + W \Lambda^{-1}\sinh({t\Lambda})W^{-1} \Sigma C(0)^{-1} D(0)\\
    E(t) &= \Sigma^{-1}W\sinh({t\Lambda})\Lambda W^{-1} D(0) + \Sigma^{-1}W\cosh({t\Lambda})W^{-1} \Sigma C(0)^{-1} D(0).
    \end{split}
    \end{equation}
Alternatively, we can write 
\begin{equation}\label{eq:explicit_t_Qnondeg2}
    \begin{split}
D(t) &= C_\infty \cdot \left[\exp(t(Q\Sigma)^{1/2}) \cdot \tfrac{1}{2}(C_\infty^{-1}C_0 - I) -  \exp(-t(Q\Sigma)^{1/2}) \cdot \tfrac{1}{2}(C_\infty^{-1}C_0 + I)\right] \\
E(t) &=\exp(t(Q\Sigma)^{1/2}) \cdot \tfrac{1}{2}(C_\infty^{-1}C_0 - I) +  \exp(-t(Q\Sigma)^{1/2}) \cdot \tfrac{1}{2}(C_\infty^{-1}C_0 + I),
    \end{split}
    \end{equation}
for the steady-state solution $C_\infty = \lim_{t\to\infty} C(t)$, the properties of which are described in the next item.
\end{itemize}
\item \textbf{Explicit solution for asymptotic covariance.} 
\begin{itemize}
    \item In the special case where $Q = 0$, $C(\infty)^{-1} = \mathbf 0$
    \item If $Q$ is non-degenerate, then the asymptotic state for $C(t)$ exists and is given by $$C_\infty = \lim_{t\to \infty} C(t) = W\Lambda^{-1} W^{-1}\Sigma = W\Lambda W^{-1} Q^{-1}.$$ 
    
    The asymptotic covariance $C_\infty$ can be shown to coincide with the geometric mean of $\Sigma$ and $Q^{-1}$ in the sense of \cite{ando2004geometric,moakher2005differential,nakamura2007geometric}, i.e. 
    $$C_\infty = G(Q^{-1},\Sigma) = Q^{-1}(Q \Sigma)^{1/2}.$$ In the same way, $C_\infty^{-1} = G(Q, \Sigma^{-1})$. 
    
    Also, $\Sigma C_\infty^{-1} = C_\infty Q = (\Sigma Q)^{1/2}$ and $Q C_\infty = C_\infty^{-1} \Sigma = (Q\Sigma)^{1/2}$. 
\end{itemize}
\end{enumerate}

\begin{rem}
Note that $Q^{-1}$ is the covariance of the best linear unbiased estimator for the minimiser $u$ of the linear inverse problem $\min_u \|A u -y\|_{\Gamma}^2$. This means that the asymptotic matrix $C_\infty$ is the geometric mean between this covariance and the evolution covariance $\Sigma$.
\end{rem}

\begin{rem}
    An explicit geometrical expression for the case of degenerate $Q$, but $Q$ not being identical to $0$, would be of interest.  We leave this for future consideration.
\end{rem}

\end{lem}
\begin{rem}
    It is likely possible to derive a closed-form solution for the covariance matrix for $Q$ being symmetric and positive semi-definite (rather than positive definite), as well as possibly degenerate mutation covariance $\Sigma$, but the technicalities involved with non-commutativity of $C_0, \Sigma,$ and $\Gamma$ and the necessary decomposition of trait space into degenerate and orthogonal components of $\Sigma Q$ make this issue too non-trivial for the scope of this manuscript and we defer this to future work.
\end{rem}
Having found an explicit solution for the covariance, we now give some details on the evolution of the mean of the population under the influence of a replicator-mutator equation. For mathematical simplicity, we will assume that $C_0 = C_\infty$, which is to be understood to mean that the population has already had some time to adapt such that the covariance of the population has stabilised at the asymptotic covariance. In practice (and in numerical simulations) we see that the statements we prove below hold in the more general case where $C_0 \neq C_\infty$ (with the same asymptotics appearing after an initial ``burn-in time''). This means that the formulae for $m$ should be understood to in principle (or ``in spirit'') hold for arbitrary $C_0$.

\begin{lem}[Mean evolution for the replicator-mutator equation]
\label{lem:repmuttime_mean}
We consider the setting of lemma \ref{lem:repmuttime}  Let $q_0 = P_0\cdot \mathcal N(m_0, C_0)$. In the following we set $Q := A^\top \Gamma^{-1} A$ for brevity. We assume that $\Sigma$ is symmetric and positive definite, and that $Q$ is symmetric and positive semi-definite.  Set $\alpha := 1-s$.
If $C_0 = C_\infty$, then
\begin{align*}
    \dot m &= -  \alpha(\Sigma Q)^{1/2}\left(m - (\Gamma^{-1/2}A)^+\Gamma^{-1/2}y(t)\right)
\end{align*}
and thus, abbreviating $A^- = (\Gamma^{-1/2}A)^+\Gamma^{-1/2}$,
\begin{align} \label{eq:steadystate_m_2}
    m(t) &= e^{-\alpha(\Sigma Q)^{1/2}t}\left[\alpha(\Sigma Q)^{1/2}\int_0^t e^{\alpha(\Sigma Q)^{1/2}u} A^- y(u)\d u  + m(0)\right]
\end{align}
and
\begin{align}\label{eq:steadystate_m}
    m(t) - A^- y(t) = e^{-\alpha(\Sigma Q)^{1/2}t}\left[\int_0^t e^{\alpha(\Sigma Q)^{1/2}u}A^{-} \frac{\d y(u)}{\d u}\d u + \left(m(0) - A^- y(0)\right)\right].
\end{align}
\begin{proof}

We set $(\Gamma^{-1/2}A)^+\Gamma^{-1/2} =: A^-$ (and it indeed behaves like a pseudo-inverse). Assuming that $C_0 = C_\infty$ allows us to use the fact that the steady state of the covariance equation is reached, i.e.,
\begin{align*}
    \dot m &= -(1-s)C_\infty A^\top \Gamma^{-1}Am + (1-s)C_\infty A^\top\Gamma^{-1}y \\
     &= - (1-s)C_\infty A^\top \Gamma^{-1}AC_\infty  \cdot C_\infty ^{-1}m + (1-s)C_\infty (A^\top \Gamma^{-1}A)(\Gamma^{-1/2}A)^+\Gamma^{-1/2}y\\
     &= - \alpha C_\infty  QC_\infty  \cdot \left(C_\infty ^{-1}m - C_\infty ^{-1}(\Gamma^{-1/2}A)^+\Gamma^{-1/2}y\right)\\
     &= - \alpha\Sigma C_\infty ^{-1}\left(m - (\Gamma^{-1/2}A)^+\Gamma^{-1/2}y\right) \\
     &= -  \alpha(\Sigma Q)^{1/2} \left(m - (\Gamma^{-1/2}A)^+\Gamma^{-1/2}y\right)
\end{align*}
A standard variation of constants argument shows characterisation \eqref{eq:steadystate_m_2}. Defining $d(t) := m(t) - {A^-}y(t)$, the above equation is equivalent to
\begin{align*}
    \dot d(t) &= -\alpha(\Sigma Q)^{1/2} d(t) - A^-\dot y(t)
\end{align*}
from which \eqref{eq:steadystate_m} follows by again by a standard variation of constants computation.
\end{proof}
\end{lem}

\begin{rem} \label{rem:locallyweighted}
\begin{itemize}
    \item Note that \eqref{eq:steadystate_m_2} and \eqref{eq:steadystate_m} indeed have $y(u)$ and $\frac{\d y(u)}{\d u}$ in the integrand, respectively.
    \item We can re-interpret \eqref{eq:steadystate_m_2} as a convolution
    \begin{align*}
        m(t)&= K_t\star (A^- y) + e^{-\alpha(\Sigma Q)^{1/2}t}m(0)
    \end{align*}
    with the asymmetric exponential kernel $K(u) = \alpha(\Sigma Q)^{1/2}\exp(-\alpha(\Sigma Q)^{1/2}u)\chi_{[0,t]}(u)$. This means that $m(t)$ performs a smoothing of the data $y(t)$, weighting the most recent $y(s)$ for $s\approx t$ more strongly than more historical data. For finite $t$, $K_t$ is not normalised since $\int_0^t K(u) < 1$, but we can rewrite
    \begin{align*}
        m(t)&= (I - e^{-\alpha(\Sigma Q)^{1/2}t}) \tilde K_t\star (A^- y) + e^{-\alpha(\Sigma Q)^{1/2}t}m(0)
    \end{align*}
    for a re-weighted kernel $\tilde K_t = (I - e^{-\alpha(\Sigma Q)^{1/2}t})^{-1}K_t$, which shows that $m(t)$ performs a weighted interpolation between a convolution term taking into account data $y(t)$ (with more emphasis on most recent $y(u)$, $u\approx t$), and its initial value. The dependence on the initial value drops off exponentially in time, i.e., asymptotically, the convoluted data term prevails.
    \item Equation \eqref{eq:steadystate_m} shows that the residual $m(t) - A^-y(t)$ is governed by a feedback-type equation featuring a similar convolution term taking into account the trajectory of $y(t)$ in a Riemann-Stieltjes-type of integral term.
\end{itemize}

\end{rem}

\begin{cor}\label{cor:expl_special}
    We consider the setting of Lemma \ref{lem:repmuttime}. The following special cases are of particular interest.
    \begin{enumerate}
        \item If $Q = A^\top \Gamma^{-1}A$ and $\Sigma$ commute, then $Q$ and $\Sigma$ are jointly diagonalisable with $QW = W\Lambda_1$, $\Sigma W = W\Lambda_2$ where $\Lambda^2 = \Lambda_1\Lambda_2$. Also, $C_\infty = W \diag(\sqrt{\lambda_1/\lambda_2}) W^{-1}$. This illustrates the more generally valid fact that $C_\infty$ is the geometric mean of $Q^{-1}$ and $\Sigma$.
        \item If $\Sigma = \sigma^2 I$, and in particular the preceding item holds, $QW = W\Lambda_1$ and $\Sigma W = W \sigma^2$. Furthermore, $C_\infty = \sigma Q^{-1/2}$.
        \item In the one-dimensional case, setting $A = 1$ for simplicity, with $\Gamma = \gamma^2$, $Q =  \frac{1}{\gamma^2}$, and $\Sigma = \sigma^2$, we have $C_\infty = \sigma\gamma$ and in the asymptotic limit, 
        \begin{equation}\label{eq:mean_1d_repmut}
            m(t) = e^{-(1-s) \frac{\sigma}{\gamma}t} \left[(1-s) \frac{\sigma}{\gamma}\int_0^t e^{(1-s)\frac{\sigma}{\gamma}u} y(u)\d u + m(0)\right]
        \end{equation}
        and
        \begin{equation}\label{eq:mass_1d_repmut}
            \dot P(t) = - \frac{1-s}{\gamma^2}(m(t) - y(t))^2 - \frac{\sigma}{\gamma} + K
        \end{equation}
        \item If $\Sigma$ is non-degenerate and $y(t) = y$, i.e. $y$ is constant, then $\lim_{t\to\infty} \|Am(t)-y\|_\Gamma^2 =  \|(I - AA^-)y\|_\Gamma^2$, which can be geometrically interpreted as 
        \[ \|(I - AA^-)y\|_\Gamma^2 = \|(I - (\Gamma^{-1/2}A)^+(\Gamma^{-1/2}A)) \cdot \Gamma^{-1/2}y\|^2 = \|\Pi_{\ran(\Gamma^{-1/2}A)^\bot}  \Gamma^{-1/2}y\|^2,\]
        the norm of the orthogonal projection of $\Gamma^{-1/2}y$ onto the orthogonal complement of $\ran(\Gamma^{-1/2}A)$.

        In particular, if $y \in\ran(\Gamma^{-1/2}A)$, then
        \begin{equation}
            \lim_{t\to\infty}  \|Am(t)-y\|_\Gamma^2 = 0
        \end{equation}
        and, asymptotically, the total population size follows an exponential growth or decay with rate
        \begin{equation}
            \frac{\dot P(t)}{P(t)} =  -\|\Pi_{\ran(\Gamma^{-1/2}A)^\bot}  \Gamma^{-1/2}y\|^2  -\tr [(Q\Sigma)^{1/2}] + K
        \end{equation}

    \end{enumerate}
    \begin{proof}
        The first statements follow by simple evaluations, but the last statement is true because from \eqref{eq:steadystate_m},
        \begin{align*}
            Am(t) - y &= A(m(t) - A^-y) - (I - AA^-)y \\
            &= Ae^{-\alpha(\Sigma Q)^{1/2}t}\left[\int_0^t e^{\alpha(\Sigma Q)^{1/2}u}A^{-} \frac{\d y(u)}{\d u}\d u + \left(m(0) - A^- y(0)\right)\right] - (I - AA^-)y\\
            &= Ae^{-\alpha(\Sigma Q)^{1/2}t} \left(m(0) - A^- y(0)\right) - (I - AA^-)y,
        \end{align*}
        so, using the fact that $(I - AA^-)$ is a projector onto the $\Gamma$-orthogonal complementary subspace of $\ran(A)$,
        \begin{align*}
            \|Am(t)-y\|_\Gamma^2 &= \|\Gamma^{-1/2}Ae^{-\alpha(\Sigma Q)^{1/2}t} \left(m(0) - A^- y(0)\right)\|^2 + \|(I - AA^-)y\|_\Gamma^2 + 0\\
            &\xrightarrow{t\to\infty} \|(I - AA^-)y\|_\Gamma^2 ,
        \end{align*}
        since $\Sigma Q$ was assumed to be non-degenerate and $\Gamma^{-1/2}Ae^{-\alpha(\Sigma Q)^{1/2}t} v\xrightarrow{t\to\infty} 0$ for any $v$: We can decompose $v = w^\bot + \sum_i \beta_i w_i$ where $\Sigma Q w^\bot = 0$ and $\Sigma Q w_i = \lambda_i w_i$ for $\lambda_i$. The components with non-vanishing eigenvalues decay exponentially for $t\to\infty$ since $e^{-\alpha (\Sigma Q)^{1/2})t}w_i = e^{-\alpha \sqrt{\lambda_i}t}w_i\to 0$. The orthogonal component $ w^\bot$ is being removed, as well: Since $e^{-\alpha (\Sigma Q)^{1/2})t}w^\bot = e^{-\alpha \cdot 0 \cdot t}w^\bot = w^\bot$,
        \begin{align*}
            \Gamma^{-1/2}Ae^{-\alpha(\Sigma Q)^{1/2}t}w^\bot = \Gamma^{-1/2}A w^\bot = 0,
        \end{align*}
        because $Q = A^\top \Gamma^{-1}A$ and for any matrix $B$, $\ke(B^\top B) = \ke B$.
    \end{proof}
\end{cor}

\section{Interesting phenomena and numerical examples }\label{sec:examples}

The preceding theoretical statements allow us to investigate a series of interesting phenomena which have to a large degree been described either heuristically, derived from observations, on experimental data, or using in-silico particle-based simulations. By virtue of the fact that we have in some cases explicit solutions for the (asymptotic) behavior of mean, covariance, and total mass (i.e., population size) of the population over time, we can directly see the influence of phenomena like the flying kite effect, or the fixed lag property (described further down below) in the structure of the mathematical solution. We illustrate those effects with simulations.

While these effects are are only provably present in the mathematically idealised scenarios of Gaussian initial population, and quadratic fitness function, the qualitative behaviour of scenarios deviating from these restrictions can exhibit to some extent similar features. 

We start by recapitulating the mathematical model of evolution which we are going to be adopting from here on. The Euclidean space $X=\R^d$ models a space of traits (or, alternatively, species). This means, two elements $x,z\in X$ could model two individuals with a slightly differing genome, or different phenotype.

We further assume that there is an external reference ``feature'' $y(t)\in Y=\R^k$ possibly depending on time, which represents an optimal summary of traits in the sense that corresponding evolutionary fitness is maximised, and a coupling $A:X\to Y$ which translates any trait $x\in X$ into a feature $A(x)\in Y$. 

Next, a symmetric and positive definite penalty matrix $\Gamma \in \R^{k\times k}$ quantifies selection pressure: We define the misfit function
\[\Phi(x) = \frac{1}{2}\|A(x) - y(t)\|_\Gamma^2 ,\]
which describes the exact (quadratic) way in which deviating from the optimal feature (which would satisfy $Ax\approx y$) incurs selection pressure: If $\Gamma$ is ``small'' (in the sense of its eigenvalues), then even small deviations from the optimum means large reduction in fitness. 

This allows us to define a fitness function of $x$ via $\pi(x) = -\Phi(x) + K$. Alternatively, we can also define a fitness payoff function for an individual of trait $x$ impacted by the presence of an individual of trait $z$ via $f(x,z) = -\frac{s}{2} \|Ax-Az\|_\Gamma^2 - \frac{1-s}{2} \|Ax-y(t)\|_\Gamma^2- \frac{1-s}{2} \|Az-y(t)\|_\Gamma^2 + K$, which translates (as described in previous sections) into a fitness function $\pi_{q_t}(x) = \frac{\int f(x,z)q_t(z)\d z}{\int q(t,z)\d z}$ of trait $x$ in an environment shared with the full population $q_t$.

Let $p_0$ be an initial population of traits, modelled as an unnormalised probability measure, i.e. there is a $P_0 > 0$ such that $p_0/P_0$ is a probability measure. In all following examples we will assume that $p_0$ is an unnormalised multivariate Gaussian distribution on trait space.

Evolutionary dynamics acts on the population in the form of the unnormalised replicator-mutator equation with drift-neutral mutation (which is \eqref{eq:repmutlin} with $G=0$),
\begin{align}
    \label{eq:repmutlin_nodrift}
    \partial_t q(t,x) =  \frac{1}{2}  \nabla \cdot (\Sigma \nabla q(t,x))  +  q(t,x)\pi_{q_t}(x)  
\end{align}
The first term models the effect of random mutation, the second term applying selection pressure.

This setting fulfills the assumptions of lemmata \ref{lem:repmuttime} and \ref{lem:repmuttime_mean}, so we can explicitly and completely characterise both intermediate and asymptotic evolution of the unnormalised Gaussian via mean $m(t)$, covariance $C(t)$, and total mass $P(t)$.

\begin{example}\label{ex:giraffe}
    As an (admittedly very naive, but hopefully illustrative) example, we consider giraffes, assuming we are only modelling their neck length $x_1\in \R$ and their leg length $x_2\in \R$. Modelled in a joint trait space, each individual giraffe (with a specific neck length $x_1$ and leg length $x_2$) is represented by a point $x\in \R^2$. We will keep coming back to this example in order to illustrate all relevant mathematical objects.
    
    Continuing, let $A_\text{tot}: X\to \R = Y$ be the mapping $A_\text{tot}(x_1,x_2) = x_1+x_2$ which is a surrogate for total body height of the giraffe represented by neck length $x_1$ and leg length $x_2$. If $y_\text{tot}(t)\in\R$ is the optimal feature, then giraffes with any combination of neck and leg length such that $x_1+x_2\approx y_\text{tot}(t)$ would be considered ``fittest''. The penalty matrix $\Gamma_\text{tot} \in \R$ is then just a positive number $\Gamma= \gamma^2$, and fitness is described by the misfit function
    \[ \Phi_\text{tot}(x) = \frac{1}{2\gamma^2}(x_1 + x_2 - y(t))^2.\]
    Depending on the value of $\gamma$, deviation of total body length from the optimum $y(t)$ is then more or less strongly penalised.
    
    As a second example for a fitness function, let $A_\text{nl}: X\to \R^2 = Y$ be the identity mapping $A_\text{nl}(x_1,x_2) = (x_1,x_2)$. This means that $y_\text{nl}(t)\in\R^2$ now corresponds to an optimal neck and leg length pair $y_\text{nl}(t) = (y_1(t), y_2(t))$. If we choose 
    \[\Gamma_\text{nl} = \frac{1}{2}\begin{pmatrix}
        1+\eps & -1+\eps\\ -1+\eps & 1+\eps
    \end{pmatrix}\in \R^2\] 
    as a penalty matrix, then the misfit function is 
    \[\Phi_\text{nl}(x) = \frac{1}{2}\|y(t) - x\|_\Gamma^2.\]
    If $\eps >0$ is reasonably small, then the effect of this misfit function is that perturbation in the direction $(1,1)^\top$ (i.e., changing total body length) are punished relatively quickly while small perturbations in the direction $(1,-1)^\top$ (i.e., conserving total body length) are (only to some extent) tolerated. This means that the fitness-relevant feature of a giraffe is primarily adaptation of total body length to some optimal reference, and secondarily limited deviation from both reference neck and leg length individually.
    Of course this completely ignores the fact that neck length and leg length tend to be strongly positively correlated in reality, as evident from empirical data \cite{cavener2023giraffe}. Both misfit functions are visualised in the top row of figure \ref{fig:misfit_giraffe}, together with the initial population described next.
    
     We set $p_0 = \mathcal N(m_0,C_0)$ for $m_0 = (1,2)^\top$ and 
     \[C_0 = \begin{pmatrix}
         0.15 & 0.05\\ 0.05 & 0.05
     \end{pmatrix}.\]
    
    Figure \ref{fig:misfit_giraffe} shows the long-time behavior of the volution of the replicator-mutator equation in the two settings for the specific choice of an anisotropic mutation term $\Sigma = \diag(0.8, 0.05)$, i.e., with stronger mutation along $x_1$ (neck length) than $x_2$ (leg length). 
    
    In the first example, selection pressure only acts orthogonal to the dark shaded structure (which we'll call \textit{manifold of indifference}), and mutation freely diffuses the population along it. In the asymptotic limit does the population converge to a degenerate Gaussian with infinite variance along the manifold of indifference, but with equilibrium variance orthogonal to it. This equilibrium variance arises as the compromise between mutation driving the population away from the manifold, and selection pulling it in. Mathematically, the precision of the population converges to the precision of the degenerate Gaussian $C_\infty^{-1} = \Sigma^{-1}(\Sigma Q)^{1/2}$, which is the geometric mean of the mutation precision $\Sigma^{-1}$ and the matrix $Q = A^\top\Gamma^{-1}A$, as described by lemma \ref{lem:repmuttime}. 
    
    In the second example, the mean converges to the minimum of the misfit function, and the asymptotic covariance is given by the geometric mean $C_\infty = \Sigma(\Sigma^{-1}Q^{-1})^{1/2}$, which, because $A=I$, is in this case given by $C_\infty = \Sigma(\Sigma^{-1}\Gamma^{-1})^{1/2}$. In particular does the asymptotic covariance not align with the inverse penalty matrix $\Gamma^{-1}$, nor with the mutation covariance $\Sigma$, but rather (again) with the geometric mean between them.
\end{example}

\begin{figure}
    \centering
    \includegraphics[width=0.5\linewidth]{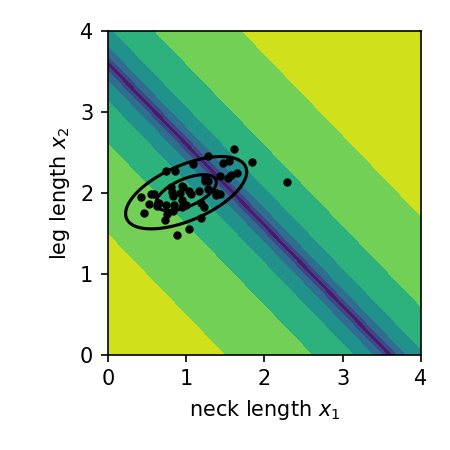}%
    \includegraphics[width=0.5\linewidth]{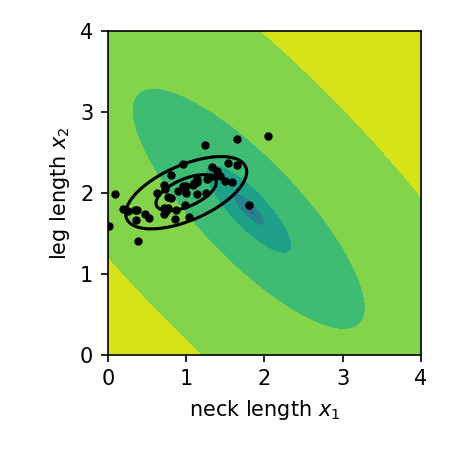}\\
    \includegraphics[width=0.5\linewidth]{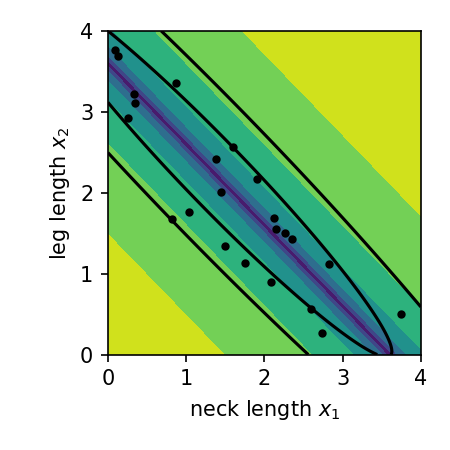}%
    \includegraphics[width=0.5\linewidth]{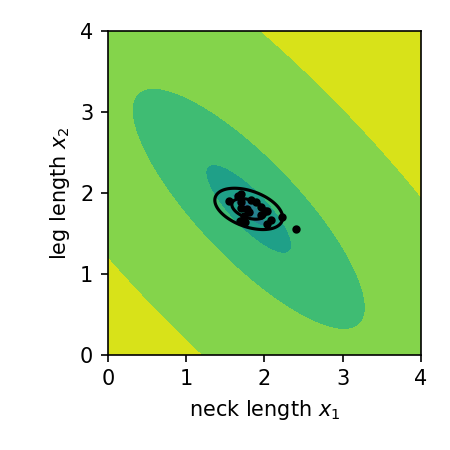}
    \caption{Misfit functions from example \ref{ex:giraffe}. Left column features a misfit $\Phi_\text{tot}$ that penalises deviation from absolute body length $x_1+x_2$ from an optimal value of $y(t) = 3.6m$. The right column considers the misfit $\Phi_\text{nl}$, an ideal neck-leg configuration $y(t) = (1.8m,1.8m)^\top$, with body-length conserving deviations (parallel to $(1,-1)$) less harshly penalised than body-length altering deviations (parallel to $(1,1)$). Top row: Both fitness landscapes are shown with a superimposed sample population (as an empirical measure) and contours of the initial population's density function. Bottom left: For $t\to\infty$, population mutates along direction of no selection pressure $(1,-1)^\top$, and keeps an equilibrium distance orthogonal to it $(1,1)^\top$ such that mutation and selection are balanced. Bottom right: Population concentrates around optimum, but equilibrates with a covariance $C_\infty$ given by the geometric mean between $Q^{-1}$ and  mutation covariance $\Sigma$.}
    \label{fig:misfit_giraffe}
\end{figure}

\subsection{The flying kite effect}\label{sec:flying}
The flying kite effect means that adaptation (in the sense of evolution of the the mean phenotype of a population) does not happen along the straightest line towards the point of optimal fitness, but is biased towards directions of larger genotype variation within the population. One of the first descriptions of this effect is \cite[eq. (8b)]{lande1979quantitative}, and described in detail as dynamics along ``genetic lines of least resistance'' in \cite{schluter1996adaptive}.  The term ``flying kite effect'' seems to have been coined by \cite{jones2004evolution} who demonstrated this behavior in the context of a moving optimum in the fitness landscape for a two-dimensional example, simulated with an agent-based model described in \cite{burger1994distribution,jones2003stability}. The computations were based on generations modelled discretely in time, which means that the mathematical model was similar in spirit to the discrete-continuous replicator-mutator equation. More recently, \cite{chevin2013genetic,kopp_rapid_2014} gave a comprehensive treatment on behavior of the discrete-continuous replicator-mutator equation in a variety of settings (in particular, for optima moving in time, changing abruptly, or fluctuating randomly with varying degrees of autocorrelation), with simulations acquired by ``adaptive walks'' as in \cite{kopp2018phenotypic}. 

To the best of our knowledge, the nature of the flying kite effect has not been studied in the time-continuous setting so far. For this reason, we show here how the replicator-mutator equation exhibits qualitatively similar behavior and how we can explicitly quantify behavior in the continuous-time setting. By choosing a suitably fine time-discretisation, these results will be more or less applicable to the discrete-time setting, too.

We give an example to illustrate the effect on a simple two-dimensional example: Let $A = \mathbf{I}\in \R^{2\times 2}$, and consider an initial population $p_0 = \mathcal N(m_0,C_0)$ with $m_0 = \mathbf{0}\in \R^{2}$ and $C_0=\begin{pmatrix}
    8.5&7.5\\7.5&8.5
\end{pmatrix}$ which corresponds to a multivariate Gaussian with strong correlation between the two traits, and is illustrated by the largest (blue) ellipse in figure \ref{fig:flying_kite}. If we set $y = Au_\text{opt} = (10,0)^\top$, which corresponds (due the form of $A$) to an optimal trait $u_\text{opt} = (10,0)^\top$, we can observe how the replicator-mutator equation \eqref{eq:ccre_norm} models the population's adaptation to this optimal trait. Figure \ref{fig:flying_kite} shows the evolution of this population in time, with means and covariances visualised at times $t\in\{0,0.1,1,3\}$ towards the optimum $u_\text{opt} = A^-y = (10,0)^\top$. It is apparent that the population (and in particular its mean, with its trajectory plotted as the curved solid black line) does not approach the optimum in a straight line, but is initially biased towards the population's line of largest variation. With the effect of selection reducing the variance along this line more quickly than in the orthogonal direction, the ellipse is being foreshortened, and the biasing deflection is mitigated over time, with the trajectory angling towards the optimum after some time. The same effect can be observed in higher dimensions, with the velocity of adaptation aligning with the directions of higher variance initially.

In mathematical terms, this is a simple effect of the mean evolution in \eqref{eq:moment_ODEs} being preconditioned by the population covariance: In the simple case of the example in this section, the mean evolves according to $\dot m(t) = -C(t)(m(t) - y)$. If (working with an exaggerated case for illustration) $C(t)$ is degenerate (i.e. contracted to a thin line) along all directions with exception of a main (normalised) direction of variance $v_{\max}$ with $\|v_{\max}\| = 1$, then the covariance matrix is given by $C(t) = \sigma(t)^2 v_{\max} \cdot v_{\max}^\top$ for a scalar variance $\sigma(t)^2$, which means that $\dot m(t) = -\sigma(t)^2 v_{\max} \langle v_{\max}, m(t)-y\rangle$, i.e. the velocity of adaptation $\dot m(t)$ is parallel to the direction of largest covariance $v_{\max}$. A more comprehensive argument in the more general case can easily be made using a singular value decomposition of $(t) = \sum_{i=1}^n \sigma_i(t)^2 v_i v_i^\top$ with $\sigma_1 \geq \sigma_2 \geq \cdots \sigma_n$ being the ordered variations along directions $v_i$, $i=1,\ldots,n$.

In the context of the Ensemble Kalman inversion, the same kind of behavior can be observed, see \cite[fig. 1]{blomker2022continuous} and \cite{bungert2023complete}, which is of course not surprising given the fact that the moment equations are very similar in these setting.
\begin{figure}
    \centering
    \includegraphics[width=0.7\textwidth]{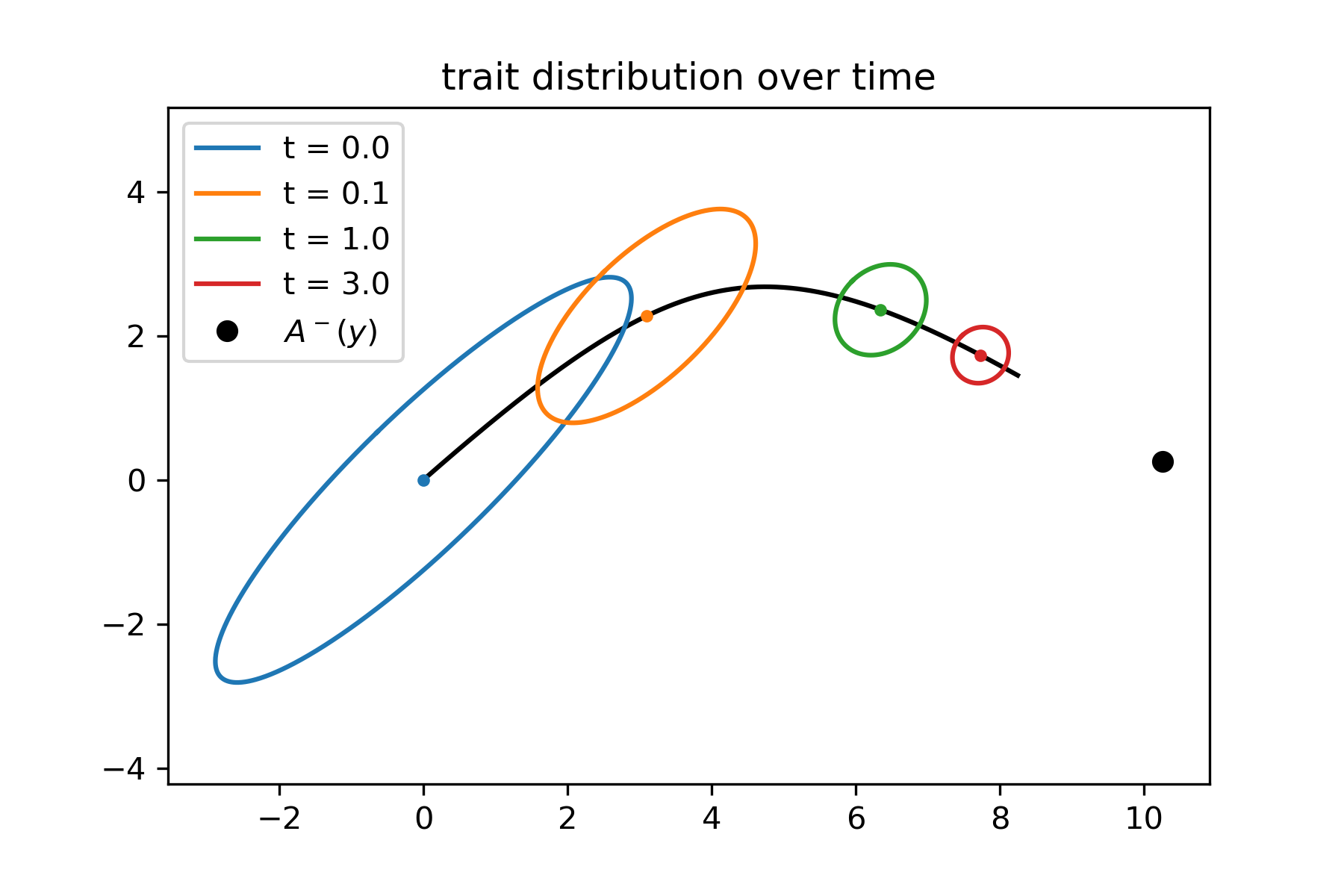}
    \caption{Flying kite effect: Mean population does not approach optimum $A^-y$ via a straight line, but initially starts off ``upwards'' along the direction of dominating variance. This flying kite effect is due to the mean dynamics in \eqref{eq:moment_ODEs} being preconditioned with the covariance matrix.}
    \label{fig:flying_kite}
\end{figure}

\subsection{Survival of the flattest}
Described in \cite{sardanyes2008simple}, it is now understood that for a given fixed strength of mutation, it is not always the fittest population that prevails. To illustrate this effect, we consider two populations. Population A is concentrated around the global fitness optimum, where the fitness peak decays sharply. Population B is concentrated around a local optimum in fitness, less than the global fitness, but with a less harsh dropoff in fitness away from this local optimum.

Evidence for this phenomenon was provided by in-silico experiments in \cite{sardanyes2008simple}, simulating individuals in a virtual population over time. Using the results we derived, we can demonstrate the same behavior by a simple algebraic condition: Rather than comparing two populations in a shared fitness landscape with two peaks, we can equivalently compare two ``universes'': Universe A features a population centered around a ``high and sharp peak'', while Universe B contains a population huddling around a ``shallow hill''. Since we ignore interactions between the populations (and this was also assumed in \cite{sardanyes2008simple}), comparing the relative growth rate of both populations in the different ``universes'' gives us the same result as having two population on a shared bimodal fitness landscape (but allows us to use quadratic fitness functions, and thus we can use the results obtained in section \ref{sec:explicit}).

Let us thus consider two populations\footnote{$q^{(i)}$ is an unnormalised Gaussian with mean $m^{(i)}$, covariance $C^{(i)}$ and total mass $P^{(i)}$.} $q^{(i)} = P^{(i)}\mathcal N(m^{(i)}, C^{(i)})$ under the influence of the replicator-mutator equation
\begin{align*}
    \label{eq:repmutlin_nodrift}
    \partial_t q(t,x) =  \frac{1}{2}  \nabla \cdot (\Sigma \nabla q(t,x))  +  q(t,x)\pi_{q_t}(x)  
\end{align*}
with (as in section \ref{sec:explicit})
\begin{equation*}
\begin{split}
    \pi_{q_t}(x)     &=  -\frac{1}{2}\|Ax-y(t)\|_\Gamma^2 -\frac{1}{2}\|Am-y(t)\|_\Gamma^2 + s \langle Ax-y(t), Am(t)-y(t)\rangle_\Gamma - \tr A^\top \Gamma^{-1} A C + K.
\end{split}
\end{equation*}
For simplicity we assume $s=0$ (fitness of individuals is independent of the ambient total population) and $y(t) = y_0 \in \ran(\Gamma^{-1/2}A)$ for all $t$, i.e., the fitness landscape is time-independent. We assume that mean and covariance have already equilibrated, i.e. $m^{(i)} = m_\infty$ and $C^{(i)} = C_\infty$, albeit $P^{(i)}(t)$ continues to follow its moment equation (from \eqref{eq:moment_ODEs}, and its special case in \ref{cor:expl_special})
\begin{equation*}\begin{split}
        \dot P(t) &= P(t)\cdot \left(K  -(1-s)\|Am(t)-y(t)\|_\Gamma^2   -\tr [A^\top \Gamma^{-1}AC(t)] \right)\\
        &=P(t)\cdot \left(K   -\tr( [Q \Sigma]^{1/2}) \right),
    \end{split}\end{equation*}
    where we used the fact that in equilibrium $\|Am_\infty - y\|_\Gamma^2 = 0$, which follows from corollary \ref{cor:expl_special}, the shorthand $Q = A^\top\Gamma^{-1}A$, and the fact that $QC_\infty = (Q\Sigma)^{1/2}$.

This means that the population size follows an exponential growth/decay of form
\begin{equation}
    P(t) = P(0)\exp\Big( t(K   -\tr( [Q \Sigma]^{1/2})\Big).
\end{equation}

Now we can clearly see the effect of the shape (narrow/wide, governed by $\Gamma$, which is part of $Q$) and its height (i.e. its reference optimal fitness $K$) on the exponential growth rate of the population size $P^{(i)}$: If the peak becomes more narrow, then $\Gamma$ gets smaller, which means that $\tr( [Q \Sigma]^{1/2}) $ increases. In order to keep the growth rate constant, $K$ (i.e. the height of the peak) needs to increase, as well.

This becomes even clearer if we consider the one-dimensional example, where $A = 1$, $\Gamma = \gamma^2$, and mutation covariance $\Sigma = \sigma^2$. Here, the exponential growth rate is given by $K - \frac{\sigma}{\gamma}$. Two populations $q^{(1)},q^{(2)}$ then have the same asymptotic growth rate, if their ambient fitness landscape's peak $K^{(i)}$ and width $\gamma^{(i)}$ are related via
\begin{equation}\label{eq:interplay}
    K^{(1)} = K^{(2)} + \sigma \left(\frac{1}{\gamma^{(1)}}-\frac{1}{\gamma^{(2)}} \right).
\end{equation}
In particular, even if $K^{(1)} \gg K^{(2)}$, i.e., population 1 sits in a fitness optimum much better than population 2, then this can be offset by $\gamma^{(2)} \gg \gamma^{(1)}$, i.e. the slope of the fitness neighborhood of population 2 being a lot more ``forgiving'' than the one of population 1. We can also directly quantify the effect of mutation in this setting: This offset is linearly mediated by the mutation size. In particular, if $\sigma = 0$ (i.e., there is no mutation), then the fitness peak's height $K$ is the only determining factor. Summarising: A population thrives in a fitness landscape if
\begin{itemize}
    \item the trait-independent (global) fitness $K$ is large, or
    \item the quadratic coefficient of the fitness landscape is large (more forgiving), or
    \item mutation is small enough (diffusion out of the fitness landscape is not too quickly),
\end{itemize}
or with a combination of these, the exact interplay described in \eqref{eq:interplay}. In particular we can see that it is possible to have a population go extinct in a fitness landscape with very high maximal fitness, if the combination of mutation rate (via $\Sigma$ / $\sigma$) and slope of the fitness landscape (via $\Gamma$ / $\gamma$) is damaging enough. Or, equivalently, considering two populations in parallel: In the presence of mutation it is possible that a population in a lower-level, but more shallow fitness landscape can grow more quickly than a population in a very high, but sharply peaked fitness landscape. For a demonstration of this effect, see figure \ref{fig:survival_flattest}. Note that illustrations of the survival of the flattest effects (like the one in \cite{sardanyes2008simple}) often use the time-discrete (Wrightian) fitness rather than time-continuous (Malthusian) fitness), so we show both the original Malthusian fitness and (an approximation of) the corresponding Wrightian fitness. \begin{figure}
    \centering
    \includegraphics[width=\linewidth]{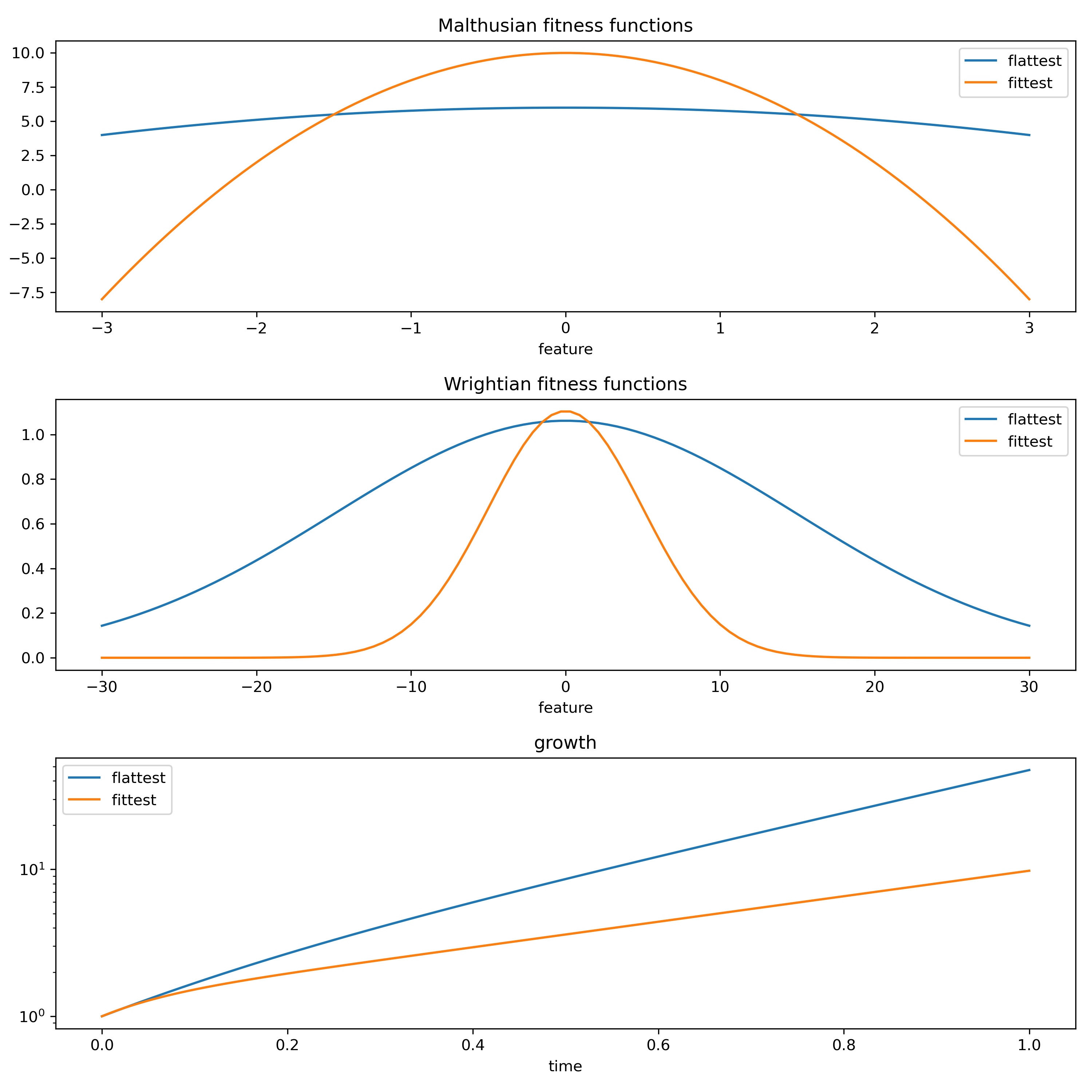}
    \caption{A population sitting in a fitness optimum with a lower maximal fitness (but a flatter fitness decay) can outgrow a population in a fitness optimum with higher fitness value if the latter is more sharply peaked. }
    \label{fig:survival_flattest}
\end{figure}

\subsection{Fixed optimum}
We now consider a population in a fitness landscape where the optimal trait is fixed, but possibly far away from the current population's mean trait. We can imagine that this arises from a sudden change in the environment (eg, the introduction of new species, or a sudden climate event in the environment). Because we derived exact solutions for the total population size, we can explore when the population can survive this situation (and when it goes extinct). We recall that the total population size is given by
\[\dot P(t) = P(t)\cdot \Big(K  -(1-s)\|Am(t)-y(t)\|_\Gamma^2   -\tr [A^\top \Gamma^{-1}AC(t)] \Big)\]
This means we can study extinction by considering the asymptotic exponential growth rate $ \Big(K  -(1-s)\|Am(t)-y(t)\|_\Gamma^2   -\tr [A^\top \Gamma^{-1}AC(t)] \Big)$. If this rate is positive, then the population eventually grows exponentially, otherwise it eventually decays at this rate. Note that this ignores the possibility of `extinction by exhaustion', i.e., if the adaptation to the fixed optimum takes too long, then the population size can intermittently drop below a threshold of minimal viability, from which it could not actually recover (0.07 giraffes clearly would not be able to sustain the population until adaptation is complete). 

\begin{lem}[Extinction of a population with fixed optimal trait]
    We consider the case where $y(t)=y$ is a constant optimal trait. Then, in the long-time behaviour, the population dies out (i.e., $P(t)\searrow 0$) if and only if $K < (1-s) \|\Pi_{(\im \Gamma^{-1/2}A)^\top} \Gamma^{-1/2}y\|^2 + \tr(Q\Sigma)^{1/2}$. In particular, if $y \in \im \Gamma^{-1/2}A$, then $P$ does not converge to $0$. 
    \begin{proof}
        Corollary \ref{cor:expl_special} shows that
        \[\|Am-y\|_\Gamma^2 \to \|\Pi_{(\im \Gamma^{-1/2}A)^\top} \Gamma^{-1/2}y\|^2\]
        and the trace term in the asymptotic limit is given by $QC = (Q\Sigma)^{1/2}$.
    \end{proof}
\end{lem}
This statement gives us a clear idea of how the overall fitness $K$, the `uniformity term' $s$, and the interplay between fitness penalty $\Gamma$ and mutation $\Sigma$ play a role in species survival.
\subsection{Optimum is moving with bounded speed, the fixed lag phenomenon}

The following lemma examines the asymptotic behavior of the replicator-mutator equation for a moving optimum $y(t) = y_0 + t\cdot v_y$. This is a classic situation analysed in the mathematical biology literature, see the references in section \ref{sec:flying}, but results are usually obtained either for a simplified one-dimensional case, or obtained by numerical simulations. Using the theoretical statements we have derived, we are able to shine a bit more light on the mathematical reasons for the behavior observed.

This lemma generalises known results for fixed lag behaviour to the multivariate case \cite{burger1995evolution,matuszewski2015catch}.

The main idea here is that a population adapting to an optimum moving constant in time (in feature space) can never truly catch up with this optimum, and tracks it with a fixed lag (which is not necessarily aligned with the direction of movement in feature space).
\begin{lem}\label{lem:fixed_lag}
    We consider the setting of lemma \ref{lem:repmuttime} and nondegenerate $Q=A^\top\Gamma^{-1}A$. If $\dot y = v_y = \text{const.}$, i.e. $y(t)$ is moving with constant velocity $v_y$, and if $C_0 = C_\infty$, then 
\begin{align*}
    m(t) &\simeq A^- y(t) + \underbrace{(\Sigma Q)^{-1/2} A^- v_y}_{\text{fixed lag}}
\end{align*}
The population asymptotically goes extinct if and only if $K < (1-s)\|\Gamma^{-1/2}(\Sigma Q)^{1/2}A^- v_y\|^2 + \tr[(Q\Sigma)^{1/2})].$
\begin{proof}
    This is a result of the last statement of lemma \ref{lem:repmuttime}, which we repeat here out of convenience:
    \begin{align*}
         m(t) - A^- y(t) = e^{-\alpha(\Sigma Q)^{1/2}t}\left[ \int_0^t e^{\alpha(\Sigma Q)^{1/2}s}A^{-} \dot y(s)\d s + \left(m(0) - A^- y(0)\right)\right].
    \end{align*}
    Since $\dot y(s) = v_y$ for all $s$, the integral is straightforward to compute and we obtain
    \begin{align*}
        m(t) - A^- y(t) &= e^{-\alpha(\Sigma Q)^{1/2}t}\left[ \alpha(\Sigma Q)^{-1/2}\left(e^{\alpha(\Sigma Q)^{1/2}t} - I\right)A^{-}  v_y + \left(m(0) - A^- y(0)\right)\right]\\
        &= \alpha(\Sigma Q)^{-1/2} \left( I - e^{-\alpha(\Sigma Q)^{1/2}t} \right) A^-  v_y + e^{-\alpha(\Sigma Q)^{1/2}t}\left(m(0) - A^- y(0)\right).
    \end{align*}
    since both $\Sigma$ and $Q$ are non-degenerate, the involved matrix exponentials converge to the $0$ matrix, and thus
    \begin{align*}
        m(t) - A^- y(t) &\simeq  (\Sigma Q)^{-1/2} A^-  v_y.
    \end{align*}
    By plugging the asymptotic state of the covariance and the fixed lag into the equation for $P$ in \eqref{eq:moment_ODEs}, we obtain
    \begin{align*}
        \dot P(t) &= P(t) \cdot\left(- (1-s)\|\Gamma^{-1/2}(\Sigma Q)^{1/2}A^- v_y\|^2 - \tr[(Q\Sigma)^{1/2})] + K\right) =: P(t)\cdot \chi
    \end{align*}
    which proves exponential growth or decay with rate $\chi$.
\end{proof}
\end{lem}

\begin{figure}
    \centering
    \includegraphics[width=0.35\textwidth]{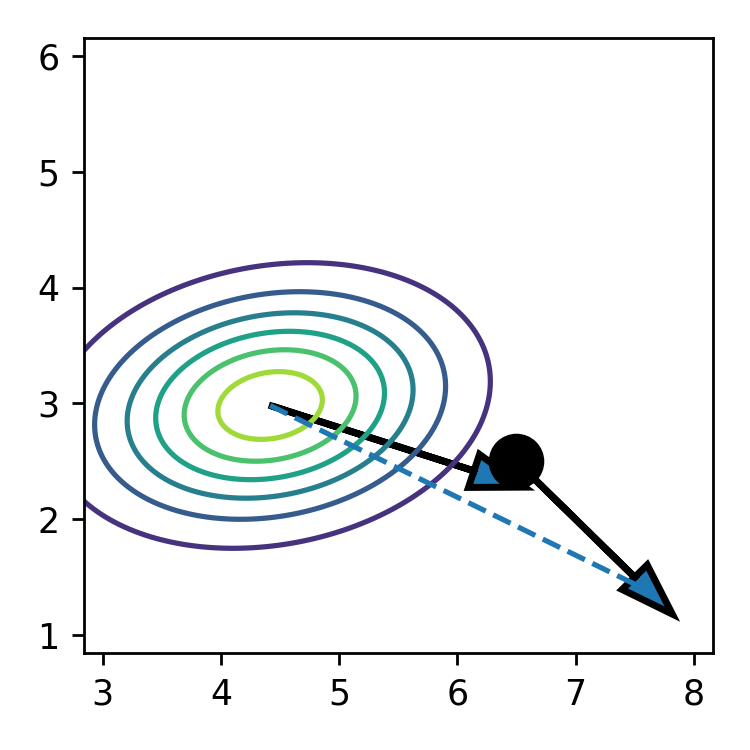}%
    \includegraphics[width=0.35\textwidth]{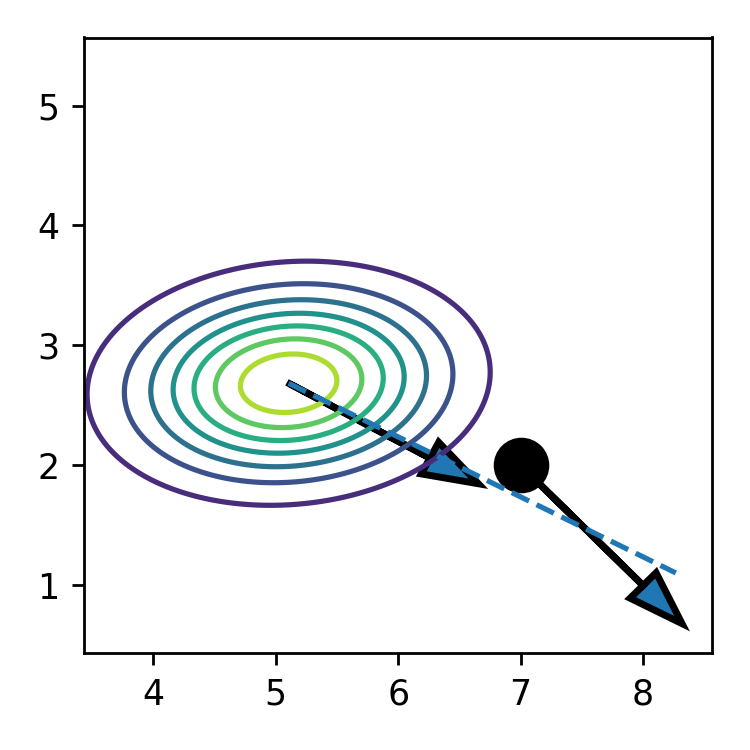}\\
    \includegraphics[width=0.35\textwidth]{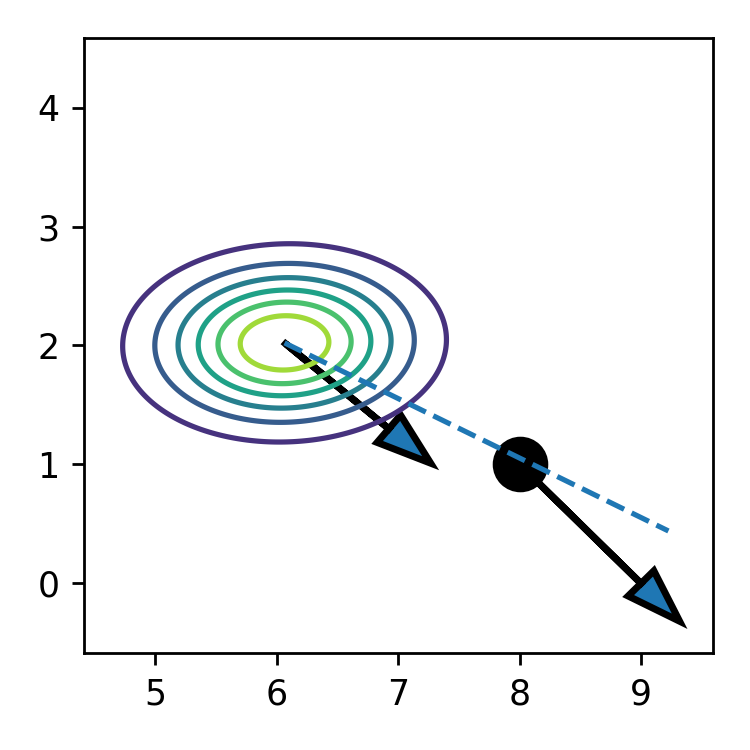}%
    \includegraphics[width=0.35\textwidth]{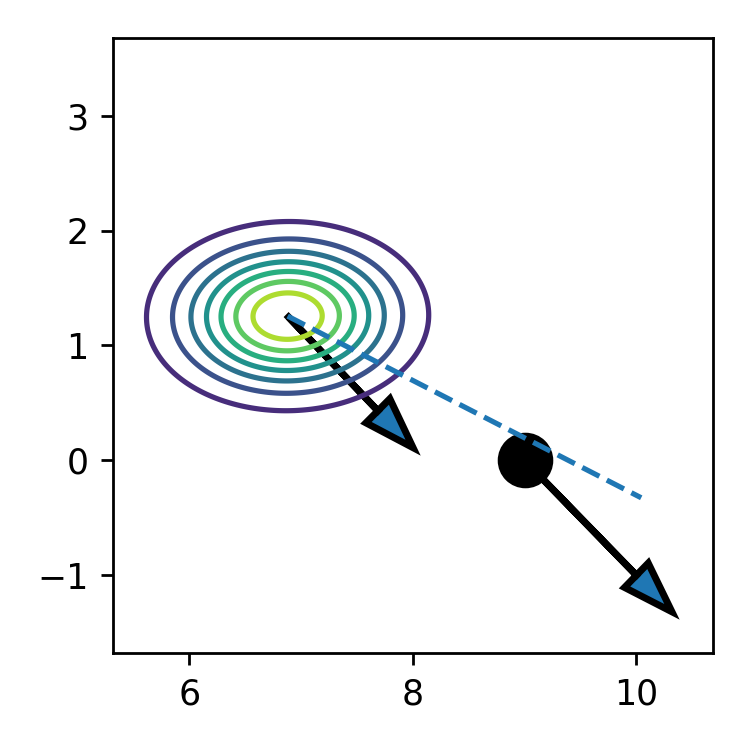}\\
    \includegraphics[width=0.35\textwidth]{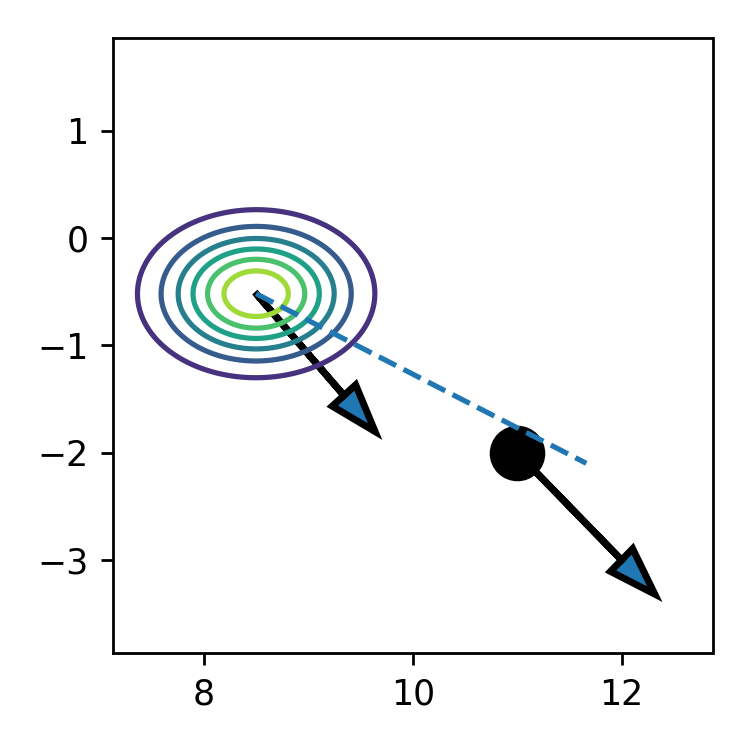}%
    \includegraphics[width=0.35\textwidth]{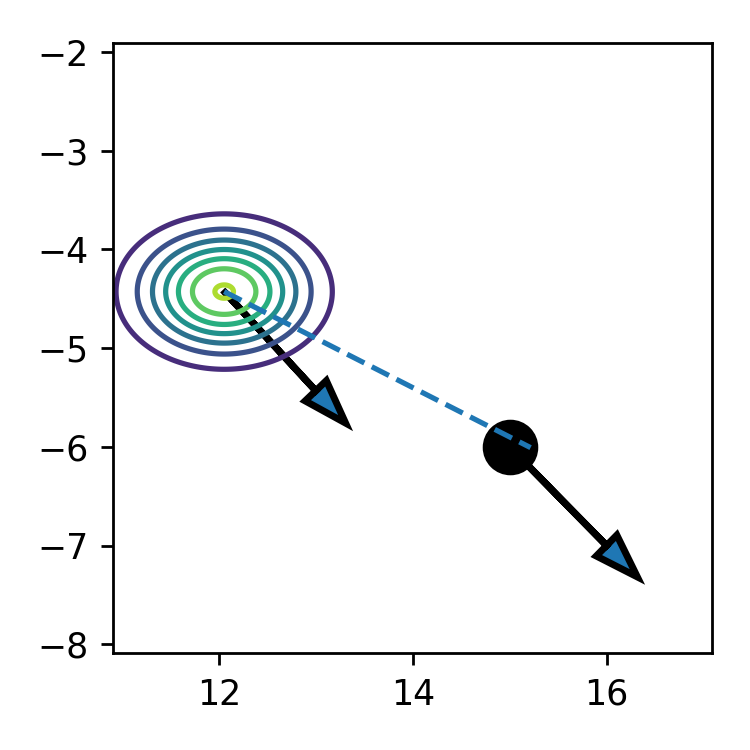}
    \caption{Fixed lag for a 2-dimensional example. Moving optimum $y(t)$ is marked with a solid black circle. Velocities of moving optimum and mean $m(t)$ are each shown with an arrow. Contour lines are ellipses corresponding to covariance matrix of the current distribution. Dashed line shows theoretical fixed lag, which is asymptotically reached after some time has passed.}
    \label{fig:moving_opt_2d}
\end{figure}

\begin{rem}
    The non-degeneracy assumption on $Q$ can likely be dropped by arguments found in various parts of this manuscript, and by modifying the fixed lag to account for directions of degeneracy in $Q$, with the mean ``chasing'' only the component of $y(t)$ observable by $A$ (i.e. $\Gamma$-orthogonal to $\ke A$), but we won't get into the details of maximal generality here.
\end{rem}

Having the optimum move with constant velocity is a strong restriction but it turns out we can prove a very similar statement just under the assumption that the speed of the moving optimum is bounded.
\begin{lem}If $\|A^-\dot y(s)\|\leq R$, then we can explicitly bound the lag between the mean $m(t)$ and the signal $y(t)$.
\begin{proof}
Let $\lambda_{\min}$ be the minimal eigenvalue of $(\Sigma Q)^{1/2}$ and recall the generalised Young inequality $ab\leq \frac{\varepsilon}{2}a^2 + \frac{1}{2\varepsilon}b^2$ for any $\varepsilon > 0$.
\begin{align*}
    m(t) - A^-y(t) &= e^{-(\Sigma Q)^{1/2}t}\left[ \int_0^t e^{(\Sigma Q)^{1/2}s}A^{-} \dot y(s)\d s + \left(m(0) - A^- y(0)\right)\right]\\
    \frac{1}{2}\frac{\d}{\d t}\|m(t) - A^-y(t)\|^2 &= \left\langle m(t) - A^-y(t), -(\Sigma Q)^{1/2} (m(t) - A^-y(t)) + e^{-(\Sigma Q)^{1/2}t}\left(e^{(\Sigma Q)^{1/2}t}A^-\dot y(t) \right) \right\rangle \\
    &= - \|m(t) - A^- y(t)\|_{(\Sigma Q)^{-1/2}}^2 +  \langle m(t) - A^- y(t), A^- \dot y(t)\rangle\\
    &\leq -\lambda_{\min}\|m(t) - A^-y(t)\|^2 + \frac{\lambda_{\min}}{2}\|m(t) - A^-y(t)\|^2 + \frac{1}{2\lambda_{\min}}\|A^- \dot y(t)\|^2\\
    &=-\frac{\lambda_{\min}}{2}\|m(t) - A^-y(t)\|^2 + S,
\end{align*}
where $S = \frac{1}{\lambda_{\min}}\sup_t \|A^- \dot y(t)\|^2$. This means that, setting $D(t) = \frac12\|m(t) - A^-y(t)\|^2,$ we have $\dot D(t) \leq -\lambda_{\min} D(t) + S$, which means that (using a differential inequality on the explicitly solvable 1d ODE for $D$) $D(t) \leq \left(S/\lambda_{\min} - D(0) \right)e^{-\lambda_{\min}t} + \frac{S}{\lambda_{\min}} \simeq \frac{1}{2\lambda_{\min}^2}\sup_t \|A^- \dot y(t)\|^2$. This means that if $\|A^-\dot y(t)\| \leq R,$ then $\|m(t) - A^-y(t)\| \simeq \lambda_{\min}^{-1} R$. Note that this expression is very similar to the fixed lag proven in \ref{lem:fixed_lag}.
\end{proof}
\end{lem}

\subsection{Tracking an oscillating or periodic optimum}
For simplicity we consider a one-dimensional trait space here. In this case, all matrices are scalars, i.e., $A = a = 1$ (without loss of generality), $\Gamma = \gamma^2$, $Q = 1/\gamma^2$, $C_0 = c_0^2$, $\Sigma = \sigma^2$ and $A^- = 1$. Let $y(t) = y_0 + \xi \sin(\theta t)$ be an oscillating optimum, and we set $s=0$.
\begin{align*}
    m(t) - y(t) &= e^{-\sigma/\gamma t} \left( \int_0^t e^{\sigma/\gamma s} \xi \theta \cos(\theta s)\d s + m(0) - y(0)\right)\\
    &= \frac{\frac{\sigma}{\gamma} \theta \xi \cos(\theta t) + \theta^2 \xi \sin(\theta t) + e^{-\sigma/\gamma t}(m(0)-y(0) -\sigma/\gamma)}{\frac{\sigma^2}{\gamma^2} + \theta^2}\\
    &\simeq \frac{\frac{\sigma}{\gamma} \theta \xi \cos(\theta t) + \theta^2 (y(t) - y_0)}{\frac{\sigma^2}{\gamma^2} + \theta^2}
\end{align*}
where the asymptotics $\simeq$ can be understood as the long-term behaviour. The fast oscillation limit ($\theta \to \infty$) is given by $m(t) - y(t) \simeq y(t) - y_0$, i.e. $m(t) \to y_0$.

The asymptotic (pure oscilation-regime) population dynamics is then given by \eqref{eq:moment_ODEs}, i.e.,
\begin{align*}
    \frac{\dot P(t)}{P(t)} &= - \gamma^{-2}(m-y(t))^2 - \frac{\sigma}{\gamma} + K = -\frac{\xi^2}{\gamma^2} \frac{\frac{\sigma^2}{\gamma^2}\theta^2 \cos^2(\theta t) + 2\frac{\sigma}{\gamma}\theta^3 \sin(\theta t)\cos(\theta t) + \theta^4 \sin^2(\theta t)}{\left(\frac{\sigma^2}{\gamma^2} + \theta^2\right)^2}- \frac{\sigma}{\gamma} + K
\end{align*}
This means that $\dot P(t) = P(t) a(t)$ for $a(t)$ being defined as the right hand side of the equation above. Since the solution of this ODE is given by $P(t) = P(0) \exp\left(\int_0^t a(s)\d s\right)$, the effective exponential growth or decay constant is given by (using $\int_0^{2\pi}\sin^2(s)\d s = \pi$ and $\int_0^{2\pi}\sin(s)\cos(s)\d s = 0$)
\begin{align*}
    \int_0^{k\cdot (2\pi)/\theta} a(s)\d s &= -\frac{k}{\theta}\frac{\xi^2}{\gamma^2}\frac{\frac{\sigma^2}{\gamma^2}\theta^2 \cdot \pi + 0 + \theta^4\pi}{\left(\frac{\sigma^2}{\gamma^2} + \theta^2\right)^2} - 2\frac{k}{\theta}\pi \frac{\sigma}{\gamma} + 2\frac{k}{\theta}\pi K \\
    &= -\frac{2k\pi}{\theta} \left(\frac{\xi^2}{2\gamma^2} \frac{\theta^2}{\frac{\sigma^2}{\gamma^2} + \theta^2} + \frac{\sigma}{\gamma} - K\right),
\end{align*}
i.e.,
\begin{align*}
    t^{-1}\int_0^t a(s)\d s &\simeq - \left(\frac{\xi^2}{2\gamma^2} \frac{\theta^2}{\frac{\sigma^2}{\gamma^2} + \theta^2} + \frac{\sigma}{\gamma} - K\right)
\end{align*}
and $P$ exhibits asymptotic exponential growth if and only if $K > \frac{\xi^2}{2\gamma^2} \frac{\theta^2}{\frac{\sigma^2}{\gamma^2} + \theta^2} + \frac{\sigma}{\gamma}$ and exponential decay otherwise.

\begin{figure}
    \centering
    \includegraphics[width=0.9\textwidth]{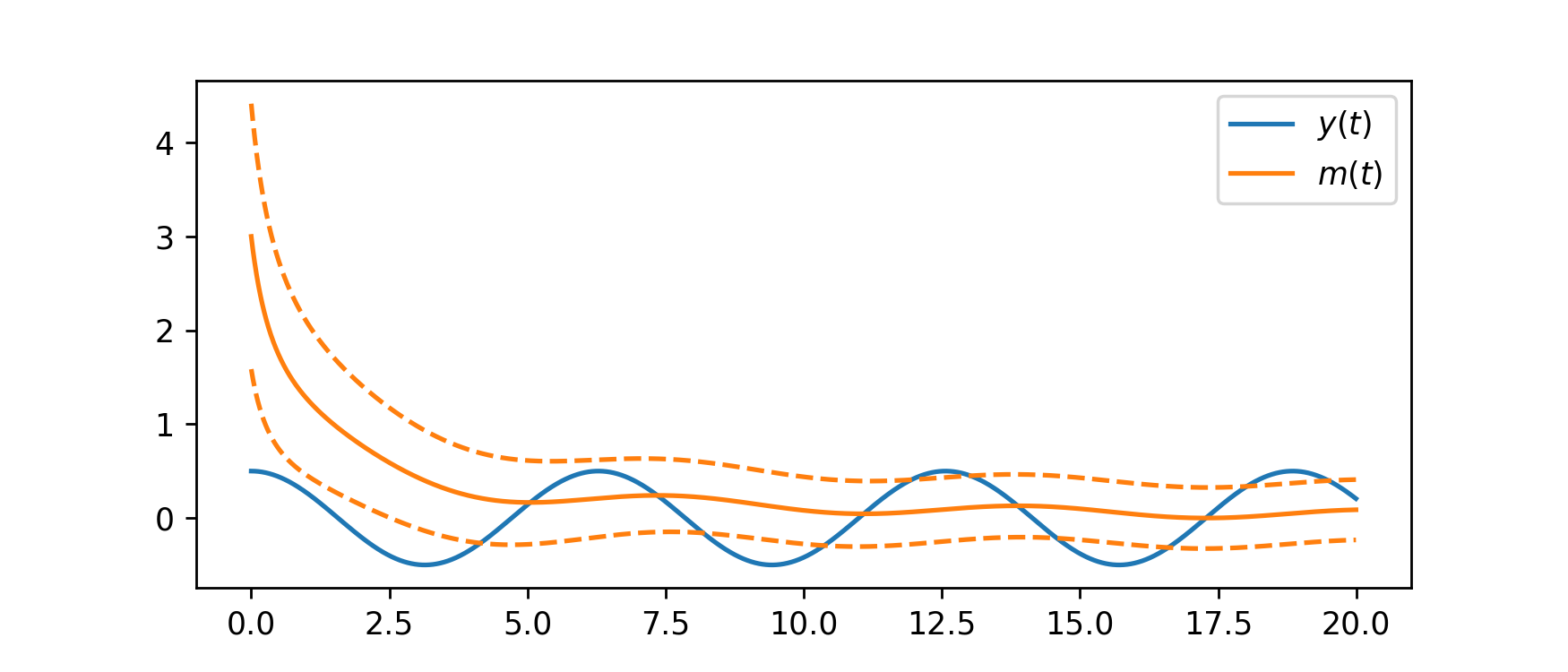}\caption{Oscillating  signal $y(t)$ (blue line) being tracked by the mean of the distribution. The distribution's covariance is sketched using $\pm$ one standard deviation above and below the mean.}
    \label{fig:oscopt}
\end{figure}

\subsection{Randomly fluctuating optimum}
In this section we will briefly investigate whether a population tracking a randomly fluctuating optimum goes extinct. Let $y(t) = a W(t)$ where $W$ is standard Brownian motion and consider the same scalar case as in the oscillating optimum case. 

We consider the deviation of the mean from the true signal by $ D(t) = m(t) - aW(t)$. The corresponding SDE is 
\begin{align*}
    \d  D &= \d m - a\d W = -\frac{\sigma}{\gamma} (m - aW)\d t - a\d W = -\frac{\sigma}{\gamma} D\d t - a\d W,
\end{align*}
which is an Ornstein-Uhlenbeck process with mean $0$ solved by 
\begin{align*}
     D(t) &= -a\int_0^t \exp\left(-\frac{\sigma}{\gamma}(t-s)\right) \d W(s).
\end{align*}
Its square $D^2 = (m(t)-a W(t))^2$ is a Cox-Ingersoll-Ross process of form
\begin{equation}
    \d D^2 = -2\frac{\sigma}{\gamma} D^2 \d t -2a \sqrt{D^2}\d W + a^2\d t = 2\frac{\sigma}{\gamma}\left( \frac{a^2\gamma}{2\sigma} - D^2\right)\d t - 2a\sqrt{D^2} \d W
\end{equation}
The asymptotic distribution of $D^2$ is known to be a Gamma distribution with mean $\frac{\gamma a^2}{2\sigma}$ and variance $\frac{\gamma^2 a^4}{2\sigma^2}$.

We recall that the population size dynamics is given by (from \eqref{eq:moment_ODEs})
   \begin{equation*}\begin{split}
        \dot P(t) &= P(t)\cdot \left(K  -(1-s)\|Am(t)-y(t)\|_\Gamma^2   -\tr [A^\top \Gamma^{-1}AC(t)] \right).
    \end{split}\end{equation*}
In order to obtain a (biased deterministic approximation of the) mean rate of exponential growth (or decay), we need to compute the expectation of the time-averaged squared deviation $\tau^{-1}\int_0^\tau \mathcal \|Am(t) - y(t)\|_\Gamma^2\d t =\tau^{-1}\int_0^\tau \gamma^{-2} \mathcal D(t)^2\d t$. Since $D^2$ is an ergodic process, its time-integrated average converges to the mean of the asymptotic distribution, i.e. 

\[\tau^{-1}\int_0^\tau \mathcal \|Am(t) - y(t)\|_\Gamma^2\d t = \frac{1}{\gamma^2}\tau^{-1}\int_0^\tau D^2(t)\d t \to \frac{1}{\gamma^2} \cdot \frac{\gamma a^2}{2\sigma}\]
almost surely.

Note that the ODE for $P$ has explicit solution
\begin{equation*}
    P(t) = P_0 \exp\left(t\left( K  -(1-s)\|Am(t)-y(t)\|_\Gamma^2   -\tr [A^\top \Gamma^{-1}AC(t)]\right) \right)
\end{equation*}

This means that a biased (because of the nonlinear transformation given by the logarithm) estimator of the asymptotic expected population growth rate (at least the first component of it) is given by
\begin{align*}
    \frac{1}{\tau}\E \log P(t) &= -\gamma^{-2} \frac{1}{\tau}\E  \int_0^t \mathcal D^2(s)\d s - \frac{\sigma}{\gamma} + K \\
    &= -\frac{a^2}{2\sigma \gamma} - \frac{\sigma}{\gamma} + K = K - \frac{\sigma}{\gamma}\left(1 + \frac{a^2}{2\sigma^2}\right).
\end{align*}

Figure \ref{fig:fluctopt} shows both a sample run and a statistics of 1000 independent runs, comparing their population size over time with the theoretical rate $K - \frac{\sigma}{\gamma}\left(1 + \frac{a^2}{2\sigma^2}\right)$. We want to point out that the characterisation of the mean dynamics (pointed out in remark \ref{rem:locallyweighted}) as a locally weighted reconstruction of the ``data'' $y(t)$ means that the path of $m(t)$ is much smoother than the data itself.  

\begin{figure}
    \centering
    \includegraphics[width=\textwidth]{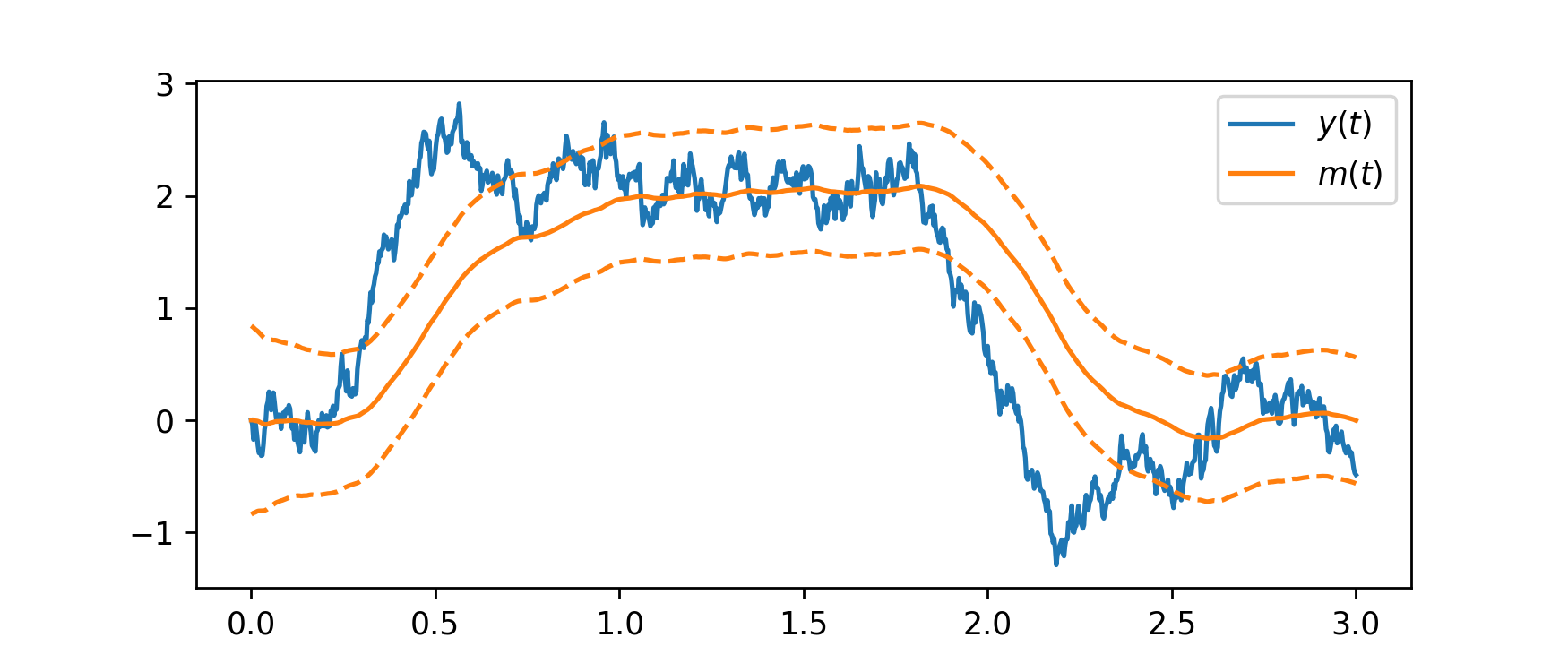}
    \includegraphics[width=\textwidth]{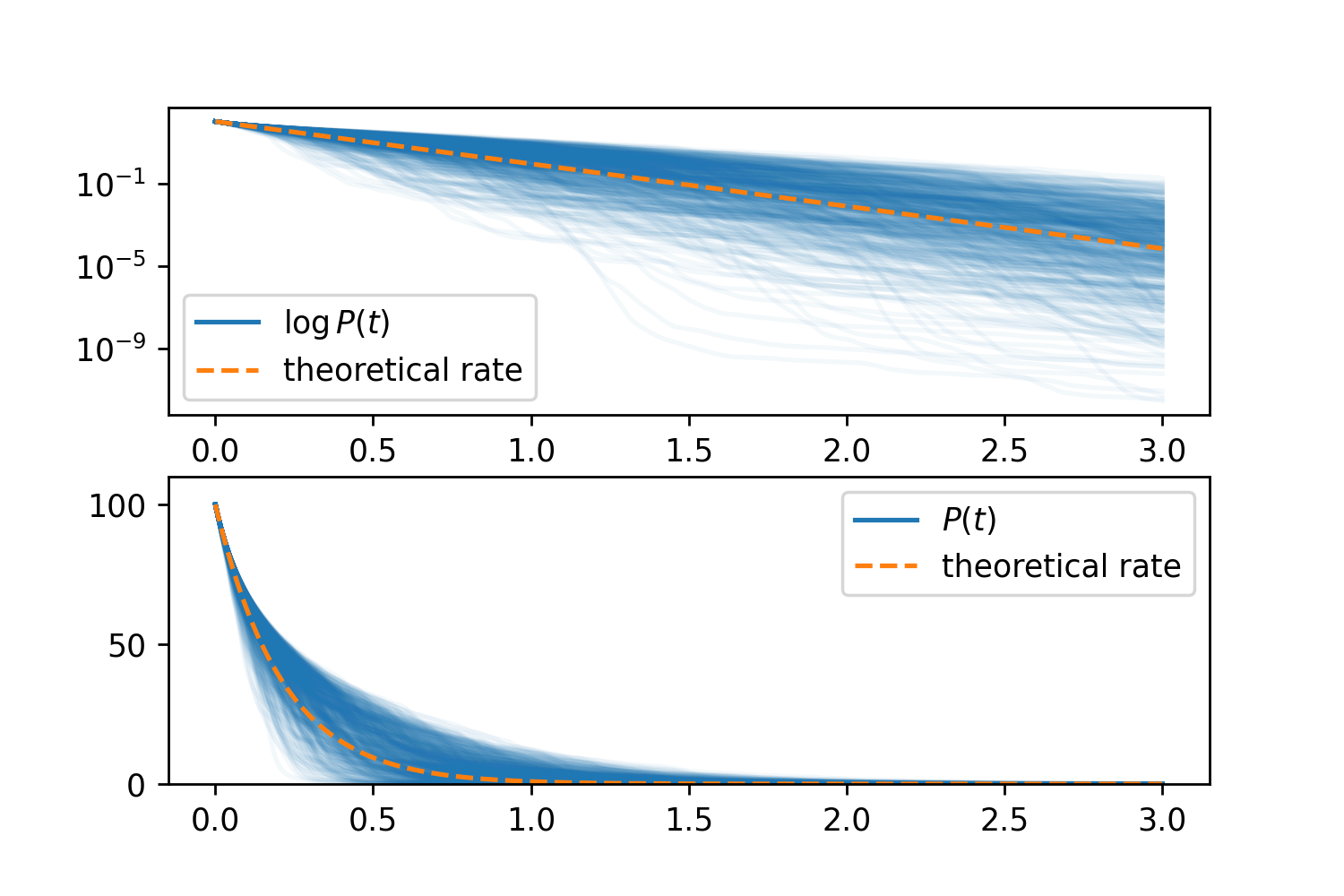}
    \caption{Top: Randomly fluctuating signal $y(t)$ (blue line) being tracked by the mean of the distribution. The distribution's covariance is sketched using $\pm$ one standard deviation above and below the mean. Bottom: Population growth/decay for 1000 independent runs, with theoretical average rate, visualised on exponential and linear scale.}
    \label{fig:fluctopt}
\end{figure}

\section*{Conclusion}
We have derived moment equations for the continuous-time, continuous-trait replicator-mutator equation with a quadratic fitness function depending on traits through a linear mapping. 
 Analytical solutions for covariance dynamics and asymptotic mean and population size of this system have been derived, which to the best of our knowledge, have not been determined for this setting in the literature. With these analytical solutions we demonstrated a range of interesting evolutionary phenomena without needing to resort to numerical simulations. Avenues of future research include generalisation of several mathematical statements to trickier edge cases (like degenerate mutation), the analysis of ``intermittent extinction due to attrition'' in the process of adaptation to an optimum (say, after a sudden change in optimum to a new fixed level), as well as a similar analysis for other popular models of evolution like the quasispecies equation.

\subsection*{Conflicts of interest}
None to disclose.
\begin{appendix}

\newpage 
\section{Notation and mathematical concepts}

We denote by $\R^n$ the standard $n$-dimensional Euclidean space, and $\R^{m\times n}$ as the space of $m\times n$ matrices. We write $A^\top\in \R^{n\times m}$ as the transpose of $A\in\R^{m\times n}$. A square matrix $M\in\R^{n\times n}$ is called symmetric and positive semi-definite if $M^\top = M$ and for any $0\neq x\in \R^n$, we have $x^\top M x \geq 0$. If even $x^\top M x > 0$ for any such $x$, then $M$ is called positive definite. A symmetric and positive semi-definite matrix $M\in \R^{n\times n}$ has only non-negative real eigenvalues and the corresponding eigenvectors are a basis of $\R^n$. In this case we can decompose $M = U \Lambda  U^\top$, where $\Lambda$ is a diagonal matrix containing the non-negative eigenvalues and $U$ is an orthogonal matrix, with the eigenvectors being columns of $U$. A symmetric and positive definite matrix has only positive real eigenvalues.  For a symmetric and positive definite matrix $\Gamma\in \R^{n\times n}$ we define the preconditioned norm as $\|x\|_{\Gamma}^2 = x^\top \Gamma^{-1}x$ and the inner product $\langle x, y\rangle_\Gamma = x^\top \Gamma^{-1}y$. We call a matrix $T\in \R^{m\times n}$ a stochastic matrix if its entries are in the interval $[0,1]$ and if $\sum_i T_{ji} = 1$ for all $j$. A matrix $G\in \R^{m\times n}$ is called a generator matrix if $\sum_i G_{ji} = 0$. If $A\in\R^{m\times n}$ and $B\in\R^{m\times k}$ then $[A|B]\in \R^{m\times (n+k)}$ is understood as the block matrix obtained by concatenation of columns. For a mean vector $m$ and a symmetric and positive definite covariance matrix $C$ we denote a multivariate Gaussian distribution with this mean and covariance by $\mathcal N(m,C)$. For a square matrix $M\in \R^{n\times n}$ we define the matrix exponential as $\exp(M) = \sum_{k=0}^\infty \frac{M^k}{k!}$. Many properties of the scalar exponential functions hold, but it is to be noted that in general, $\exp(M+N) \neq \exp(M)\exp(N)$, unless $M$ and $N$ commute. In order to evaluate the expression on the left hand side, we will need the Lie product formula, $\exp(M+N) = \lim_{n\to \infty} \left(\exp(M/n)\exp(N/n)\right)^n$.  The trace of a square matrix $M\in\R^{n\times n}$ is defined as $\tr(M) = \sum_i M_{ii}$. The Moore-Penrose pseudo-inverse of a matrix $M\in\R^{m\times n}$ is the unique matrix $M^+$ satisfying the following four equations: $MM^+M = M$, $M^+MM^+ = M^+$, $(MM^+)^\top = MM^+$, and $(M^+M)^\top = M^+M$. The matrix square root $M^{1/2}$ of a square matrix $M\in \R^{m\times m}$ is any matrix $L$ satisfying $LL = M$. If $M$ is symmetric and positive definite, then the matrix square root is unique and can be computed from the eigenvalue decomposition: If $M = U\Lambda U^\top$, then $M^{1/2} = U \Lambda^{1/2}U^\top,$ where $\Lambda^{1/2}$ is the straight-forward matrix square root obtained by taking the square root of all entries in the diagonal matrix $\Lambda$.
\section{Useful lemmata}
\begin{lem}[Woodbury matrix identity]
    \[ \left(A+UCV\right)^{-1}=A^{-1}-A^{-1}U\left(C^{-1}+VA^{-1}U\right)^{-1}VA^{-1}\]
\end{lem}

\begin{lem}\label{lem:properties_M}
    Let $C_0$ and $A$ be the matrices in the setting of lemma \ref{lem:cov_ccrep}. Then, defining $M(t) = C(t)C_0^{-1} = (I + rt C_0 A^\top \Gamma^{-1} A)^{-1}$ and $L(t):= \exp(\alpha \log M(t)) = M(t)^\alpha$ (where existence of this object is proven below), we can prove the following:
    \begin{enumerate}
        \item $M(t)$ is diagonalisable.
        \item The matrix logarithm $\log M(t)$ exists, so $L(t)$ is well-defined.
        \item $[M,\dot M] = 0$ and $[\log M(t), d/dt\log M(t)] = 0$
        \item $\dot L(t) = \alpha \dot M M^{\alpha-1}$
        \item For $u\neq 1$, we have $\int M(s)^{u}\d s \cdot C_0A^\top  \Gamma^{-1}A = -(u-1)^{-1}M(t)^{u-1}$. The case $u = 1$ is covered by $\int M(s) \d s \cdot C_0 A^\top \Gamma^{-1} A = - \log M(t)$
    \end{enumerate}
    \begin{proof}
         \cite[Theorem 2.3]{bungert2023complete} shows that there is a time-independent diagonalisation $M(t) = S\cdot E(t) S^{-1}$ (and simultaneous diagonalisation of $C_0A^\top \Gamma^{-1} A = SDS^{-1}$), so that $M$ is indeed diagonalisable and $[M, d/dt M(t)]=0$ follows immediately from its time-dependent diagonalisation. Note that (again by simultaneous diagonalisation)
         \[\dot M(t) = d/dt (C(t)C_0^{-1}) = -C(t)A^\top \Gamma^{-1} AC(t)C_0^{-1}= -M(t) C_0A^\top \Gamma^{-1} AM(t) = -M(t)^2C_0A^\top \Gamma^{-1} A\] Also, setting $\log M(t) = S \log E(t) S^{-1}$ yields its unique real logarithm, which again commutes with its derivative. This means that by elementary properties of the matrix exponential (and due to the commutator vanishing, $[\log M(t), d/dt\log M(t)] = 0$),
         \[\frac{d}{dt}M(t) = \frac{d}{dt} \exp(\log M(t))  = \exp(\log M(t)) \frac{d}{dt}\log M(t) = M(t)  \frac{d}{dt}\log M(t), \]
         i.e. 
        \[ \frac{d}{dt}\log M(t) = M(t)^{-1} \frac{d}{dt}M(t) = \frac{d}{dt}M(t) M(t)^{-1}.\]
        This shows that (again using commutativity of all matrices involved)
        \[\dot L(t) = \exp(\alpha \log M(t)) \cdot \alpha \frac{d}{dt}\log M(t)  = \alpha M(t)^\alpha M(t)^{-1}\dot M(t) = \alpha \dot M M^{\alpha-1} \]
        
        We take the derivative of $M(t)^{u-1}$:
        \begin{align*}
            d/dt M(t)^{u-1} &= (u-1) M(t)^{u-2} \cdot \dot M(t) = (u-1)M(t)^{u-2} \cdot (- M(t)^2 C_0A^\top \Gamma^{-1} A) \\
            &= -(u-1)M(t)^u C_0A^\top \Gamma^{-1} A
        \end{align*}
        which proves the first part of property 5. The case $u=1$ follows from 
        \begin{align*}
            d/dt \log M(t) &= M(t)^{-1}\dot M(t) = -rM(t)^{-1}M(t)^2 C_0A^\top \Gamma^{-1} A = - M(t) C_0A^\top \Gamma^{-1} A
        \end{align*}
    \end{proof}
\end{lem}

\begin{lem}\label{lem:product_of_gaussians}
    \begin{equation}
        \mathcal N(z,\Gamma)(x)\mathcal N(m,C)(z) = \mathcal N(m, \Gamma + C)(x)\mathcal N(\mu, Q)(z),
    \end{equation}
    where 
    \begin{align*}
        Q &= (C^{-1}+\Gamma^{-1})^{-1}\\
        \mu &= Q(C^{-1}m +\Gamma^{-1}x)
    \end{align*}
\end{lem}

\begin{lem}\label{lem:gaussian_moments}These equalities follow from application of Isserlis' theorem.
    \begin{align}
        \int \|Rz-b\|^2 \d \mathcal N(m, \Sigma)(z) &= \|Rm-b\|^2 + \tr (R^\top R\Sigma)\\
        \int  \|Rz-b\|^2\, z \d \mathcal N(m, \Sigma)(z) &= \|Rm-b\|^2\, m + 2\Sigma R^\top (Rm-b) + \tr (R^\top R\Sigma)\, m\\
        \int (z-m)\otimes (z-m) \langle Bz, z\rangle \d \mathcal N(m, \Sigma)(z) &= \tr(B\Sigma)\cdot \Sigma + \Sigma(B+B^\top)\Sigma + \langle Bm, m\rangle \Sigma
    \end{align}
\end{lem}

\section{Proof of lemma \ref{lem:repmuttime}}

\begin{proof}

Ad 1.: We start by deriving a general solution for the Riccati equation $\dot C = \Sigma - CQC$ where $\Sigma \in \R^{n\times n}$ is symmetric and positive definite, but $Q\in \R^{n\times n}$ symmetric and only positive semi-definite (in general).

Our ansatz is to write $C(t) = D(t) E(t)^{-1}$ for matrices $D(t),E(t)\in \R^{n\times n}$ where

\begin{equation}\label{eq:cont_time_riccati}
        \begin{pmatrix}
            \dot D(t)\\\dot E(t)
        \end{pmatrix} = \begin{pmatrix}
            0 & \Sigma\\ Q & 0
        \end{pmatrix}
        \begin{pmatrix}
             D(t) \\ E(t)
        \end{pmatrix}.
    \end{equation}

Then indeed, 
\begin{align*}
    \dot C(t) &= \dot D(t) E(t)^{-1} - D(t) E(t)^{-1}\dot E(t) E(t)^{-1} = \Sigma E(t)E(t)^{-1} - (D(t) E(t)^{-1})Q(D(t) E(t)^{-1})\\
    &= \Sigma - C(t)QC(t).
\end{align*} The initial decomposition $C(0) = D(0)E(0)^{-1}$ can be chosen in an arbitrary suitable way, so we set $E(0) = C(0)^{-1}D(0)$.

Note that \eqref{eq:cont_time_riccati} is a matrix-valued ODE in $\R^{(2n)\times n}$, and can be interpreted as a list of $n$ vector-valued ODEs of form
\begin{equation*}
    \begin{pmatrix}
            \dot D_i(t)\\\dot E_i(t)
        \end{pmatrix} = 
        \begin{pmatrix}
            0 & \Sigma\\ Q & 0
        \end{pmatrix} 
        \begin{pmatrix}
             D_i(t) \\ E_i(t)
        \end{pmatrix},
\end{equation*}
with $D_i(t)$ and $E_i(t)$ being the $i$-th column of $D(t)$ and $E(t)$, respectively. This means that standard solution theory of autonomous linear ODE systems apply, and a closed form solution of \eqref{eq:cont_time_riccati} is given by
\begin{equation}
    \label{eq:cont_time_riccati_solution}    
        \begin{pmatrix}
             D(t) \\ E(t)
        \end{pmatrix} = \exp\left(t\begin{pmatrix}
            0 & \Sigma\\ Q & 0
        \end{pmatrix} \right)
        \begin{pmatrix}
             D(0) \\ E(0)
        \end{pmatrix}.
\end{equation}
This means that next we need to compute the matrix exponential in \eqref{eq:cont_time_riccati_solution}, which we are going to do by analysing the spectrum of $\Sigma Q$.

Since $Q$ is symmetric and positive semi-definite, and $\Sigma$ is symmetric and positive definite, then (see, e.g., \cite[Lemma 2.1]{bungert2023complete}) $\Sigma Q$ is also diagonalisable with non-negative eigenvalues. Let $0 = \lambda_1^2 = \cdots = \lambda_k^2 < \lambda_{k+1}^2 \leq\lambda_{k+2}^2\leq \cdots \lambda_{n}^2$ be the eigenvalues of $ \Sigma Q$ with associated eigenvectors $w_i$, $i=1,\ldots,n$. This includes the possibilities of $k=0$, i.e. all eigenvalues being strictly positive, and $k=n$, where $Q = 0$.

We collect the eigenvectors in a matrix $W = [w_1,w_2,\ldots, w_n]\in\R^{n\times n}$ and define its components as $W = [W_\ke| W_\bot]$ with $W_\ke \in \R^{n\times k}$ and $W_\bot \in \R^{n\times (n-k)}$. 

Similarly, we define $\Lambda = \mathrm{diag}(\lambda_i) \in \R^{n\times n}$ and its invertible component $\Lambda_\bot = \mathrm{diag}(\lambda_i, i \in \{k+1,\ldots, n\})\in \R^{(n-k)\times (n-k)}$.

Finally, we define $V = \Sigma^{-1}W$, in particular $V_\ke = \Sigma^{-1}W_\ke$ and $V_\bot = \Sigma^{-1}W_\bot$. This means that

\begin{align*}
    \Sigma Q W &= \Sigma Q [W_\ke W_\bot] = [\mathbf 0 | W_\bot \Lambda_\bot^2] = W \begin{pmatrix}
        \mathbf 0 & \mathbf 0\\ \mathbf 0 & \Lambda_\bot^2
    \end{pmatrix}
\end{align*}
In the non-degenerate case, this becomes
\begin{equation}
    \Sigma Q W = W\Lambda^2,
\end{equation}
and where $QW\Lambda^{-1} = \Sigma^{-1} W \Lambda$ is another useful expression.

With these definitions and properties,
\begin{align*}
    \begin{pmatrix}
            \mathbf0 & \Sigma\\ Q & \mathbf0
        \end{pmatrix} &\cdot 
        \begin{pmatrix}
            [W_\ke| W_\bot \Lambda_\bot^{-1}] & \mathbf0\\ \mathbf0 &[V_\ke| V_\bot]
        \end{pmatrix} = \begin{pmatrix}
            \mathbf0 & [\Sigma V_\ke | \Sigma V_\bot] \\ [QW_\ke | QW_\bot \Lambda_\bot^{-1}] & \mathbf0
        \end{pmatrix} \\
        &= \begin{pmatrix}
            \mathbf0 & [W_\ke| W_\bot] \\ [\mathbf 0 | \Sigma^{-1} W_\bot \Lambda_\bot^2 \Lambda_\bot^{-1}] & \mathbf0
        \end{pmatrix} = \begin{pmatrix}
            \mathbf0 & [W_\ke| W_\bot] \\ [\mathbf 0 | V_\bot \Lambda_\bot ]& \mathbf0
        \end{pmatrix}\\
        &= \begin{pmatrix}
            [W_\ke| W_\bot \Lambda_\bot^{-1}] & \mathbf0\\ \mathbf0 &[V_\ke| V_\bot]
        \end{pmatrix} \cdot \begin{pmatrix}
            \mathbf 0 & \begin{pmatrix}
                I & \mathbf0 \\ \mathbf0 & \Lambda_\bot
            \end{pmatrix}\\
            \begin{pmatrix}
                \mathbf0 & \mathbf0 \\ \mathbf0 & \Lambda_\bot
            \end{pmatrix} & \mathbf 0
        \end{pmatrix} =: S \cdot \Xi
\end{align*}
In other words, $
    \begin{pmatrix}
            \mathbf0 & \Sigma\\ Q & \mathbf0
        \end{pmatrix} = S \cdot \Xi \cdot S^{-1}
$ where 
\begin{align*}    
 S &=  \begin{pmatrix}
            [W_\ke| W_\bot \Lambda_\bot^{-1}] & \mathbf0\\ \mathbf0 &[V_\ke| V_\bot]
        \end{pmatrix} = \begin{pmatrix}
            W \begin{pmatrix}
                I & \mathbf0 \\ \mathbf0 & \Lambda_\bot^{-1}
            \end{pmatrix} & \mathbf0\\ \mathbf0 & \Sigma^{-1}W
        \end{pmatrix}\\
    S^{-1} &= \begin{pmatrix}
        \begin{pmatrix}
                I & \mathbf0 \\ \mathbf0 & \Lambda_\bot
            \end{pmatrix} W^{-1} & \mathbf 0\\ \mathbf 0 & W^{-1}\Sigma
    \end{pmatrix}      ~  \text{ and } ~
    \Xi = \begin{pmatrix}
            \mathbf 0 & \begin{pmatrix}
                I & \mathbf0 \\ \mathbf0 & \Lambda_\bot
            \end{pmatrix}\\
            \begin{pmatrix}
                \mathbf0 & \mathbf0 \\ \mathbf0 & \Lambda_\bot
            \end{pmatrix} & \mathbf 0
        \end{pmatrix}
        \end{align*}

The diagonalisation derived above means that
\begin{equation}
   \exp\left(t\begin{pmatrix}
            0 & \Sigma\\ Q & 0
        \end{pmatrix} \right) = S \exp\left(t\Xi \right) S^{-1}
\end{equation}
Following the definition of the matrix exponential,
\begin{equation}
    \label{eq:matrix_exp_Xi}
    \exp\left(t\Xi \right) = \begin{pmatrix}
        \begin{matrix}
            \mathbf I & \mathbf 0 \\ \mathbf 0 & \cosh(t\Lambda)
        \end{matrix} &  \begin{matrix}
            t\mathbf I & \mathbf 0 \\ \mathbf 0 & \sinh(t\Lambda)
        \end{matrix} \\
        \begin{matrix}
            \mathbf 0 & \mathbf 0 \\ \mathbf 0 & \sinh(t\Lambda)
        \end{matrix} & \begin{matrix}
            \mathbf I & \mathbf 0 \\ \mathbf 0 & \cosh(t\Lambda)
        \end{matrix}
    \end{pmatrix}.
\end{equation}
This means that (using $\Lambda_\bot^{-1}\cosh(t\Lambda_\bot)\Lambda_\bot = \cosh(t\Lambda_\bot)$)
\begin{align*}
    \exp\left(t\begin{pmatrix}
            0 & \Sigma\\ Q & 0
        \end{pmatrix} \right)  &= \begin{pmatrix}
            W \begin{pmatrix}
                \mathbf I & \mathbf 0 \\ \mathbf 0&\cosh(t\Lambda_\bot)
            \end{pmatrix} W^{-1} & W \begin{pmatrix}
                t\mathbf I &\quad \mathbf 0 \\ \mathbf 0&\Lambda_\bot^{-1}\sinh(t\Lambda_\bot)
            \end{pmatrix} W^{-1}\Sigma  \\[1em]
            \Sigma^{-1}W \begin{pmatrix}
                \mathbf 0 & \mathbf 0 \\ \mathbf 0&\sinh(t\Lambda_\bot)\Lambda_\bot
            \end{pmatrix} W^{-1} &  \Sigma^{-1}W \begin{pmatrix}
                \mathbf I & \mathbf 0 \\ \mathbf 0&\cosh(t\Lambda_\bot)
            \end{pmatrix} W^{-1}\Sigma
        \end{pmatrix}
\end{align*}
All in all,
\begin{align*}
    \begin{pmatrix}
             D(t) \\ E(t)
        \end{pmatrix} &= \exp\left(t\begin{pmatrix}
            0 & \Sigma\\ Q & 0
        \end{pmatrix} \right)
        \begin{pmatrix}
             D(0) \\ E(0)
        \end{pmatrix} = S \exp\left(t\Xi \right) S^{-1}\begin{pmatrix}
             D(0) \\ C(0)^{-1}D(0)
        \end{pmatrix} \\
        &= \begin{pmatrix}
            W \begin{pmatrix}
            \mathbf I & \mathbf 0 \\ \mathbf 0 & \cosh(t\Lambda_\bot)
        \end{pmatrix} W^{-1}D(0) + W\begin{pmatrix}
            t\mathbf I & \mathbf 0 \\ \mathbf 0 & \Lambda_\bot^{-1}\sinh(t\Lambda_\bot)
        \end{pmatrix} W^{-1}\Sigma C(0)^{-1}D(0)\\[1em]
        \Sigma^{-1}W\begin{pmatrix}
            \mathbf 0 & \mathbf 0 \\ \mathbf 0 & \sinh(t\Lambda_\bot)\Lambda_\bot
        \end{pmatrix}W^{-1} D(0) + \Sigma^{-1} W \begin{pmatrix}
            \mathbf I & \mathbf 0 \\ \mathbf 0 & \cosh(t\Lambda_\bot)
        \end{pmatrix} W^{-1}\Sigma C(0)^{-1}D(0)
        \end{pmatrix},
\end{align*}
which proves \eqref{eq:explicit_t_general}. Two particularly interesting special cases: If $Q = 0$, then $k=n$, i.e. $W = W_\ke$ and
\begin{align*}
     \begin{pmatrix}
             D(t) \\ E(t)
        \end{pmatrix} &= \begin{pmatrix}
            W \mathbf I W^{-1} D(0) + t W \mathbf I W^{-1}\Sigma C(0)^{-1}D(0)\\ 0 + \Sigma^{-1}W \mathbf I W^{-1} \Sigma C(0)^{-1} D(0)
        \end{pmatrix}\\
        &= \begin{pmatrix}
            D(0) + t \Sigma C(0)^{-1}D(0)\\
            C(0)^{-1}D(0)
        \end{pmatrix}
\end{align*}
and $C(t) = D(t)E(t)^{-1} = C(0) + t\Sigma$, proving \eqref{eq:explicit_t_Q0}.

On the other hand, if $Q$ is positive definite rather than positive semidefinite, $k= 0$, $W = W_\bot$, and
\begin{align*}
      \begin{pmatrix}
             D(t) \\ E(t)
        \end{pmatrix} &= \begin{pmatrix}
            W \cosh({t\Lambda_\bot}) W^{-1}D(0) + W \Lambda_\bot^{-1}\sinh({t\Lambda_\bot})W^{-1} \Sigma C(0)^{-1} D(0)\\
            \Sigma^{-1}W\sinh({t\Lambda_\bot})\Lambda_\bot W^{-1} D(0) + \Sigma^{-1}W\cosh({t\Lambda_\bot})W^{-1} \Sigma C(0)^{-1} D(0)
        \end{pmatrix},
\end{align*}
which is \eqref{eq:explicit_t_Qnondeg}. The second characterisation \eqref{eq:explicit_t_Qnondeg2} is derived further down below, after we have computed the steady-state solution for $C$.

Ad 2.: The first statement follows immediately from \eqref{eq:explicit_t_Q0}.

We assume now that $Q$ is non-degenerate. Then the steady-state Riccati equation  $C_\infty Q C_\infty = \Sigma$ is solved by $C_\infty := W\Lambda^{-1}W^{-1}\Sigma = W\Lambda W^{-1}Q^{-1}$, where the second equality is a direct result of the defining equality $\Sigma Q W = W\Lambda^2$. Indeed, $C_\infty Q C_\infty = (W\Lambda W^{-1} Q^{-1})Q(W\Lambda^{-1} W^{-1}\Sigma) = \Sigma$. It remains to show that $C_\infty$ is a valid covariance matrix. Parts of the following proof are an adaptation of the arguments found in \cite{laub1979schur}, modified to the setting we are considering. See also \cite{shirilord2022closed} for a related discussion and a closed-form solution for the asymmetric Riccati equation.

$C_\infty$ is symmetric, which is what we are going to prove next. In fact, we define $K := \Lambda W^\top \Sigma^{-1} W$, such that $C_\infty = (W^{-1}\Sigma)^\top K (W^{-1}\Sigma)$. If we can show that $K = K^\top$, then symmetry of $C_\infty$ follows from this relation. In order to prove symmetry of $K$, we define $L := K - K^\top$ and consider
\begin{align*}
    \Lambda L + L\Lambda &= \Lambda^2 W^\top \Sigma^{-1} W - \Lambda W^\top \Sigma^{-1}W\Lambda  +\Lambda W^\top \Sigma^{-1}W\Lambda - W^\top \Sigma^{-1} W \Lambda^2 \\
    &= W^\top Q W - W^\top Q W = 0,
\end{align*}
where we used the spectral property $W\Lambda^2 = \Sigma Q W$. Since $\Lambda L + L\Lambda = 0$ and $\Lambda$ is non-degenerate, Lyapunov theory proves that $0 = L = K - K^\top$, i.e. indeed $C_\infty$ is symmetric. 

In order to prove positive definiteness of $C_\infty$, we first set
\begin{align*}
    \Omega(t) &= \begin{pmatrix}
        W\\ \Sigma^{-1}W\Lambda
    \end{pmatrix} e^{-t\Lambda},
\end{align*}
so that 
\begin{align*}
    \dot \Omega(t) &= \begin{pmatrix}
        W\Lambda\\ \Sigma^{-1}W\Lambda^2 
    \end{pmatrix}e^{-t\Lambda} = \begin{pmatrix}
        \Sigma \cdot (\Sigma^{-1} W \Lambda) \\ QW
    \end{pmatrix}e^{-t\Lambda}\\
    &= \begin{pmatrix}
        0 &\Sigma \\ Q & 0
    \end{pmatrix}e^{-t\Lambda}.
\end{align*}
Furthermore, setting $\omega(t) = e^{-t\Lambda} K e^{-t\Lambda} = e^{-t\Lambda} (\Lambda W^\top \Sigma^{-1} W) e^{-t\Lambda}$,
\begin{align*}
    \dot \omega(t) &=-e^{-t\Lambda} \left(\Lambda^2 W^\top \Sigma^{-1}W + \Lambda W^\top \Sigma^{-1} W \Lambda \right) e^{-t\Lambda}\\
    &=- e^{-t\Lambda} \begin{pmatrix}
        W \\ \Sigma^{-1} W\Lambda
    \end{pmatrix}^\top \cdot \begin{pmatrix}
        Q & 0 \\ 0 & \Sigma 
    \end{pmatrix}\cdot  \begin{pmatrix}
        W \\ \Sigma^{-1} W\Lambda
    \end{pmatrix}e^{-t\Lambda} =- \Omega(t)^\top \begin{pmatrix}
        Q & 0 \\ 0 & \Sigma 
    \end{pmatrix} \Omega(t),
\end{align*}
i.e., $-\dot\omega(t)$ is a symmetric and positive definite matrix for all $t$. This means that 
    \begin{align*}
        \omega(0) - \omega(t) &= -\int_0^t \dot \omega(s)\d s \geq 0,
    \end{align*}
    and, taking the limit $t\to\infty$, 
    \begin{equation}
        K = \omega(0) \geq \lim_{t\to\infty} \omega(t) = 0.
    \end{equation}
    This shows that $K$ is positive definite, so $C_\infty$ is, too.

    We next show that $C_\infty^{-1}$ is the geometric mean of $\Sigma^{-1}$ and $Q$. Indeed,
    \begin{align*}
        G(\Sigma^{-1},Q) &= \Sigma^{-1}(\Sigma Q)^{1/2} = \Sigma^{-1}(W\Lambda^2 W^{-1})^{1/2} = \Sigma^{-1} W\Lambda W^{-1}\\
        &=\Sigma^{-1}W\Lambda W^{-1}  = C_\infty^{-1}
    \end{align*}
    In the same way, $C_\infty$ is the geometric mean of $\Sigma$ and $Q^{-1}$

From $C_\infty = G(Q^{-1},\Sigma) = Q^{-1}(Q\Sigma)^{1/2}$ and the steady-state Riccati equation, we can see that $\Sigma C_\infty^{-1} = QC_\infty = (Q\Sigma)^{1/2} = W^{-\top}\Lambda W^\top$, where the last equality is clear from $W^{-\top} \Lambda^2 W^\top = Q\Sigma$. The other representation is proved similarly.

Now that we have characterised $C_\infty$, we are able to derive \eqref{eq:explicit_t_Qnondeg2}. In this case a different matrix decomposition is more helpful: 
\begin{align*}
    \begin{pmatrix}
        0 & \Sigma \\ Q & 0
    \end{pmatrix} \cdot \begin{pmatrix}
        C_\infty & -C_\infty \\ I & I
    \end{pmatrix} &= \begin{pmatrix}
        C_\infty & -C_\infty \\ I & I
    \end{pmatrix} \cdot \begin{pmatrix}
        QC_\infty & 0 \\ 0 & -QC_\infty
    \end{pmatrix}
\end{align*}
This means that 
\begin{align*}
    \exp\left( t\begin{pmatrix}
        0 & \Sigma \\ Q & 0
    \end{pmatrix} \right) &=  \underbrace{\begin{pmatrix}
        C_\infty & -C_\infty \\ I & I
    \end{pmatrix}}_{=: S} \cdot \begin{pmatrix}
        QC_\infty & 0 \\ 0 & -QC_\infty
    \end{pmatrix} \cdot \underbrace{\frac{1}{2} \begin{pmatrix}
        C_\infty^{-1} & I\\ -C_\infty^{-1} & I
    \end{pmatrix} }_{= S^{-1}}
\end{align*}
Next we think about suitable ways of setting $D(0)$ and $E(0)$. Since we only require $D(0)E(0)^{-1} = C(0),$ we have considerable degrees of freedom. In fact, the following choice, characterised via the matrix $S$ is especially suitable:
\begin{align*}
    \begin{pmatrix}
        D_0\\ E_0
    \end{pmatrix} &= S \begin{pmatrix}\tfrac{1}{2}(I + C_\infty^{-1}C_0) \\[1em] 
        \tfrac{1}{2}(I - C_\infty^{-1}C_0)
    \end{pmatrix} = \begin{pmatrix}
        C_0\\ I
    \end{pmatrix}.
\end{align*}
In this case,
\begin{align*}
    \begin{pmatrix}
        D(t) \\E(t)
    \end{pmatrix} &= \exp\left( t\begin{pmatrix}
        0 & \Sigma \\ Q & 0
    \end{pmatrix} \right) \cdot \begin{pmatrix}
        D_0 \\ E_0
    \end{pmatrix} = S \exp\begin{pmatrix}
        tQC_\infty & 0 \\ 0 & -tQC_\infty
    \end{pmatrix} S^{-1} \cdot \begin{pmatrix}
        D_0 \\ E_0
    \end{pmatrix}\\
    &= S \begin{pmatrix}
        \exp(tQC_\infty) & 0 \\ 0 & \exp(-tQC_\infty)
    \end{pmatrix}\cdot \begin{pmatrix}\tfrac{1}{2}(I + C_\infty^{-1}C_0) \\[1em] 
        \tfrac{1}{2}(I - C_\infty^{-1}C_0)
    \end{pmatrix}\\
    &= \begin{pmatrix}
        C_\infty & -C_\infty \\ I & I
    \end{pmatrix} \begin{pmatrix}
        \exp(tQC_\infty) \cdot \tfrac{1}{2}(I + C_\infty^{-1}C_0) \\
        \exp(-tQC_\infty) \cdot  \tfrac{1}{2}(I - C_\infty^{-1}C_0)
    \end{pmatrix}\\
    &= \begin{pmatrix}
        C_\infty \cdot \left[\exp(t(Q\Sigma)^{1/2}) \cdot \tfrac{1}{2}(I + C_\infty^{-1}C_0)-  \exp(-t(Q\Sigma)^{1/2}) \cdot \tfrac{1}{2}(I - C_\infty^{-1}C_0)\right] \\
        \exp(t(Q\Sigma)^{1/2}) \cdot \tfrac{1}{2}(I + C_\infty^{-1}C_0)+  \exp(-t(Q\Sigma)^{1/2}) \cdot  \tfrac{1}{2}(I - C_\infty^{-1}C_0)
    \end{pmatrix}
\end{align*}
We can see that the continuous-time Riccati equation indeed converges to the steady-state solution $C_\infty$ since 
\begin{align*}
    D(t) E(t)^{-1} &= C_\infty \exp(t(Q\Sigma)^{1/2}) \cdot \left[ \tfrac{1}{2}(I + C_\infty^{-1}C_0) -  \exp(-2t(Q\Sigma)^{1/2}) \cdot\tfrac{1}{2}(I - C_\infty^{-1}C_0)\right] \cdot\\
    &\quad \cdot \left[  \tfrac{1}{2}(I + C_\infty^{-1}C_0)  +  \exp(-2t(Q\Sigma)^{1/2}) \cdot\tfrac{1}{2}(I - C_\infty^{-1}C_0)\right]^{-1}\exp(-t(Q\Sigma)^{1/2}) \\
    &\xrightarrow{t\to \infty} C_\infty\qedhere
\end{align*}

\end{proof}

\end{appendix}

\bibliographystyle{alpha}
\bibliography{mybibfile.bib}

\end{document}